\documentclass[conference,letterpaper,onecolumn]{IEEEtran}
\usepackage[utf8]{inputenc} 
\usepackage[T1]{fontenc}
\usepackage{url}
\usepackage{ifthen}
\usepackage{cite}
\usepackage{braket}
\usepackage[cmex10]{amsmath} 
\usepackage{color,graphicx,amsmath,amssymb,amsthm,epsfig,mathrsfs,cite}
\usepackage{tikz} 
\usepackage{algorithm2e}
\RestyleAlgo{ruled}
\usepackage{physics}
\usepackage{qcircuit}
\usepackage{wrapfig}
\usepackage{mathtools}
\usepackage{pgfplots} 
\usepackage{balance}
\usepackage{amssymb}
\usepackage{amsthm}
\usepackage{float}
\usepackage{caption,subcaption}
\usetikzlibrary{%
  arrows,
  decorations,
  shapes.misc,
  shapes.arrows,%
  shapes.callouts, %
  shapes,%
  shadows,%
  shadows.blur,%
  chains,%
  matrix,%
  positioning,
  patterns,
  scopes,patterns,calc,
decorations.markings,
decorations.pathmorphing
}
\theoremstyle{definition}
\newtheorem{definition}{Definition}
\newtheorem{theorem}{Theorem}
\newtheorem{lemma}{Lemma}

\newtheorem{example}{Example}

\begin{document}
\title{Linear Time Iterative Decoders for Hypergraph-Product and Lifted-Product Codes}

\author{%
  \IEEEauthorblockN{Asit Kumar Pradhan, Nithin Raveendran, Narayanan Rengaswamy, and Bane Vasi\'c}
  \IEEEauthorblockA{\textit{Department of Electrical and Computer Engineering, The University of Arizona, Tucson, AZ, 85721 USA}\\ 
 Email: \{asitpradhan, nithin , narayananr , vasic\}@arizona.edu}
 }
\maketitle

\begin{abstract}
Quantum low-density parity-check (QLDPC) codes with asymptotically non-zero rates are prominent candidates for achieving fault-tolerant quantum computation, primarily due to their syndrome-measurement circuit's low operational depth. Numerous studies advocate for the necessity of fast decoders to fully harness the capabilities of QLDPC codes, thus driving the focus towards designing low-complexity iterative decoders. However, empirical investigations indicate that such iterative decoders are susceptible to having a high error floor while decoding QLDPC codes.
The main objective of this paper is to analyze the decoding failures of the \emph{hypergraph-product} and \emph{lifted-product} codes and to design decoders that mitigate these failures, thus achieving a reduced error floor. The suboptimal performance of these codes can predominantly be ascribed to two structural phenomena: (1) stabilizer-induced trapping sets, which are subgraphs formed by stabilizers, and (2) classical trapping sets, which originate from the classical codes utilized in the construction of hypergraph-product and lifted-product codes.
The dynamics of stabilizer-induced trapping sets is examined and a straightforward modification of iterative decoders is proposed to circumvent these trapping sets.
Moreover, this work proposes a systematic methodology for designing decoders that can circumvent classical trapping sets in both hypergraph product and lifted product codes, from decoders capable of avoiding their trapping set in the parent classical LDPC code.
When decoders that can avoid stabilizer-induced trapping sets are run in parallel with those that can mitigate the effect of classical TS, the logical error rate improves significantly in the error-floor region.

\end{abstract}


\section{Introduction}

Quantum low-density parity-check (QLDPC) codes are emerging as strong contenders for both quantum computing and communications. These codes build upon the success of classical LDPC codes, known for facilitating low-complexity decoding and approaching capacity performance. As discussed by Gottesman~\cite{gottesman_fault_tolerant_ldpc} and Kovalev, Pryadko~\cite{pryadko_fault_tolerant_ldpc}, QLDPC codes enable fault-tolerant error correction with an asymptotically nonzero rate among quantum error correction (QEC) codes. In their work~\cite{LP_codes}, Panteleev and Kalachev introduce a family of QLDPC codes called \emph{lifted-product} (LP) codes, which possess an almost linear minimum distance and constant rate. 
In addition to their exceptional distance characteristics, QLDPC codes~\cite{mackay_quantum} have low-weight stabilizer generators that result in shallow syndrome-extraction circuits, enhancing their appeal for fault-tolerant quantum computations. 
\begin{figure}
\[
    \begin{array}{c}
   \Qcircuit @C=0.6em @R=1.25em {
& \lstick{\ket{\psi}} &\qw& \ctrl{1} &\qw &\qw &\qw  &\gate{S} & \qw & \ustick{T\ket{\psi}}\qw  \\
& \lstick{\ket{T}} &\qw&  \targ &\qw  &\meter &\cw & \cctrlo{-1} \\
}
\end{array}
\]
\caption{The figure shows the implementation of $T$ gate using measurement and Clifford circuit given a $T$ state given by $\frac{1}{2}\left(\ket{0}+\exp{i\frac{\pi}{4}}\ket{1}\right)$.} 
\label{fig:T_state_injection}
\end{figure}
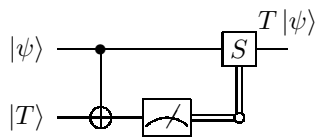

For fault-tolerant computation, in addition to having good codes, designing low-complexity decoders is paramount. 
Consider the state injection circuit depicted in Figure~\ref{fig:T_state_injection}. To implement a $T$ gate on the first qubit, the circuit initially performs a CNOT operation with the first qubit as the control and the second qubit, which is set to $\ket{T}=\frac{1}{2}\left(\ket{0}+\exp{i\frac{\pi}{4}}\ket{1}\right)$, as the target. Subsequently, the second qubit undergoes measurement, and depending on the result, a $s$ gate is applied to the first qubit to achieve complete execution of the $T$ gate.
When the $T$ state corresponds to a logical state, decoding is essential to determine the outcome of the measurement. 
Consequently, if the decoder is not fast enough, it increases the logical clock cycle, thereby diminishing the acceleration provided by the quantum algorithms\cite{Preskill_megaquop, delfosse2023choosedecoderfaulttolerantquantum}. This advantage could be entirely negated if the $T$ depth exceeds a certain number.

Although several families of QLDPC codes exhibit exceptional distance scaling, there remains uncertainty about their feasibility for hardware implementation. This is primarily attributable to the requirement for long-range connections and concerns surrounding the fault-tolerant execution of gates necessary to achieve universal computation with these codes. Recent advances suggest that \emph{hypergraph-product} (HP) and \emph{lifted-product} (LP) codes are particularly compatible with neutral atom platforms \cite{xu2024constant}, while bivariate bicycle codes are suitable for superconducting platforms \cite{Bravyi_2024}. In \cite{patra2024targetedcliffordlogicalgates}, it has been demonstrated that Clifford gates can be enacted within HP codes, which can then be extended to a universal gate set through the injection of distilled magic states into the circuit. These promising developments underscore the necessity of designing low-complexity decoders for QLDPC codes.

Iterative message-passing decoders are considered promising candidates for the decoding of QLDPC codes, primarily due to their low complexity in implementation and their ability to attain low error rates when used for the decoding of classical LDPC codes.
However, message-passing decoders do not perform well when used to decode QLDPC codes, unlike their classical counterparts. This is mainly attributable to two different types of trapping sets (TSs) encountered in QLDPC codes. 
First, since HP and LP codes are constructed by graph products of the Tanner graphs of two classical LDPC codes, the Tanner graph corresponding to the resulting product code inherits the TSs present in the constituent classical LDPC codes. 
In fact, a TS of the constituent classical LDPC appears multiple times as isomorphic copies in the Tanner graph of the product codes.
These TSs are referred to as classical TSs since they are inherited from the constituent classical LDPC codes.
 Second, QLDPC codes can be thought of as two dual-containing classical LDPC codes. This dual-containing property leads to a special type of trapping sets known as \emph{stabilizer-induced} TSs or quantum TSs (QTSs)\cite{QTS}. 
 Several decoders have been proposed to address decoder convergence due to the issues mentioned above~\cite{Poulin_2008,Nithin_stochastic_GallagerB_Quantum,inactivation_decoder,Refined_BP,pruned_neural_BP,Dimitri_two_bit_flipping,Dimitri2025enhancedminsumdecodingquantum,pradhanISTC,iolius2024almostlineartimedecodingalgorithm,yao2024beliefpropagationdecodingquantum}. 
In \cite{BP_OSD}, the authors employ a post-processing technique known as an ordered statistics decoder (OSD) when the message-passing decoder encounters TS. 
Another study \cite{inactivation_decoder} applies the message-passing decoder to parts of the Tanner graph in which TSs are absent, and uses a post-decoding step to address errors in TSs. 
Both methods described in \cite{BP_OSD,inactivation_decoder} require matrix inversion during the post-processing phase, leading to an increase in the decoding complexity. 
According to~\cite{Refined_BP}, a serially scheduled normalized belief propagation decoder can bypass TS if the normalization constant is selected carefully. 
However, the improvement in error suppression comes at the cost of latency due to the serial schedule.
The work in~\cite{Dimitri_two_bit_flipping} introduces a two-bit flipping decoder optimized to bypass TSs in both HP and generalized HP (GHP) codes, specifically those with variable nodes of degree three and check nodes of degree six. Despite its low logical error rate, the applicability of this method does not extend beyond codes with these specific node degrees. 
In~\cite{pradhanISTC}, a modified min-sum decoder is proposed, which adjusts the bias in message-passing rules to avoid TS. 
The decoder in~\cite{Dimitri2025enhancedminsumdecodingquantum} adopts a similar strategy to~\cite{pradhanISTC} prevent TSs in two block QLDPC codes. 
Both techniques depend on parameter optimization for a particular code, indicating that parameters that are fine-tuned for one code may not be compatible with others.
In contrast to the majority of existing research focused on the development of low-complexity decoders from the beginning, we propose a methodology that leverages the design of decoders for HP and LP codes by utilizing the decoders of their underlying classical LDPC codes. This approach is promising due to the existence of numerous strategies for constructing low-complexity iterative decoders that exhibit outstanding performance\cite{henery_neural,FAID_diversity}.


\subsection{Contributions}
This paper focuses on designing parallel, TS-aware, message-passing decoders for HP and LP codes that generalize to codes with arbitrary node degree and do not require a post-processing step. 
To this end, we introduce a methodology to examine the decoding dynamics in TSs induced by either a stabilizer generator or a linear combination of stabilizer generators.
To the best of our knowledge, this is the first work to consider the study of TSs induced by a stabilizer that is a linear combination of two or more stabilizer generators.
Using TS dynamics, we introduce a straightforward modification to the bit-flipping decoder to prevent stabilizer-induced TSs. Further, we characterize how isomorphic copies of a classical TS from constituent codes manifest in the Tanner graph of HP or LP codes. 
This facilitates the design of decoders that circumvent classical trapping sets in HP codes, contingent upon the existence of a suitable decoder for classical LDPC codes that is capable of avoiding such TSs. 
The proposed method thus streamlines the decoder design process, since there exist algorithms in the literature to enumerate TSs of the classical LDPC code\cite{Expansion_contraction} and decoders capable of avoiding them\cite{henery_neural}.
\subsection{Organization of the paper}
Section~\ref{sec:prelim} provides a brief description of the classical LDPC codes and the methodology to construct HP and LP codes based on them, including the necessary notational framework.
In Section~\ref{sec:TS_QLDPC}, we study the decoding dynamics within the subgraphs induced by the stabilizers and show that they constitute trapping sets.
In Section~\ref{sec:TS_aware_BF_decoder}, we propose a decoder that avoids trapping sets induced by stabilizers.
In Section~\ref{sec:QTS_LP_codes}, we show that the stabilizer-induced subgraphs present in the base graph retain their structure in the corresponding Tanner graphs of LP codes, implying that the conclusions regarding stabilizer-induced TSs of the HP codes carry over to the LP codes.
In Section~\ref{sec:classical_TS}, we characterize the trapping sets inherited by the HP and LP codes from their constituent classical codes.
In Section~\ref{sec:min_sum_decoder}, we propose an approach to design decoders based on decoder diversity to avoid both classical and stabilizer-induced trapping sets.
In Section~\ref{sec:simulation_results}, we present numerical results from simulation of decoders proposed in Section~\ref{sec:min_sum_decoder}
\section{Preliminaries}
\label{sec:prelim}
In this section, we set the notation and briefly recall the definitions related to the classical LDPC codes, the depolarizing channel, the stabilizer formalism, and quantum LDPC codes.
\subsection{Notations}
We use bold-face capital letters to denote matrices and bold-face small letters to denote vector variables. 
We use $\mathbf{M}(i,j)$ to denote the $(i,j)-th$ entry of matrix $\mathbf{M}$.
We use $\mathbf{B}_{ij}$to denote the $(i,j)$-th block of block matrix $\mathbf{B}$.
We use $[n]$ to denote the natural numbers from $1$ to $n.$
We use $\mathbf{I}_n$ to denote the identity matrix of dimension $n \times n$.
We denote cardinality of set $\mathcal{A}$ by $|\mathcal{A}|$. We will assume that vectors without transposes are row vectors unless otherwise stated. We represent the absolute value of a scalar variable by $|\cdot|.$
The symbol $\oplus$ denotes modulo two addition, while symbol $\otimes$ denotes the tensor product.

\subsection{Classical LDPC Codes}
\label{sec:classical_LDPC}
In this article, our focus is on binary LDPC codes. These codes can be represented using a bipartite graph, commonly referred to as a Tanner or Factor graph, denoted by $\mathscr{G}$. This graph includes $m$ right nodes (referred to as check nodes), $n$ left nodes (referred to as variable nodes), and $\mathcal{O}(n)$ edges. The collection of variable nodes, check nodes, and the edges that connect variable nodes and check nodes are denoted by $\mathcal{V}, \mathcal{C},$ and $\mathcal{E}$, respectively. 
The variable and check nodes connected by edge $e_i \in \mathcal{E}$ are designated by $v(e_i)$ and $c(e_i)$, respectively. 
The set of variable nodes (or check nodes) that are connected to a particular check node $c \in \mathcal{C}$ (or variable node $v \in \mathcal{V}$) is denoted by $\mathcal{N}_{c}$ ($\mathcal{N}_v$), and given by $$\mathcal{N}_{c}=\{v \in \mathcal{V}: \exists e \in \mathcal{E},  c(e)=c, \text{and } v(e)=v\}, \text{ and }$$$$ \mathcal{N}_{v}=\{c \in \mathcal{C}: \exists e \in \mathcal{E},  v(e)=v, \text{and } c(e)=c\}.$$ 
The bi-adjacency matrix $\mathbf{H} \in \mathbb{F}_2^{m \times n}$ of the Tanner graph $\mathscr{G}$ serves as the parity-check matrix for the corresponding code. Based on the parity-check matrix $\mathbf{H}$, the associated codebook is defined by \begin{equation}
    \label{eq:codebook}
   \mathcal{U} = \{\mathbf{u} \in \mathbb{F}_2^n:\mathbf{H}\mathbf{u}^{\mathsf{T}}=\mathbf{0}\},
\end{equation} where the operation of matrix-vector multiplication $\mathbf{H}\mathbf{u}^{\mathsf{T}}$ is computed in the binary field. 
The number of non-zero elements in a codeword $\mathbf{u}$ is denoted by $\text{wt}(\mathbf{u}),$ a metric that is often referred to as the Hamming weight of the codeword. For a linear code, the minimum distance, symbolized by $d_{\mathrm{min}},$ is defined by \begin{equation*}
   d_{\mathrm{min}} = \min_{\mathbf{u} \in \mathcal{U}} \text{wt}(\mathbf{u}). 
\end{equation*}
Analyzing a single code's performance is challenging, so it is standard to analyze an ensemble of codes, and then argue that the performance of a randomly chosen code from the ensemble converges to the ensemble average. 
Common ensembles include standard, multi-edge, and protograph ensembles\cite{mct}. We will briefly describe the protograph ensemble as it is used in the construction of \emph{lifted-product} codes, the main focus of this article.
\subsubsection{Protograph LDPC codes}
\label{sec:protograph-ldpc-code}
A protograph encapsulates the local structures of a collection of Tanner graphs, the collection of which is referred to as the code ensemble defined by the protograph.
A protograph $\mathscr{\underline{G}}=(\mathcal{\underline{\mathcal{V}}} \cup \underline{\mathcal{C}},\underline{\mathcal{E}})$ is a bipartite graph, where $\underline{\mathcal{V}}$ (or $\underline{\mathcal{C}}$) is the set of different types of variable nodes (or, respectively, check nodes), and $\mathcal{E}'$ is the set of different types of edges in the codes belonging to the ensemble defined by $\underline{\mathscr{G}}$.
 The nodes and edges in the protograph are ordered, and the $i$-th variable node, the check node, and the edge type are denoted, respectively, by $\underline{v}_i$, $\underline{c}_i$, and $\underline{e}_i$. 
 
A protograph is represented by a base matrix $\mathbf{B}$ of dimension $m' \times n' $, whose $(i,j)$-th element $\mathbf{B}(i,j)$ is the number of edges between $\underline{c}_i$ and $\underline{v}_j$. 
A $m \times n$ parity-check matrix from the ensemble defined by base matrix $\mathbf{B}$ is obtained by replacing $\mathbf{B}(i,j)$ with a $l \times l $ matrix $\mathbf{M}_{i,j}$ such that the columns and rows of $\mathbf{M}_{i,j}$ sum to $\mathbf{B}(i,j),$ where $\gamma=\frac{m}{m'}$
In addition to the above constraints, if the matrices $\mathbf{M}_{i,j},$ for $i \in [m']$ and $j \in [n'],$ are circulants, then the resulting parity-check matrix $\mathbf{H}$ will have quasi-cyclic structure, and hence, the corresponding code is called \emph{quasi-cyclic} LDPC code.
The set of binary circulant matrices $l \times l$ is isomorphic to the quotient polynomial ring $\mathcal{R}_\gamma = \mathbb{F}_2[x]/(x^\gamma-1).$
So $\mathbf{M}_{i,j}$ can be represented by the corresponding element from $\mathcal{R}_\gamma$.
With these definitions, the corresponding parity-check matrix $\mathbf{H}$ can be compactly represented by $\mathbf{W} \in \mathcal{R}_\gamma^{ m' \times n'},$ where the matrix representation of $\mathbf{W}(i,j) $ is $\mathbf{M}_{i,j}.$
The above process of obtaining a parity-check matrix from the base matrix is called lifting.
The lifting process can also be described using a copy and permute operation.
To obtain a Tanner graph $\mathscr{G}$ from a protograph $\underline{\mathscr{G}},$ the protograph is first copied, say $\gamma$, times. 
Let $(\underline{v},t), (\underline{c},t)$, and $(\underline{e},t)$ denote the $t$-th copy of variable node $\underline{v} \in \underline{\mathcal{V}}$, check node $\underline{c} \in \underline{\mathcal{C}}$ and edge $\underline{e} \in \underline{\mathcal{E}}$, respectively. 
Then, a permutation map $\pi_{\underline{e}}$ which maps $[\gamma]$ to $[\gamma]$ is assigned to each edge type $\underline{e} \in \underline{\mathcal{E}}$.
If an edge $\underline{e}\in \underline{\mathcal{E}}$ connects variable node $\underline{v}$ to check node $\underline{c}$ in the protograph, then after applying the permutation, the lifted edge $(\underline{e},\pi_{\underline{e}}(t))$ connects variable node $(\underline{v},t)$ to check node $(\underline{c},\pi_{\underline{e}}(t)).$
In the special case of \emph{quasi-cyclic} LDPC codes, if $\pi_{\underline{e}}(t)=a,$ then $\pi_{\underline{e}}(t+1 \quad (\text{mod }   \gamma) )=a+1 \quad  (\text{mod } \gamma).$
\begin{example}
For example, consider a base matrix 

\begin{equation}
\begin{split}
\mathbf{B}=\begin{bmatrix}
1&1&1\\
1&1&1
\end{bmatrix}
\end{split}
\quad
\begin{split}
    \mathbf{H} = 
    \begin{array}{|c|c|c|}
    \begin{matrix}
        1 & 0\\
        0 & 1\\
    \end{matrix} &
    \begin{matrix}
        0 & 1\\
        1 & 0\\
    \end{matrix} &
    \begin{matrix}
        1 & 0\\
        0 & 1\\
    \end{matrix}\\
    \hline
     \begin{matrix}
        0 & 1\\
        1 & 0\\
    \end{matrix} &
    \begin{matrix}
        1 & 0\\
        0 & 1\\
    \end{matrix} &
    \begin{matrix}
        0 & 1\\
        1 & 0\\
    \end{matrix}\\
    \end{array}
\end{split}
\label{eq:proto}
\end{equation}
The protograph corresponding to $\mathbf{B}$ is shown in Fig.~\ref{fig:protex}.
The Tanner graphs lifted from $\mathbf{B}$ will have three groups of variable nodes, two groups of check nodes, and six groups of edges, respectively, corresponding to the three columns, two rows and six non-zero entries in $\mathbf{B}.$ 
Parity-check matrix $\mathbf{H}$ in \eqref{eq:proto} is obtained by choosing
\begin{align*}
  \mathbf{M}_{11}=  \begin{bmatrix}
        1 & 0\\
        0 & 1
    \end{bmatrix},
    \quad
     \mathbf{M}_{12}=  \begin{bmatrix}
        0 & 1\\
        1 & 0
    \end{bmatrix},
    \quad
    \mathbf{M}_{13} =  \begin{bmatrix}
        1 & 0\\
        0 & 1
    \end{bmatrix},\\
      \mathbf{M}_{21}=  \begin{bmatrix}
        0& 1\\
        1 & 0
    \end{bmatrix},
    \quad
     \mathbf{M}_{22}=  \begin{bmatrix}
        1 & 0\\
        0 & 1
    \end{bmatrix},
    \quad
    \mathbf{M}_{23} =  \begin{bmatrix}
        0 & 1\\
        1 & 0
    \end{bmatrix}.
\end{align*}
The Tanner graph corresponding to  $\mathbf{H}$ is shown in Fig.~\ref{fig:Tanner_graph}. 
Since the parity-check matrix in \eqref{eq:proto} has a quasi-cyclic structure, it can be compactly written as 
\begin{equation*}
    \begin{bmatrix}
        x^0 & x^1 & x^0\\
        x^1 & x^0 & x^1
    \end{bmatrix}.
\end{equation*}
In terms of the copy and permutation, the protograph corresponding to base matrix $\mathbf{B}$ is copied two times. The permutation map assigned to different edge types are given by
\begin{equation*}
\begin{array}{rr}
\pi_{\underline{e}_1}(1)=1, \pi_{\underline{e}_1}(2)=2 &  \pi_{\underline{e}_2}(1)=2,  \pi_{\underline{e}_2}(2)=1\\
    \pi_{\underline{e}_3}(1)=1,  \pi_{\underline{e}_3}(2)=2 & \pi_{\underline{e}_4}(1)=2,  \pi_{\underline{e}_4}(2)=1\\
        \pi_{\underline{e}_5}(1)=1,  \pi_{\underline{e}_5}(2)=2 & \pi_{\underline{e}_6}(1)=2,  \pi_{\underline{e}_6}(2)=1.\\
\end{array}
\end{equation*}
\begin{figure*}
\begin{subfigure}[b]{0.4\textwidth}
    \centering
    \input{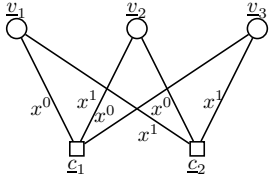}
    \caption{Protograph for base matrix in \eqref{eq:proto}.}
    \label{fig:protex}
\end{subfigure}
\begin{subfigure}[b]{0.6\textwidth}
    \centering
    \input{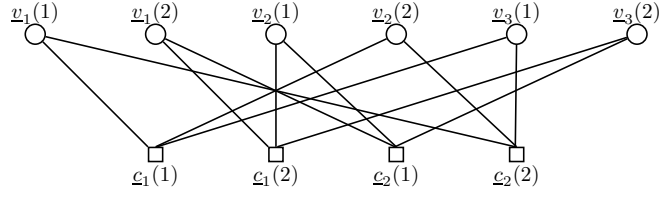}
    \caption{Tanner graph corresponding to parity-check matrix $\mathbf{H}$ in \eqref{eq:proto}.}
    \label{fig:Tanner_graph}
\end{subfigure}
\caption{The figure illustrates lifting of the protograph to a code.}
\end{figure*}
\label{ex:lifting_illustration}
\end{example}

\subsection{Stabilizer Formalism and CSS Codes}
\label{stabilizer_formalism}
Let us denote the $n$-qubit Pauli group by $\mathcal{P}_n=i^l \{I,X,Y,Z\}^{\otimes n}$, $0 \le l \le 3$, where $\otimes n$ is the $n$-fold tensor product, $X$, $Y$, and $Z$ are Pauli matrices, $I$ is the $2 \times 2$ identity matrix, and $i^l$ is the phase factor. 
Let $\mathcal{S}= \langle S_1, S_2, \cdots, S_m \rangle$, $-{I} \notin \mathcal{S}$, be an Abelian subgroup of $\mathcal{P}_n$ with generators ${S}_i$, $1 \leq i \leq m$. 
A $(n,k)$ quantum stabilizer code \cite{Gottesman96} is  a $2^k$-dimensional subspace $\mathcal{C}$ of the Hilbert space $(\mathbb{C}^2)^{\otimes n}$ given by the common $+1$ eigenspace of stabilizer group $\mathcal{S}$:
\begin{equation*}
\label{eq:StabilizerCode}
\mathcal{C} = \{\ket{\psi}, ~~ \text{s.t. } S_i\ket{\psi}= \ket{\psi}, \forall i \}.
\end{equation*}
Every element of the stabilizer group $\mathcal{S}$ is mapped to a binary tuple as follows: $I \rightarrow (0,0),\; X \rightarrow (1,0),\;Z \rightarrow (0,1),\;Y \rightarrow (1,1)$. 
This mapping gives a matrix representation of the stabilizer generators called the parity-check matrix, denoted by $\mathbf{H}$, which is given by
$\mathbf{H}=\begin{bmatrix}
\mathbf{H}_{\mathrm{X}}~|~\mathbf{H}_{\mathrm{Z}}
\end{bmatrix},$
where $\mathbf{H}_{\mathrm{X}}$ and $\mathbf{H}_{\mathrm{Z}}$ represent binary matrices for bit flip and phase flip operators, respectively. 
Note that $\mathbf{H}$ is a $m \times 2n$ matrix.
Similar to the Pauli representation, the stabilizers also commute with respect to \emph{symplectic inner product} in the binary representation \cite{Nielsen}. 
If the stabilizers $S_i$ for all $i \in [m]$ are the $n$-fold tensor product of the operators of the set $\{I, X\}$ or of the set $\{I,Z\},$ then the corresponding code is referred to as the \emph{Calderbank-Shor-Steane} (CSS) code.
Therefore, the parity-check matrices of CSS codes should have the following form 
\begin{equation}
    \label{eq:CSS}
    \mathbf{H} = \begin{bmatrix}
        \mathbf{H}_{\mathrm{X}} & \mathbf{0}\\
        \mathbf{0} & \mathbf{H}_{\mathrm{Z}}
    \end{bmatrix},
\end{equation}
such that $\mathbf{H}_{\mathrm{X}}\mathbf{H}_{\mathrm{Z}}^{\mathsf{T}}=\mathbf{0}.$
 We denote the codebook corresponding to $\mathbf{H}_{\mathrm{X}}$  $(\text{ or }\mathbf{H}_{\mathrm{Z}})$ and its dual code by  $\mathcal{U}_{\mathrm{X}}$ and $\mathcal{U}_{\mathrm{X}^{\perp}}$ (or $\mathcal{U}_{\mathrm{Z}}$ and $\mathcal{U}_{\mathrm{Z}}^{\perp}$), respectively. 
From the CSS condition given in \eqref{eq:CSS}, it can be deduced that $\mathcal{U}_{\mathrm{X}}^{\perp} \subseteq \mathcal{U}_{\mathrm{Z}}$ and $\mathcal{U}_{\mathrm{Z}}^{\perp} \subseteq \mathcal{U}_{\mathrm{X}},$ which means that two dual-containing linear codes define a CSS quantum code.
Elements of the set of operators that map the codespace to itself are called logical operators.
The binary representation of the set of $X$ (or $Z$)-logical operators is given by codebook $\mathcal{U}_{\mathrm{x}}$ corresponding to $\mathbf{H}_{\mathrm{X}}$ $(\text{ or }\mathbf{H}_{\mathrm{Z}}).$
Since $\mathcal{U}_{\mathrm{X}}^{\perp} \subseteq \mathcal{U}_{\mathrm{Z}}$ and $\mathcal{U}_{\mathrm{Z}}^{\perp} \subseteq \mathcal{U}_{\mathrm{X}}$ are sets of stabilizers, they would act on the state of the quantum codeword trivially and are called \emph{degenerate} logicals. 
The minimum distance of a CSS quantum code is the minimum-weight logicals that act non trivially on quantum codeword state, and given by
\begin{equation*}
 d_{\min} = \min\left( d_{\mathrm{X}}, d_{\mathrm{Z}}\right),
    \label{eq:min_distance_CSS}
\end{equation*}
where 
\begin{equation*}
d_{\mathrm{X}} = \min_{\mathbf{u} \in \mathcal{U}_{\mathrm{X}}\setminus \mathcal{U}_{\mathrm{Z}}^{\perp} }\text{wt}(\mathbf{u}), 
\text{ and }
d_{\mathrm{Z}} = \min_{\mathbf{u} \in \mathcal{U}_{\mathrm{Z}}\setminus \mathcal{U}_{\mathrm{X}}^{\perp} }\text{wt}(\mathbf{u}). 
    \label{eq:min_distance_cx_cz}
\end{equation*}
\subsection{Depolarizing Channel}
\label{sec:Depolarizingchannel}
In this work, we consider the depolarizing channel (memoryless Pauli channel), characterized by the depolarizing probability $p$ in which the error $E$ on each qubit is a Pauli operator, and the error on a qubit is independent of the error on other qubits. 
The set of Pauli operators is given by $\mathcal{P} = \{I,X,Y,Z\}$.
In particular, $\Pr(E=X)=\Pr(E=Y)=\Pr(E=Z)=p/3,\Pr(E=I)=1-p$. Similarly to the stabilizers, a Pauli error vector on the $n$ qubits can be expressed as a binary error vector of length $2n$ by mapping the Pauli operators to binary tuples. 
Let $\mathbf{e}=\begin{bmatrix}
    \mathbf{e}_{\mathrm{X}} & \mathbf{e}_{\mathrm{Z}}
\end{bmatrix}$ be the binary representation of Pauli error acting on the $n$ qubits.
The corresponding syndrome, denoted by $\boldsymbol{\sigma}$, is given by 

\begin{equation*}
\boldsymbol{\sigma} = \begin{bmatrix}
    \boldsymbol{\sigma}_{\mathrm{X}} & \boldsymbol{\sigma}_{\mathrm{Z}} 
\end{bmatrix}
= 
\left(\begin{bmatrix}
    \mathbf{H}_{\mathrm{Z}}\mathbf{e}_{\mathrm{X}}^{\mathsf{T}} & \mathbf{H}_{\mathrm{X}}\mathbf{e}_{\mathrm{Z}}^{\mathsf{T}}
\end{bmatrix}\right)^{\mathsf{T}} (\text{mod }  2)
\end{equation*}
\subsection{Quantum LDPC Codes}
This section briefly describes the constructions of both HP and LP codes. 
\subsubsection{Hypergraph product codes}
\label{sec:hypergraph_product}
Given two classical LDPC codes $\mathbf{H}_1$ and $\mathbf{H}_2$, respectively, of sizes $m_1 \times n_1$ and $m_2 \times n_2$, the hypergraph product gives two binary parity-check matrices $\mathbf{H}_{\mathrm{X}}$ and $\mathbf{H}_{\mathrm{Z}}$, which are the binary representation
$X$-type and $Z$-type stabilizer generators such that they satisfy the CSS condition defined in \eqref{eq:CSS}. Algebraically, $\mathbf{H}_{\mathrm{X}}$ and $\mathbf{H}_{\mathrm{Z}}$ are given by
\begin{align}
    \nonumber \mathbf{H}_{\mathrm{X}} & = \begin{bmatrix}
        \mathbf{H}_1 \otimes \mathbf{I}_{n_2} & \mathbf{I}_{m_{1}} \otimes \mathbf{H}_2^\mathsf{T}
    \end{bmatrix},\\
         \mathbf{H}_{Z} & = \begin{bmatrix}
       \mathbf{I}_{n_{1}} \otimes \mathbf{H}_2 &   \mathbf{H}_1^{\mathsf{T}} \otimes  \mathbf{I}_{m_{2}}   
    \end{bmatrix}.
     \label{eq:hyper_graph_product}
\end{align}
Since the iterative decoders are run on the Tanner graphs of $\mathbf{H}_{\mathrm{X}}$ and $\mathbf{H}_{\mathrm{Z}},$ the graphical description of the HP construction aids the understanding of decoding failures. 
Next, we describe the hypergraph product construction as the graph product of two Tanner graphs. Let $\mathscr{G}_1(\mathcal{V}_1 \cup \mathcal{C}_1, \mathcal{E}_1)$ $\left(\text{ or }\mathscr{G}_2(\mathcal{V}_2 \cup \mathcal{C}_2, \mathcal{E}_2)\right)$ be the Tanner graph corresponding to $\mathbf{H}_1$ (or $\mathbf{H}_2$). 
Then, the Tanner graphs corresponding to the $X$ and $Z$
stabilizers are given by  $\mathscr{G}_{\mathrm{Z}}(\mathcal{Q} \cup \mathcal{C}_{\mathrm{X}})$ and $\mathscr{G}_{\mathrm{X}}(\mathcal{Q} \cup \mathcal{C}_\mathrm{Z}),$ respectively, where $\mathcal{Q}= \left(\mathcal{V}_2 \times \mathcal{V}_1 \cup  \mathcal{C}_2 \times \mathcal{C}_1\right), \mathcal{C}_{\mathrm{X}} = \mathcal{C}_2 \times \mathcal{V}_1,$ and  $\mathcal{C}_{\mathrm{Z}} = \mathcal{V}_2 \times \mathcal{C}_1.$
To distinguish variable nodes (or check nodes) from $\mathcal{V}_1$ $(\text{ or }\mathcal{C}_1)$ and $\mathcal{V}_2$ $(\text{ or }\mathcal{C}_2),$ we denote variable nodes from $\mathcal{V}_1$ $(\text{ or }\mathcal{C}_1)$ by $v^1_i$ (or $c^1_j$) and $\mathcal{V}_2$ $(\text{ or }\mathcal{C}_2)$ by $v^2_i$ (or $c^2_j$).
We refer to the nodes in sets $\mathcal{V}_2 \times \mathcal{V}_1, \mathcal{C}_2 \times \mathcal{C}_1, \mathcal{C}_2 \times \mathcal{V}_1, $ and $\mathcal{V}_2 \times \mathcal{C}_1$  as VV-type, CC-type, $X$-type, and $Z$-type, respectively.

For ease of exposition, instead of defining the edges of $\mathscr{G}_{\mathrm{X}}$ and $\mathscr{G}_{\mathrm{Z}}$ in terms of the edges of $\mathscr{G}_1$ and $\mathscr{G}_2$, we define the set of variable nodes connected to every check node in $\mathcal{C}_{\mathrm{X}}$ and $\mathcal{C}_{\mathrm{Z}}$.
Let $\mathcal{N}_{\left(c^2_jv^1_i\right)}^{\mathrm{X}}$ denotes the set of variable nodes connected to check node $\left(c^2_jv^1_i\right) \in \mathcal{C}_{\mathrm{X}},$ where $c^2_j \in \mathcal{C}_{2}$ and $v^1_i \in \mathcal{V}_{1}.$
Then $\mathcal{N}_{\left(c^2_jv^1_i\right)}^{\mathrm{X}}$ is given by
\begin{equation}
\mathcal{N}_{\left(c^2_jv^1_i\right)}^{\mathrm{X}} = \left(\mathcal{N}^2_{c^2_j} \times v^1_i \right) \cup \left( c^2_j \times \mathcal{N}^1_{v^1_i}\right),
    \label{eq:hp_cx_neighbor}
\end{equation}
where $ \mathcal{N}_{c_j^2}^2$ and $\mathcal{N}_{v_i^1}^1$ are neighbors of $ c_{j}^2$ in $\mathscr{G}_2$ and $ v_{i}^1$ in $\mathscr{G}_1$,
respectively.
Similarly $\mathcal{N}_{\left(v^2_ic^1_j\right)}^{\mathrm{Z}}$ denotes the set of variable nodes connected to check node $\left(v^2_ic^1_j\right) \in \mathcal{C}_{\mathrm{Z}},$ and is given by
\begin{equation}
\mathcal{N}_{\left(v^2_ic^1_j\right)}^{\mathrm{Z}} =  \left(\mathcal{N}^2_{v^2_i} \times c^1_j \right) \cup \left(v^2_i \times \mathcal{N}^1_{c^1_j}\right),
    \label{eq:hp_cz_neighbor}
\end{equation}
where $\mathcal{N}_{v_i^2}^2 $ and $\mathcal{N}_{c_j^1}^1$ are the neighbors of $ v_{i}^2$ in $\mathscr{G}_2$ and $c_{j}^1$ in $\mathscr{G}_1,$
respectively.
\subsubsection{Lifted-product codes}
\label{sec:LP_codes}
LP codes are the generalization of the HP codes. 
Consider the base matrices $\mathbf{B}_1 \in \mathbb{F}_2^{m_{B_1}\times n_{B_1}}$ and $\mathbf{B}_2 \in \mathbb{F}_2^{m_{B_2}\times n_{B_2}}$ corresponding to two classical quasi-cyclic LDPC codes. 
Let the non-zero entries of matrix $\mathbf{W}_1 \in \mathcal{R}_l^{m_{B_1}\times n_{B_1}}$ $(\text{ or }\mathbf{W}_2 \in \mathcal{R}_l^{m_{B_2}\times n_{B_2}})$ represent the circulants corresponding to the non-zero entries of  $\mathbf{B}_1$  $(\text{ or }\mathbf{B}_2 )$.
Given these two classical base matrices, the LP construction gives the base matrices corresponding to the $X-$stabilizers, denoted by $\mathbf{B}_{\mathrm{X}},$ and $Z-$stabilizers,  denoted by $\mathbf{B}_{\mathrm{Z}},$ which are given by
\begin{align*}
    \mathbf{B}_{\mathrm{X}} & = \begin{bmatrix}
        \mathbf{B}_1 \otimes \mathbf{I}_{n_{B_2}} & \mathbf{I}_{m_{B_1}} \otimes \mathbf{B}_2^{\mathsf{T}}
    \end{bmatrix},\\
     \mathbf{B}_{\mathrm{Z}} & = \begin{bmatrix}
       \mathbf{I}_{n_{B_1}} \otimes \mathbf{B}_2 &   \mathbf{B}_1^{\mathsf{T}} \otimes  \mathbf{I}_{m_{B_2}}   \\
    \end{bmatrix}.
    \label{eq:lifted_product_base_matrices}
\end{align*}
To obtain a CSS code, two matrices $\mathbf{W}_{\mathrm{X}}$ and $\mathbf{W}_{\mathrm{Z}}$, respectively, representing circulants of $\mathbf{B}_{\mathrm{X}}$ and $\mathbf{B}_{\mathrm{Z}}$ such that the resulting 
parity-check matrices $\mathbf{H}_{\mathrm{X}}$ and $\mathbf{H}_{\mathrm{Z}}$ 
from lifting satisfy the CSS condition given in \eqref{eq:CSS}.
The details of the lifting process are described in Section~\ref{sec:protograph-ldpc-code}.
Matrices $\mathbf{W}_{\mathrm{X}}$ and $\mathbf{W}_{\mathrm{Z}}$ are obtained from their classical counterparts $\mathbf{W}_1$ and $\mathbf{W}_2$ as 
\begin{align}
   \nonumber \mathbf{W}_{\mathrm{X}} & = \begin{bmatrix}
        \mathbf{W}_1 \otimes \mathbf{I}_{n_{B_2}} & \mathbf{I}_{m_{B_1}} \otimes \mathbf{W}_2^{*}
    \end{bmatrix},\\
     \mathbf{W}_{\mathrm{Z}} & = \begin{bmatrix}
       \mathbf{I}_{n_{B_1}} \otimes \mathbf{W}_2 &   \mathbf{W}_1^{*} \otimes  \mathbf{I}_{m_{B_2}}  \\
    \end{bmatrix},
\end{align}
where $\mathbf{W}^{*}(i,j)= \left(\mathbf{W}(j,i)\right)^{*}$ and for any $r \in \mathcal{R}$
\begin{equation}
    r^{*}=
    \begin{cases}
        r^{-1}, \text{ if } r \in \mathcal{R}_l \setminus \{0\}\\
        0, \text{ if } r=0.
    \end{cases}
    \label{eq:circulant_conjugate}
\end{equation}
If $\mathbf{H}_{\mathrm{X}}$ (or $\mathbf{H}_{\mathrm{Z}}$) is a binary matrix of size $m_{\mathrm{X}} \times n,$ $(\text{ or }m_{\mathrm{Z}} \times n),$ then $n=\gamma(n_An_B+m_Am_B)$ and $m_{\mathrm{X}}=\gamma m_An_B$ $(\text{ or }m_{\mathrm{Z}}=\gamma n_Am_B).$ 
Assuming that $\mathbf{H}_{\mathrm{X}}$ and $\mathbf{H}_{\mathrm{Z}}$ are of full rank, we have a $\left(n,n-m_{\mathrm{X}} - m_{\mathrm{Z}}\right)$ quasi-cyclic QLDPC code.
Let $\underline{\mathscr{G}}_1, \underline{\mathscr{G}}_2, \underline{\mathscr{G}}_{\mathrm{X}},$ and $\underline{\mathscr{G}}_{\mathrm{Z}}$ denote the base graphs corresponding to base matrices $\mathbf{B}_1, \mathbf{B}_2, \mathbf{B}_{\mathrm{X}}, $ and $\mathbf{B} _{ \mathrm{Z}},$ respectively.
The graphical description of obtaining the base graphs $\underline{\mathscr{G}}_{\mathrm{X}}, $ and $\underline{\mathscr{G}}_{\mathrm{Z}}$  given two classical base graphs $\underline{\mathscr{G}}_1,$ and $ \underline{\mathscr{G}}_2$ is the same as obtaining $\mathscr{G}_{\mathrm{X}}$ and $\mathscr{G}_{\mathrm{Z}}$ given two classical Tanner graphs in $\mathscr{G}_{1}$ and $\mathscr{G}_{2}$ in the HP construction.
From \eqref{eq:hp_cx_neighbor}, it is evident that $\underline{c}^2_j\underline{v}^1_i$ is connected to $\underline{v}^2_{i'}\underline{v}^1_i$s and $\underline{c}^2_{j}\underline{c}^1_{j'}$s where $\underline{v}^2_{i'} \in \mathcal{N}^2_{\underline{c}^2_j}$ and $\underline{c}^1_{j'} \in \mathcal{N}^1_{\underline{v}^1_i}.$ 
The edge that connects the $X$-type check node $\underline{c}^2_j\underline{v}^1_i$ to CC-type (or VV-type) variable node $\underline{c}^2_j\underline{c}_{j'}^1$  (or $v^2_{i'}v^1_{i}$) has the same circulant as the edge that connects  $v^1_i$ (or $c^2_{j}$) to  $\underline{c}^1_{j'}$ (or $\underline{v}^2_{i'}$) in graph $\underline{\mathscr{G}}_{1}  (or \underline{\mathscr{G}}_{2}).$
Similarly, the edge that connects the $Z$-type check node $\underline{v}^2_i\underline{c}^1_j$ to VV-type (or CC-type) variable node $\underline{v}^2_i\underline{v}^1_{i'} (\text{ or }\underline{c}^2_{j'}\underline{c}^1_{j})$ will have circulant $r^{*},$ as defined in \eqref{eq:circulant_conjugate}, if the edge that connects $\underline{c}^1_j (\text{ or }\underline{v}^2_i)$ to $\underline{v}^1_{i'} (\text{ or }\underline{c}^2_{j'})$ in graph $\mathscr{G}_1 (\text{ or }\mathscr{G}_2)$ have circulant $r$.
\begin{example}
Consider two base matrices $\mathbf{B}_1 \in \mathbb{F}_{2}^{2 \times 3}$ and $\mathbf{B}_2  \in \mathbb{F}_{2}^{2 \times 3}$
\begin{equation*}
\mathbf{B}_1=
    \begin{bmatrix}
        1 & 1 & 0\\
        0 & 1 & 1\\
    \end{bmatrix},
    \mathbf{B}_2=
    \begin{bmatrix}
        1 & 1 & 1\\
        1 & 1 & 1\\
    \end{bmatrix}.
\end{equation*}
Weight matrices $\mathbf{W}_1 \in \mathcal{R}_{2}^{2 \times 3}$ and $\mathbf{W}_2 \in \mathcal{R}_{2}^{2 \times 3},$ respectively, corresponding to $\mathbf{B}_1$ and $\mathbf{B}_2$ are given by
\begin{equation*}
\mathbf{W}_1=
    \begin{bmatrix}
        x^1 & x^0 & 0\\
        0 & x^0 & x^1\\
    \end{bmatrix},
    \mathbf{W}_2=
    \begin{bmatrix}
        x^0 & x^1 & x^0\\
        x^1 & x^0 & x^1\\
    \end{bmatrix}.
\end{equation*}
The base graph corresponding to $\mathbf{B}_1$ and $\mathbf{B}_2$ are denoted by  $\underline{\mathscr{G}}_1$ and $\underline{\mathscr{G}}_2,$ respectively. 
The product of two nodes, one from  $\underline{\mathscr{G}}_1$ and the other from  $\underline{\mathscr{G}}_2,$ are represented as shown below.

\vspace{0.1in}
\centering
\tikzstyle{cnode}=[circle,minimum size=0.55 cm,draw]
\tikzstyle{scnode}=[circle,minimum size=0.15 cm,draw]
\tikzstyle{zscnode}=[circle,minimum size=0.15 cm,draw,fill=red]
\tikzstyle{zrnode}=[rectangle,draw,minimum width=0.4 cm,minimum height=0.4cm,outer sep=0pt]
\tikzstyle{cgnode}=[circle,draw]
\tikzstyle{crnode}=[circle,draw]
\tikzstyle{conode}=[rectangle,draw]
\tikzstyle{cpnode}=[circle,draw]
\tikzstyle{rnode}=[rectangle,draw,minimum width=0.4 cm,minimum height=0.4cm,outer sep=0pt]
\tikzstyle{srnode}=[rectangle,draw,minimum width=0.1 cm,minimum height=0.1cm,outer sep=0pt]
\tikzstyle{prnode}=[rectangle,rounded corners,fill=blue!50,text width=4.5em,text centered,outer sep=0pt]
\tikzstyle{prnodebig}=[rectangle,rounded corners,fill=blue!50,text width=7em,text centered,outer sep=0pt]
\tikzstyle{prnodesimple}=[rectangle,draw,text width=4.5em,text centered,outer sep=0pt]
\tikzstyle{bigsnake}=[fill=green!50,snake=snake,segment amplitude=4mm, segment length=4mm, line after snake=1mm]
\tikzstyle{smallsnake}=[snake=snake,segment amplitude=0.7mm, segment length=4mm, line after snake=1mm]
\definecolor{color1}{rgb}{1,0.2,0.3}
\definecolor{color2}{rgb}{0.4,0.5,0.7}
\definecolor{color3}{rgb}{0.1,0.8,0.5}
\definecolor{color4}{rgb}{0.5,0.3,1}
\definecolor{color5}{rgb}{0.5,1,1}
\definecolor{color6}{rgb}{0.8,0.3,0.6}
\definecolor{color7}{rgb}{0.6,0.4,0.3}

\begin{tikzpicture}[every node/.style={scale=1}]
\begin{scope}[node distance=1.5cm,semithick]
\node[cnode] (v1) {};
\node (m1) [right of = v1,xshift=-0.75cm] {$\times$};
\node[cnode] (v2) [right of =m1,xshift=-0.75cm]{};
\node (arrow1)[right of =v2,xshift=-0.65cm]{$=$};
\node[cnode] (vv)[right of=arrow1,xshift=-0.65cm] {};
\node [scnode](vv_in)[right of=arrow1,xshift=-0.65cm]{};
\node [below of=v2,yshift=0.8 cm]{\textbf{VV-type}};
\node[rnode] (c1) [right of =vv]{};
\node (m2) [right of = c1,xshift=-0.75cm] {$\times$};
\node[rnode] (c2) [right of =m2,xshift=-0.75cm]{};
\node (arrow2)[right of =c2,xshift=-0.65cm]{$=$};
\node[rnode] (cc)[right of=arrow2,xshift=-0.65cm] {};
\node[srnode] (cc_in)[right of=arrow2,xshift=-0.65cm] {};
\node [below of=c2,yshift=0.8 cm]{\textbf{CC-type}};

\node[rnode] (c3) [below of =v1]{};
\node (m3) [right of = c3,xshift=-0.75cm] {$\times$};
\node[cnode] (v3) [right of =m3,xshift=-0.75cm]{};
\node (arrow3)[right of =v3,xshift=-0.65cm]{$=$};
\node[rnode] (cv)[right of=arrow3,xshift=-0.65cm] {};
\node [scnode](cv_in)[right of=arrow3,xshift=-0.65cm]{};
\node [below of=v3,yshift=0.8 cm]{\textbf{CV-type}};
\node[cnode] (v4) [right of =cv]{};
\node (m4) [right of = v4,xshift=-0.75cm] {$\times$};
\node[rnode] (c4) [right of =m4,xshift=-0.75cm]{};
\node (arrow4)[right of =c4,xshift=-0.65cm]{$=$};
\node[cnode] (vc)[right of=arrow4,xshift=-0.65cm] {};
\node [srnode](vc_in)[right of=arrow4,xshift=-0.65cm]{};
\node [below of=c4,yshift=0.8 cm]{\textbf{VC-type}};
\end{scope}
\end{tikzpicture}
\vspace{0.1in}

The process of obtaining the base graph corresponding to the quantum code as the graph product of base graphs corresponding to $\mathbf{B}_1$ and $\mathbf{B}_2$ is illustrated in Fig.~\ref{fig:lifted_product_illustration}. 
Even though we have not labeled the nodes in the product graph shown in Fig.~\ref{fig:base_graph_product}, their labels are evident from the labels of nodes in $\underline{\mathscr{G}}_1$ and $\underline{\mathscr{G}}_2.$
Consider the neighborhood of $X$-type or $X$-check node $\underline{c}_2^2\underline{v}_1^1$ in the product graph obtained by taking the product of check node $\underline{c}_2^2$ in $\underline{\mathscr{G}}_2$ and variable node $\underline{v}_1^1$ in $\underline{\mathscr{G}}_1$.
Observe that $\mathcal{N}_{\underline{v}_1^1}^1=\{\underline{c}_1^1\} \text{ and } \mathcal{N}_{\underline{c}_2^2}^2=\{\underline{v}_1^2,\underline{v}^1_2,\underline{v}_3^2\}, $
which implies $\mathcal{N}^\mathrm{X}_{\left(\underline{c}_2^2\underline{v}^1_1\right)} = \{\underline{v}_1^2\underline{v}_1^1,\underline{v}_2^2\underline{v}_1^1,\underline{v}_3^2\underline{v}_1^1,\underline{c}_2^2\underline{c}_1^1\},$ according to \eqref{eq:hp_cx_neighbor}.
Since the edge that connects check node $\underline{c}_2^2$ to variable node $\underline{v}^1_2$ has circulant $x^1$ in $\underline{\mathscr{G}}_2,$ the edge that connects check node $\underline{c}_2^2\underline{v}^1_1$ to variable node $\underline{v}_1^2\underline{v}_1^1$ has also circulant $x^1.$
Similarly, the edge that connects check node $\underline{v}_3^2\underline{c}^1_1$ to variable node $\underline{c}_2^2\underline{c}_1^1$ has circulant $(x^1)^*,$ since the edge that connects variable node $\underline{v}_3^2$ to check node $\underline{c}^2_2$ has circulant $x^1$ in $\underline{\mathscr{G}}_2.$
Note that in this case, $(x^1)^*=x^1$ as the lifting size $\gamma=2.$ The circulants corresponding to other edges in the product graph shown in Fig.~\ref{fig:base_graph_product} can be determined in the same way. The base matrices $\mathbf{B}_\mathrm{X}$ and $\mathbf{B}_\mathrm{Z}$ can be, respectively, lifted to $\mathbf{H}_\mathrm{X}$ and $\mathbf{H}_\mathrm{Z}$ as illustrated in Example~\ref{ex:lifting_illustration}. 
\begin{figure*}
\begin{subfigure}[b]{\textwidth}
    \centering
    \input{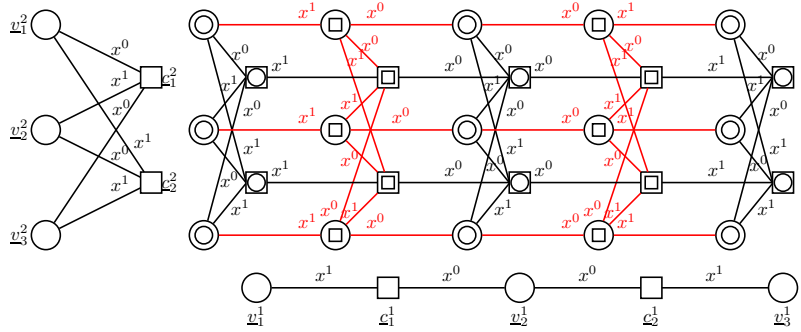}
    \caption{The figure illustrates product of two classical protographs. The edges connected to $X$-type check nodes are shown in red, whereas the edges connected to $Z$-type checks are shown in black. The edges are labeled with their corresponding circulants.}
    \label{fig:base_graph_product}
\end{subfigure}
\begin{subfigure}[b]{0.45\textwidth}
    \centering
    \input{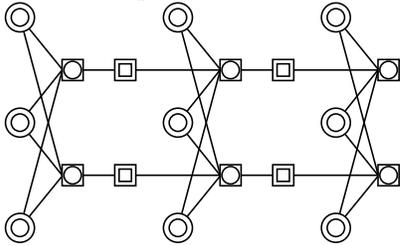}
    \caption{The figure shows the protograph corresponding to $\mathbf{B}_{\mathrm{X}}$.}
    \label{fig:X_base_graph}
\end{subfigure}
\begin{subfigure}[b]{0.45\textwidth}
    \centering
    \input{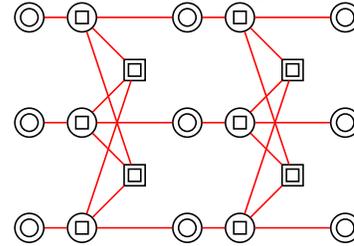}
    \caption{The figure shows the protograph corresponding to $\mathbf{B}_{\mathrm{Z}}$.}
    \label{fig:Z_base_graph}
\end{subfigure}
\caption{The figure illustrates the process of obtaining two base graphs corresponding to a lifted-product code given the base graphs of two classical protograph LDPC codes.}
\label{fig:lifted_product_illustration}
\end{figure*}
\label{ex:lifted_product_illustration}

\end{example}

\section{Trapping sets of QLDPC codes}
\label{sec:TS_QLDPC}

In this section, we introduce trapping sets (TSs) and investigate their structures in the HP and LP codes when decoded using the bit-flipping decoder. We consider the bit-flipping decoder because of its suitability for analysis.
Specifically, in Section~\ref{sec:QTS_hypergraph}, we examine the decoding dynamics within the subgraph formed by the stabilizer generators. Here, decoding dynamics refers to analyzing the convergence of the decoder when the support of the error pattern is entirely within the subgraph. In Section~\ref{sec:QTS_induced_by_combinations_of_gen}, we extend this study to stabilizer-induced subgraphs that are linear combinations of several stabilizer generators. 
Studying the decoding dynamics in these induced subgraphs becomes intractable as the sizes of the induced subgraphs grow. 
To address this challenge, we propose a concise representation of stabilizer-induced graphs that still allows us to study the decoding dynamics.
In the next section, we further elaborate on how understanding the different modes of failure of the simple bit-flipping decoder gives us insight into designing better iterative decoders.

Since we assume the independent depolarizing channel, it suffices to study the performance of the iterative decoder on either $\mathbf{H}_{\mathrm{X}}$ or $\mathbf{H}_{\mathrm{Z}}.$
In what follows, we focus on correcting $X$-type errors by decoding $\mathbf{H}_{\mathrm{Z}}.$ 
Given an error vector $\mathbf{x}$ and the measured syndrome $\boldsymbol{\sigma}^{\mathsf{T}} = \begin{bmatrix}\sigma_1 & \sigma_2 & \cdots & \sigma_{m_{\mathrm{Z}}}\end{bmatrix},$ the bit-flipping decoder outputs an estimate of the error pattern, denoted by $\hat{\mathbf{x}}.$ 
The decoder initializes the estimated error $\hat{\mathbf{x}}$ to $\boldsymbol{0}.$ In subsequent iterations, the decoder first computes the mismatched syndrome, denoted by $\hat{\boldsymbol{\sigma}},$  then flips the estimated error corresponding to a variable node if more than half of its neighboring checks have a mismatched syndrome or are unsatisfied. The decoding continues until the estimated error does not change or the iteration number reaches a predefined maximum value. 
The steps of the bit-flipping decoder are listed in Algorithm~\ref{alg:bit_flipping_decoder}.
\SetKwComment{Comment}{/* }{ */}
\begin{algorithm}
    \caption{Bit-flipping decoder}\label{alg:bit_flipping_decoder}
    \KwData{$\boldsymbol{\sigma},\mathbf{H}_{\mathrm{Z}} $}
\KwResult{$\hat{\mathbf{x}}$}
$\hat{\mathbf{x}} \gets \boldsymbol{0}$\ \Comment*[r]{0 denotes the all-zero vector}
\While{$\mathbf{H}_{\mathrm{Z}} \hat{\mathbf{x}} \neq \boldsymbol{\sigma} (\text{ mod 2 })$}{
$\hat{\boldsymbol{\sigma}} \gets \mathbf{H}_{\mathrm{Z}} \hat{\mathbf{x}} ( \text{mod 2}) $\;
$\boldsymbol{\beta} \gets \hat{\boldsymbol{\sigma}} \oplus \boldsymbol{\sigma}$\;
$\boldsymbol{\alpha} = \mathbf{H}_{\mathrm{Z}}^{\mathsf{T}}\boldsymbol{\beta}$\;
\For{$i=1$ to $n$}{
\If{$\boldsymbol{\alpha}_i > |\mathcal{N}_{v_i}|/2 $}{
$\hat{\mathbf{x}}_i=\hat{\mathbf{x}} \oplus 1 ( \text { mod 2})$\;
}
}

}
\end{algorithm}
The iteration number is not explicitly specified in the variables since it is not critical to understanding the decoder. In what follows, the superscript $t$ of any variable node indicates the iteration number. 
During the decoding process, we say that the check node $c_j$ is satisfied if there exists a positive integer $I_\mathrm{j}$ such that for all $t \geq I_\mathrm{j}$, $\hat{\sigma}^{(t)}(j) = \sigma(j)$. 
We say that the variable node $v_i$ has converged if there exists a positive integer $I_\mathrm{i}$ such that for all $t \geq I_\mathrm{i}$, $\hat{{x}}^{(t)}(i)= \hat{{x}}^{(t-1)}(i)$. 
Note that $\hat{{x}}^{(t)}(i)$ is not necessarily the correct estimate of the error on the $i-\text{th}$ variable node. 
With these definitions, we next formally define TSs. 
\begin{definition}
\label{Def:TrappingSetSyndromeModified}
A trapping set $\mathscr{T}$ for a syndrome-based iterative decoder  is a non-empty set of variable nodes in a Tanner graph $\mathscr{G}$ that have not converged or are neighbors of the check nodes that are not satisfied. If the subgraph $\mathscr{G(T})$ induced by such a set of variable nodes has $a$ variable nodes and $b$ unsatisfied check nodes, then $\mathscr{T}$ is classified as an $(a,b)$ trapping set.
\end{definition}

\subsection{TSs induced by stabilizer generators}
\label{sec:QTS_hypergraph}
 In classical LDPC codes, short cycle compositions cause TSs, leading to decoding failure. QLDPC codes have unique TSs called \emph{stabilizer-induced} TSs or \emph{quantum} TSs (QTS), identified in~\cite{QTS}. Recall that the $Z$-logicals in $\mathcal{U}_{\mathrm{X}}^{\perp} \subset \mathcal{U}_{Z},$ are \emph{degenerate} logicals acting trivially on code states. Consider the subgraph induced by a logical from $\mathcal{U}_{\mathrm{X}}^{\perp}$. 
  Assume that it is a stabilizer generator of $\mathbf{H}_{\mathrm{X}},$ denoted by $\mathbf{h}$. 
   Select error patterns $\mathbf{x}$ and $\mathbf{y}$ with trivially intersecting supports and $\mathbf{x} \oplus \mathbf{y} = \mathbf{h}$. 
   Since stabilizer generator $\mathbf{h}$ is an error \emph{degenerate}, $\mathbf{H}_{\mathrm{Z}}\mathbf{h} = \boldsymbol{0}$ implies that error patterns $\mathbf{x}$ and $\mathbf{y}$ lead to the same syndrome.
   Should the neighborhoods of error patterns $\mathbf{x}$ and $\mathbf{y}$ exhibit a particular symmetry within the Tanner graph, both $\mathbf{x}$ and $\mathbf{y}$ emerge as viable candidates for the error estimate, which causes a decoding failure. 
   In this section, we investigate the symmetries present in the Tanner graph that lead to decoder failure.
   This investigation holds significance due to the low Hamming weight of $\mathbf{h}$, as it is a stabilizer generator of the QLDPC code. 
   As a result, the supports of $\mathbf{x}$ and $\mathbf{y}$ should also have a low weight. 
   Furthermore, it is important to note that the Hamming weights of $\mathbf{x}$ and $\mathbf{y}$ do not grow as the minimum distance of the code grows, since they are derived from stabilizer generators whose weights are generally constant and do not grow as the number of qubits increases. 
   This suggests that the decoder may not be able to estimate a low-weight error pattern, regardless of the minimum distance of the code. 

Consider Tanner graphs $\mathscr{G}_{\mathrm{X}}$ and $\mathscr{G}_{\mathrm{Z}}$, respectively, corresponding to parity-check matrices $\mathbf{H}_{\mathrm{X}}$ and $\mathbf{H}_{\mathrm{Z}}$ of an HP code, which are obtained by taking the graph product of two classical Tanner graphs $\mathscr{G}_1$ and $\mathscr{G}_2$ as described in Section~\ref{sec:hypergraph_product}. 
Next, we focus on failures of the bit-flipping decoder when used to decode $\mathscr{G}_{\mathrm{Z}}$ or correct $X$ errors.
To do so, it is evident from our earlier discussion that we need to investigate the structure of the subgraphs of $\mathscr{G}_{\mathrm{Z}}$ that are induced by low-weight $X$-type stabilizers.
Since the degenerate errors lie in the row space of $\mathbf{H}_{\mathrm{X}}$, we first investigate the subgraphs induced by the rows of $\mathbf{H}_{\mathrm{X}}$ and then systematically expand our investigation to include the subgraphs induced by other low-weight $X-$stabilizers.
Recall from the graphical description of HP codes, the variable nodes connected to an $X$-check/row of $\mathbf{H}_{\mathrm{X}},$ say $\mathbf{c}^{\mathrm{X}}=c_j^2v_i^1,$ is given by 
$$\mathcal{N}_{c_j^2v_i^1}^{\mathrm{X}} = \left(\mathcal{N}^2_{c^2_j} \times v^1_i \right) \cup \left( c^2_j \times \mathcal{N}^1_{v^1_i}\right).$$
The set of $Z$-type check nodes that are connected to a variable node in $\mathcal{N}_{c_j^2v_i^1}^{\mathrm{X}}$ is given by $\cup_{v \in \mathcal{N}_{c_j^2v_i^1}^{\mathrm{X}}} \mathcal{N}_v^{\mathrm{Z}}.$
We denote the subgraph induced by the neighboring variable nodes of  $c_j^2v_i^1$ in $\mathscr{G}_{\mathrm{Z}}$ by $\mathscr{T}(c_j^2v_i^1).$
Note that  $\mathscr{T}(c_j^2v_i^1)$ is a bipartite graph with $\mathcal{N}_{c_j^2v_i^1}^{\mathrm{X}}$ as left nodes and $\cup_{v \in \mathcal{N}_{c_j^2v_i^1}} \mathcal{N}_v^{\mathrm{Z}}$ as right nodes, and represented as $\mathscr{T}\left(\mathcal{N}_{c_j^2v_i^1}^{\mathrm{X}} \cup \left(\cup_{v \in \mathcal{N}_{c_j^2v_i^1}} \mathcal{N}_v^{\mathrm{Z}}\right)\right).$

For a visual representation, consider the subgraph induced by a $X$-type check, denoted by $c^2v^1$, which is obtained by taking the graph product of the degree-three variable node $v^1$ with the degree-four check node $c^2$. Figure~\ref{fig:X_neighborhood} shows the subgraph induced by the $X$-type check $c^2v^1$  in $\mathscr{G}_{\mathrm{Z}}$.
It is redrawn in Fig.~\ref{fig:symmetry_in_subgraph} to highlight its symmetry. Observe that CC-type (or VV-type) variable nodes have
non-intersecting sets of neighboring checks, and each check node is connected to exactly one VV-type and one CC-type variable nodes.
These observations have been formalized in Lemma~\ref{lemma:structures_of_graph_induced_by_stabilizers} for the subgraph induced by any stabilizer generator. 
In Fig.~\ref{fig:symmetry_in_subgraph}, CC-type nodes $c^2c_1^1$ and $c^2c_2^1$ are in error. Since they have non-intersecting neighbors, the eight checks connected to them are unsatisfied in Fig.~\ref{fig:symmetry_in_subgraph}. 
As expected, each VV-type node has two unsatisfied checks as neighbors. 
This observation has been formalized in Lemma~\ref{lemma:unsatisfied_checks_in_TS}.
Consider the bit-flipping decoder described in Algorithm~\ref{alg:bit_flipping_decoder}.
If all CC-type variable nodes are in error,  while none of the VV-type variable nodes are, as in Figure~\ref{fig:decoding_illustration_iteration2}, then all the check nodes in $\mathscr{T}(c^2v^1)$ are not satisfied.
In this case, the decoder alternatively predicts errors on all the CC-type and VV-type variable nodes, as shown in Fig.~\ref{fig:decoding_illustration_iteration2} and Fig.~\ref{fig:decoding_illustration_iteration3}.
This has been formalized in Lemma~\ref{lemma:TSs_induced_by_stabilizers}.
Also note that all CC-type (or VV-type) have identical neighborhoods in Fig.~\ref{fig:symmetry_in_subgraph}, implying that when the errors are only on a subset of CC-type (or VV-type) variable nodes, each VV-type (CC-type) variable node has the same number of unsatisfied checks as neighbors. 
The decoding process is illustrated in Fig.~\ref{fig:decoding_illustration} when $c^2c_1^1$ and $c^2c_2^1$ are in error. 
Since half of neighboring checks for each of the VV-type nodes are unsatisfied, the decoder flips all the VV-type variable nodes and the CC-type variable nodes $c^2c_1^1$ and $c^2c_2^1.$ In subsequent iterations, all the check nodes are unsatisfied and, as a result, all the variable nodes are flipped, causing the decoder to oscillate.
    \begin{figure}
        \centering
        \input{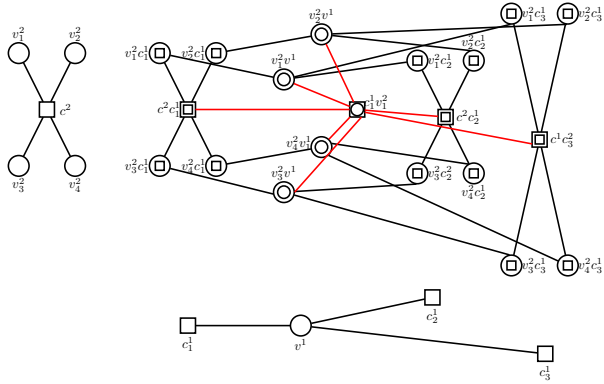}
        \caption{Subgraph $\mathscr{T}(c^2v^1)$ induced by the support of $X$-type check node $c^2v^1$ in $\mathscr{G}_\mathrm{Z}$.  Note that $c^2v^1$ is not part of $\mathscr{T}(c^2v^1).$ The edges connected to $c^2v^1$ in $\mathscr{G}_{\mathrm{X}}$ are shown in red, whereas those in $\mathscr{T}(c^2v^1)$ are shown in black.}
        \label{fig:X_neighborhood}
    \end{figure}
    \begin{figure*}
        \centering
        \input{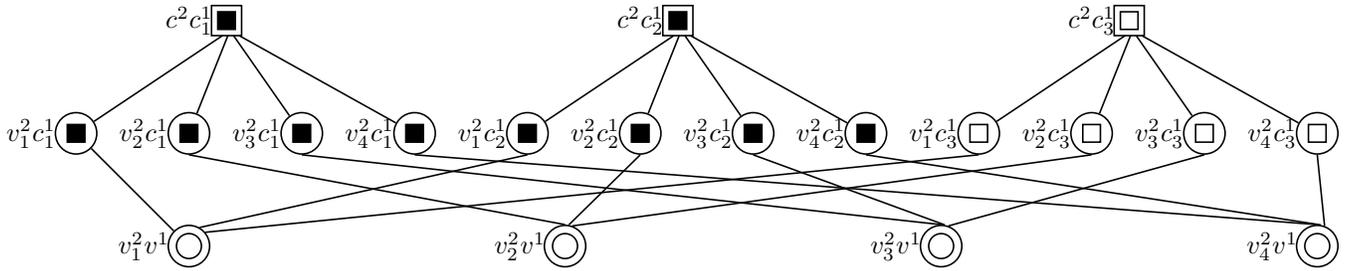}
        \caption{The figure shows a different view of $T(c^2v^1)$ shown in Fig. 3. Note that the CC-type (or VV-type) variable nodes have non-intersecting set of neighboring checks and have identical neighborhoods, meaning when the errors are only on a subset of CC-type (or VV-type) variable nodes, each VV-type (or CC-type) variable node has the same number of unsatisfied checks as neighbors. In this figure, CC-type variable nodes $c^2c^1_1$ and $c^2c_2^1$, shown in black, are in error, causing their eight neighboring check nodes to be unsatisfied. As a result, each of the VV-type variable nodes is connected to two unsatisfied checks. } 
        \label{fig:symmetry_in_subgraph}
    \end{figure*}
\begin{figure*}
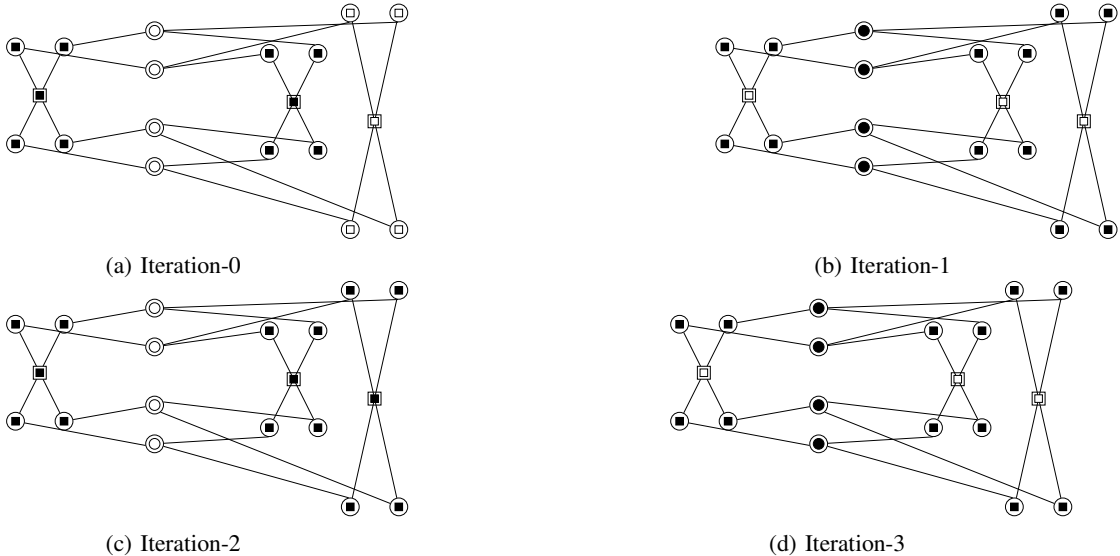

\begin{subfigure}[b]{0.48\textwidth}
\centering
     \scalebox{0.43} {\input{Figures/Journal_figures/decoding_illustration/trapping_set_decoding_iter0}}
     \caption{Iteration-0}
\end{subfigure}
\begin{subfigure}[b]{0.48\textwidth}
\centering
     \scalebox{0.43} {\input{Figures/Journal_figures/decoding_illustration/trapping_set_decoding_iter1}}
     \caption{Iteration-1}
\end{subfigure}
\begin{subfigure}[b]{0.48\textwidth}
\centering
     \scalebox{0.43} {\input{Figures/Journal_figures/decoding_illustration/trapping_set_decoding_iter2}}
     \caption{Iteration-2}
     \label{fig:decoding_illustration_iteration2}
\end{subfigure}
\begin{subfigure}[b]{0.48\textwidth}
\centering
     \scalebox{0.43} {\input{Figures/Journal_figures/decoding_illustration/trapping_set_decoding_iter3}}
     \caption{Iteration-3}
     \label{fig:decoding_illustration_iteration3}
\end{subfigure}
\caption{The figure shows the decoding iterations when two CC-type variable nodes are in error. The unsatisfied checks
are shown in black. Also, the variable nodes on which the estimated error does not match the actual are shown in black. As
expected, at the beginning, each VV-type variable node is connected to two unsatisfied checks and, hence, flipped. From the second
iteration onwards, all the checks are unsatisfied, and as a result, the neighboring checks of every variable node are unsatisfied
and, hence, flipped. From iteration 3, the decoder predicts error on either all CC-type or VV-type variable nodes.}
\label{fig:decoding_illustration}
\end{figure*}
\begin{lemma}
    \label{lemma:TSs_induced_by_stabilizers}
    Let us consider an HP code that is characterized by two Tanner graphs, $\mathscr{G}_{\mathrm{X}}$ and $\mathscr{G}_{\mathrm{Z}}$, derived from two classical codes, which themselves are associated with Tanner graphs $\mathscr{G}_1$ and $\mathscr{G}_2$. Let $\mathscr{G}_1$ and $\mathscr{G}_{2}$ be $(d_c^1,d_v^1)$ and  $(d_c^2,d_v^2)$  regular graphs, respectively. Denote by $\mathscr{T}(c^2_jv^1_{i})$ the subgraph that is induced by the $X-$check $c^2_jv^1_{i}$ on the Tanner graph $\mathscr{G}_{\mathrm{Z}}$.  If neither $\mathscr{G}_1$ nor $\mathscr{G}_2$ contains cycles of length four, then the following statements for the decoder described in Algorithm~\ref{alg:bit_flipping_decoder} hold.
    \begin{enumerate}
        
        \item \label{lemma1_part1} Assuming that the support for the error pattern lies precisely in the variable nodes of type VV or CC, the location of the mismatched error alternates between the VV-type nodes and the CC-type nodes, rendering $\mathscr{T}(c^2_jv^1_{i})$ a TS.
        
        \item \label{lemma1_part2}Let $d_c^2$ (or $d_v^1$) be the degree of CC-type (or VV-type) variable nodes.
       Consider an error pattern that acts nontrivially only on $\alpha$  VV-type and $\beta$ CC-type variable nodes in $\mathscr{T}(c^2_jv^1_{i})$. If $\alpha\geq \lfloor\frac{d_c^2}{2}\rfloor+1$ and $\beta < \lfloor\frac{d_v^1}{2}\rfloor$ or alternatively, $\alpha < \lfloor\frac{d_c^2}{2}\rfloor$ and $\beta \geq \lfloor\frac{d_v^1}{2}\rfloor+1$, then the error pattern constitutes a TS-inducing error pattern. 
    \end{enumerate}
\end{lemma}
\begin{proof}
The proof is given in Appendix~\ref{sec:appendix_proof_of_TSs_induced_by_stabilizers}.
\end{proof}
\begin{figure*}
        \centering
        \begin{subfigure}{0.5\textwidth}
        \input{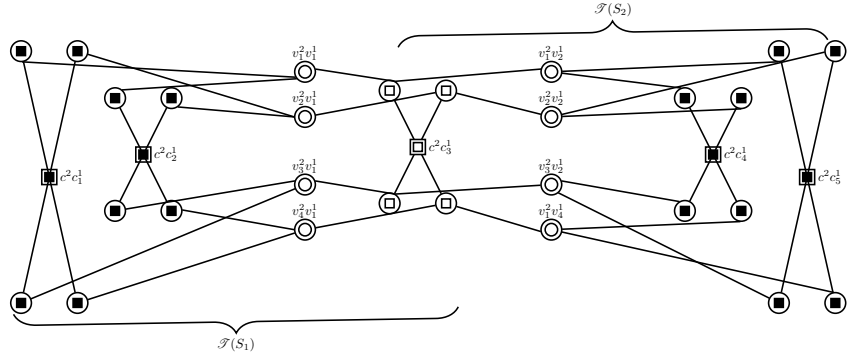}
        \subcaption{Iteration~1}
        \label{fig:illustration_decoding_on_mult_stab_gen_iteration1}
        \end{subfigure}
        \centering
        \begin{subfigure}{0.5\textwidth}
             \input{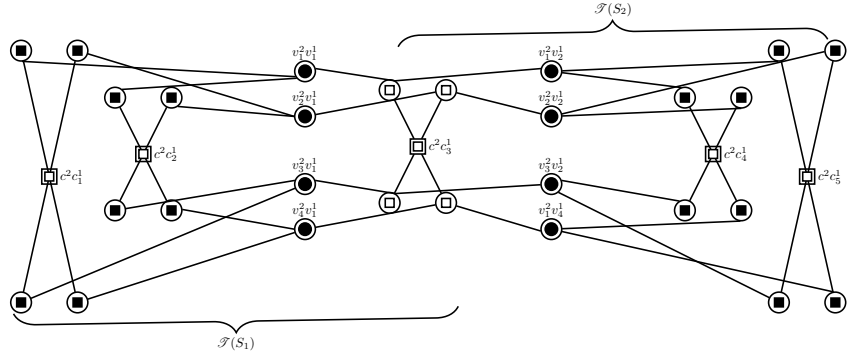}
        \subcaption{Iteration~2}
        \label{fig:illustration_decoding_on_mult_stab_gen_iteration2}
        \end{subfigure}
        \caption{The figure illustrates bit-flipping decoding on the graph induced by stabilizer $\mathbf{h}$, which is the sum of stabilizer generators $S_1$ and $S_2$. The inner shape of the unsatisfied checks and erroneous variable nodes is shown in black. Since the error estimate oscillates between two error patterns, the induced graph constitutes a TS.}
        \label{fig:illustration_decoding_on_mult_stab_gen}
    \end{figure*}
\subsection{QTSs induced by linear combinations of stabilizer generators}
\label{sec:QTS_induced_by_combinations_of_gen}
In Lemma~\ref{lemma:TSs_induced_by_stabilizers}, we characterized the dynamics of the bit-flipping decoder in the Tanner graph $\mathscr{G}_{\mathrm{Z}}$ (or $\mathscr{G}_{\mathrm{X}}$) when errors are exclusively in the subgraph induced by a $X$ (or $Z$) type stabilizer generator, say $c^2v^1$ (or $v^2c^1$).
Examining the TSs within stabilizer-induced subgraphs is crucial, because certain TSs within these graphs possess a small critical number, which do not scale proportionately with the number of stabilizer generators necessary to form the stabilizer in question. 
To fully understand the effect of stabilizers on the decoding process, we next investigate the decoding dynamics when the error lies exclusively in the subgraph induced by a stabilizer that is a linear combination of more than one stabilizer generator.

To understand the structures of stabilizer-induced graphs that are linear combinations of two or more than two stabilizer generators, consider the subgraph shown in Figure~\ref{fig:illustration_decoding_on_mult_stab_gen} induced by stabilizer $\mathbf{h}=S_1 \oplus S_2$, where $S_1$ and $S_2$ are X-type stabilizer generators.
Denote the set of VV-type (or CC-type) variable nodes within $\mathscr{T}(S_1)$ that are present in $\mathscr{T}(S_1)$, but not in $\mathscr{T}(S_2)$, by $\Lambda_{S_1}$ (or $\Gamma_{S_1}$).
Similarly, denote the set of VV-type (or CC-type) variable nodes within $\mathscr{T}(S_2)$ that are present in $\mathscr{T}(S_2)$, but not in $\mathscr{T}(S_1)$, by $\Lambda_{S_2}$ (or $\Gamma_{S_2}$).
Define $\Lambda(\mathbf{h}) \coloneqq \Lambda_{S_1} \cup \Lambda_{S_2}$ and $\Gamma(\mathbf{h}) \coloneqq \Gamma_{S_1} \cup \Gamma_{S_2}$.
For the induced graph shown in Figure~\ref{fig:illustration_decoding_on_mult_stab_gen}, $\Gamma(\mathbf{h})$ contains all CC-type variable nodes within except $c^2c_3^1$ and $\Lambda({\mathbf{h}})$ contains all VV-type variable nodes.
In Figure~\ref{fig:illustration_decoding_on_mult_stab_gen}, consider a VV-type variable node in $\Lambda_{S_i}$ and a CC-type variable node in $\Gamma_{S_i}$, for $i \in \{1,2\}$; There exists a corresponding check node of degree two within $\mathscr{T}(S_i)$, which is connected to both of these variable nodes. This observation is formalized in Lemma~\ref{lemma:TS_compositions_condition_intermediate}.
Also, observe that each of the CC-type (or VV-type) variable nodes in $\Gamma_{S_i}$ ($\Lambda_{S_i}$) is connected to $|\Lambda_{S_i}|=4$ ($|\Gamma_{S_i}|=2$) degree two check nodes whose neighboring variable nodes are in $\Gamma_{S_i}\cup \Lambda_{S_i}$.
From which it follows that when all the variable nodes in $\Gamma(\mathbf{h})$ are erroneous, as in Figure~\ref{fig:illustration_decoding_on_mult_stab_gen_iteration1}, each variable node in $\Gamma_{S_i}$ is connected to at least $|\Lambda_{S_i}|=4$ unsatisfied checks, and each variable node in $\Lambda_{S_i}$ is connected to at least $|\Gamma_{S_i}|=2$ unsatisfied checks.
Since the number of unsatisfied checks exceeds the threshold set by the bit-flipping decoder, it flips all the variable nodes in $\Gamma_{S_i}\cup \Lambda_{S_i}$.
So, at the beginning of iteration~2, as seen in Figure~\ref{fig:illustration_decoding_on_mult_stab_gen_iteration2} all the variable nodes in $\Lambda(\mathbf{h})$ are in error, while all the variable nodes in $\Gamma(\mathbf{h})$ are not.
With arguments similar to those used in iteration~1, it follows that the decoder again flips all variable nodes in $\Gamma_{S_i}\cup \Lambda_{S_i}$, causing the error estimate to oscillate.
In Lemma~\ref{lemma:TS_compositions_condition}, the above arguments have been formalized for any linear combinations of stabilizer generators.

Subsequently, we provide a methodology for enumeration of all stabilizer-induced subgraphs of an HP code, derived from its constituent classical codes. This can be done by directly examining the subgraph induced by stabilizers of type $X$ (or type $Z$) within $\mathscr{G}_{\mathrm{Z}}$ (or $\mathscr{G}_{\mathrm{X}}$). However, it is important to note that these subgraphs can contain two disconnected components; hence, the convergence of the decoder can be assessed independently within each component.
The enumeration of stabilizer-induced subgraphs without any disjoint components can be obtained by examining the graph products resulting from the connected subgraph induced by check nodes within $\mathscr{G}_2$ with connected subgraphs induced by variable nodes within $\mathscr{G}_1$.
Next, we introduce the necessary notation required to formalize this notion.
Recall that the $X$ type check nodes are labeled $c^2v^1$, where $c^2$ is a check node in Tanner graph $\mathscr{G}_2$ and $v^1$ is a variable node in the Tanner graph $\mathscr{G}_1$.
Since the rows of the matrix $\mathbf{H}_{\mathrm{X}}$ are the binary form of $X$-type generators, these rows can be labeled as $c^2v^1$. 
Consider the $X$-type stabilizer $\mathbf{h}=\sum_{a\in \mathcal{I}, b \in \mathcal{J}}c_b^2v_a^1$, where $c_b^2$, for $b \in \mathcal{J}$, is a check node in $\mathscr{G}_2$ and $v_a^1$, for $a \in \mathcal{I}$, is a variable node in $\mathscr{G}_1.$
Notice that the subgraph induced by the variable nodes connected to $\mathbf{h}$ within $\mathscr{G}_\mathrm{Z}$ is given by
\begin{equation*}
  \mathscr{T}(\mathbf{h})\coloneqq \cup_{a \in \mathcal{I},b \in \mathcal{J}}\mathscr{T}(c_j^2v_i^1)= \cup_{b \in \mathcal{J}}\mathscr{T}(c_b^2) \times \cup_{a \in \mathcal{I}}\mathscr{T}(v_a^1), 
\end{equation*} 
where $\cup_{b \in \mathcal{J}}\mathscr{T}(c_b^2)$ is a subgraph of $\mathscr{G}_2$ induced by check nodes whose indices lie in $\mathcal{J}$, and $\cup_{a \in \mathcal{I}}\mathscr{T}(v_a^1)$ is a subgraph of $\mathscr{G}_1$ induced by  variable nodes whose indices lie in $\mathcal{I}$.
The dynamics of the decoder within $\cup_{a \in \mathcal{I}\setminus a' , b \in \mathcal{J}\setminus b' }\mathscr{T}(c_j^2v_i^1)$ affects the dynamics within $\mathscr{T}(c_{a'}^2v_{b'}^1)$ only if 
\begin{equation*}
   \left(\cup_{a \in \mathcal{I}\setminus a' , b \in \mathcal{J}\setminus b' }\mathscr{T}(c_b^2v_a^1)\right) \cap  \mathscr{T}(c_{b'}^2v_{a'}^1) \neq \emptyset,
\end{equation*}
which holds true if there is a path either between $c_b^2$ and $c_{b'}^2$, for $b' \in \mathcal{J}\setminus b,$ within $\mathscr{G}_2$ or between $v_a^1$ and $v_{a'}^1$, for $a' \in \mathcal{I}\setminus a$ within $\mathscr{G}_1$.
Consequently, in the analysis of decoder dynamics, it is sufficient to examine graphs that are the graph products of $\cup_{b \in \mathcal{J}}\mathscr{T}(c_b^2)$ and $\cup_{a \in \mathcal{I}}\mathscr{T}(v_a^1)$ when both $\cup_{b \in \mathcal{J}}\mathscr{T}(c_b^2)$ and $\cup_{a \in \mathcal{I}}\mathscr{T}(v_a^1)$ are connected subgraphs of $\mathscr{G}_1$ and $\mathscr{G}_2$, respectively.

\begin{lemma}
    \label{lemma:TS_compositions_condition}
   Consider the stabilizer-induced subgraph in the $Z$ Tanner graph $\mathscr{G}_{\mathrm{Z}}$ of an HP code given by $\mathscr{T}(\mathbf{h})$,
    where $\mathbf{h}$ is a stabilizer of type $X$ formed as a linear combination of stabilizer generators of type $X$, i.e., $\mathbf{h} = \sum_{a \in \mathcal{I},b \in \mathcal{J}}c^2_bv_a^1$.
    Consider the scenario where the CC-type variable nodes have a degree of $d_c^2$, while the VV-type variable nodes have a degree of $d_v^1$.
    If for all $i \in \mathcal{I}$ and $j \in \mathcal{J}$ there exist at least $\left\lfloor\frac{d_v^1}{2}\right\rfloor+1$ CC-type and $\left\lfloor\frac{d_c^2}{2}\right\rfloor+1$ VV-type variable nodes that are only in $\mathscr{T}(c_j^2v_i^1)$ and not in any $\mathscr{T}(c_{j'}^2v_{i'}^1)$ when $i \neq i'$ or $j \neq j'$, then $\mathscr{T}(\mathbf{h})$ is a TS.
    \end{lemma}
\begin{proof}
    The proof is given in Appendix~\ref{sec:appendix_proof_of_sufficient_condtions}.
\end{proof}
\subsection{Concise representation of stabilizer-induced subgraph}
\label{sec:concise_rep_stabilizer_induced_TS}

Note that Lemma~\ref{lemma:TS_compositions_condition} delineates the sufficient conditions under which a stabilizer-induced graph constitutes a TS; however, it does not cover all possible linear combinations of stabilizer generators. To understand the decoding dynamics within the subgraph induced by the linear combinations of stabilizer generators not covered by the sufficient conditions specified in Lemma~\ref{lemma:TS_compositions_condition}, an individual examination of each case is required. When a stabilizer is a combination of two or more generators, the induced subgraph becomes excessively large, making their representation cumbersome. 
In the subsequent discussion, we present a succinct representation of stabilizer-induced graphs, which deliberately omits the depiction of check nodes but nevertheless permits the enumeration of the number of unsatisfied checks associated with each of the variable nodes.
In this concise representation, stabilizer generators are depicted as hexagonal nodes. The set of VV-type or CC-type variable nodes, which are connected to only a single stabilizer generator, is shown within an ellipse, with their connection to the stabilizer generator represented by an edge extending from the ellipse to the respective hexagon. For a variable node that is associated with two or more stabilizer generators, its relationship with the multiple stabilizer generators is represented by edges extending from the variable node to the respective hexagons.
 Figure~\ref{fig:consise_rep_TS_comp} shows the concise representation of the stabilizer-induced subgraph shown in Figure~\ref{fig:TS_comp_map}.

\begin{figure*}
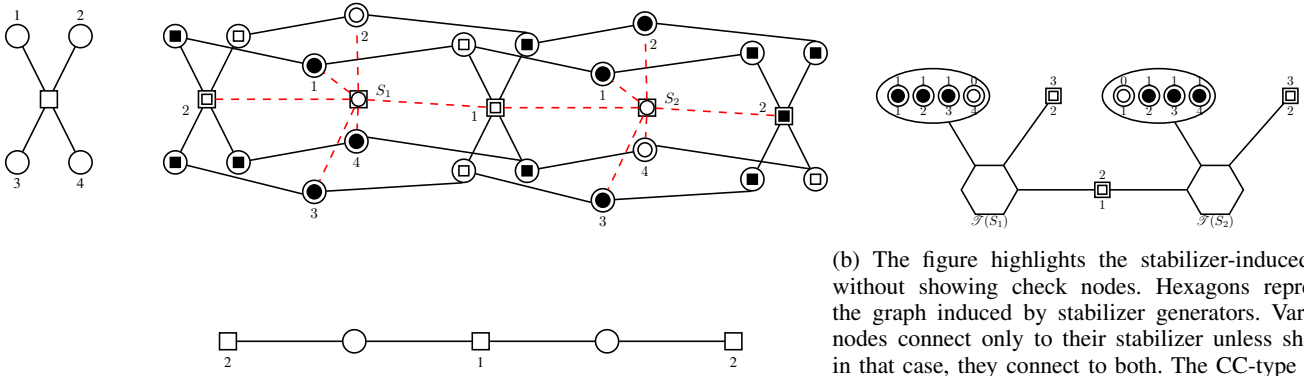

\begin{subfigure}[c]{0.6\textwidth}
\centering
    \input{Figures/Journal_figures/Example_TS_Comp/mapping_TS}
    
    \subcaption{The figure shows the subgraph of $\mathscr{G}_{\mathrm{Z}}$ induced by the combination of two $X$-type stabilizers. The edges from $\mathscr{G}_{\mathrm{Z}}$ are shown in black, whereas the edges from $\mathscr{G}_{\mathrm{X}}$ are shown in red dotted lines. These two $X$-type stabilizers share a CC-type variable node. As per the labeling scheme described in Section~\ref{sec:concise_rep_stabilizer_induced_TS}, the figure illustrates that variable nodes in stabilize-induced subgraphs can be mapped to numbers, ensuring that all the VV-type (or CC-type) variable nodes connected to a check node are assigned the same number. The number associated with a variable node is the label given below it in the figure. }
    \label{fig:TS_comp_map}
\end{subfigure}
\begin{subfigure}[c]{0.38\textwidth}
\centering
    \input{Figures/Journal_figures/Example_TS_Comp/concise_representation_of_TS}

    \subcaption{The figure highlights the stabilizer-induced TS without showing check nodes. Hexagons represent the graph induced by stabilizer generators. Variable nodes connect only to their stabilizer unless shared; in that case, they connect to both. The CC-type node linked to two stabilizers is shared. Numbers below each node follow the labeling scheme described in Section~\ref{sec:concise_rep_stabilizer_induced_TS}. Although implicit, this representation allows calculation of the number of  unsatisfied checks linked to erroneous variable nodes. The number of unsatisfied check counts, is shown above each variable nodes.}
    \label{fig:consise_rep_TS_comp}
\end{subfigure}
\caption{The inner shape of the erroneous variable nodes and unsatisfied checks is depicted in black. It should be noted that enumerating the unsatisfied checks connected to a variable node by applying the Lemma~\ref{lemma:vn_only_graphical_description} to the succinct representation yields an outcome identical to that obtained by directly enumerating the unsatisfied checks.}
 \label{fig:TS_comp}
\end{figure*}

As the concise representation of the stabilizer-induced subgraph omits an explicit representation of the check nodes, enumerating the unsatisfied checks associated with each variable node remains unfeasible within these representations.
To enable the enumeration of unsatisfied checks associated with variable nodes in the concise representation of stabilizer-induced subgraphs, we introduce a scheme to label variable nodes that satisfy the following three properties:

\begin{enumerate}
    \item  All VV-type variable nodes linked to a given check node have the same label; similarly, all CC-type variable nodes connected to that check node have the same label. 
    \item  VV-type (or CC-type) variable nodes corresponding to any two checks in the neighborhood of a CC-type (or VV-type) node possess distinct labels.
    \item  VV-type (or CC-type) variable nodes within a stabilizer-induced subgraph have distinct labels. 
\end{enumerate}
Since the check nodes are not shown explicitly in the concise depiction of the stabilizer-induced subgraph, it is imperative to develop a methodology to distinguish the check nodes associated with a specific variable node. 
This differentiation is crucial to accurately enumerate the unsatisfied checks linked to a variable node, and the first two properties enable us to do so. 
As a result of the initial two properties, check nodes associated with a CC-type (or VV-type) variable node are distinguished by aligning them with the label of their neighboring VV-type (or CC-type) variable nodes.
Therefore, when quantifying the number of unsatisfied checks linked to a CC-type (or VV-type) variable node, its adjacent checks are distinguished from one another based on the labels of the corresponding VV-type (or CC-type) variable nodes to which they are connected. 
Furthermore, given the fact that a check node is connected to a specific VV-type (or CC-type) variable node within a stabilizer-induced subgraph, the third property facilitates the identification of that particular variable node within the stabilizer-induced graph.

To understand the labeling scheme, consider the $Z$-Tanner graph denoted by $\mathscr{G}_{\mathrm{Z}}$ of an HP code constructed as the graph product of Tanner graphs $\mathscr{G}_1$ and $\mathscr{G}_2$.
According to Brook's Theorem on vertex coloring, the vertices of a graph can be assigned labels using a number of labels equal to the graph's maximum degree, ensuring that no two adjacent vertices share the same label\cite{Brooks_1941}.
So, the nodes of the Tanner graph $\mathscr{G}_1$ (or $\mathscr{G}_2$) can be labeled using maximum $\max(d_v^1,d_c^1)$ (or $\max(d_c^2,d_v^2)$) labels such that variable nodes connected to a check node have distinct labels and vice versa.
Consider the set of VV-type variable nodes associated with the $Z$ check $v^2c^1$ of the Tanner graph $\mathscr{G}_{\mathrm{Z}}$. This set is given by $v^2 \times \mathcal{N}_{c^1}^1,$ where $v^2$ is a variable node in the Tanner graph $\mathscr{G}_2$ and $c^1$ is a check node in the Tanner graph $\mathscr{G}_1$.
Assign labels to all VV-type variable nodes of the form $v^2v_{i}^1$, for $v_{i}^1\in \mathcal{N}_{c^1}^1$, using the label of $v^2$ in $\mathscr{G}_2$, thereby ensuring that all VV-type variable nodes associated with the $Z$-type check node $v_i^1c^2$ maintain the same label, as the first property states.
Now consider the CC-type variable nodes connected to $Z$-type check $v^2c^1$ from the Tanner graph $\mathscr{G}_{\mathrm{Z}}$ given by $\mathcal{N}^2_{v^2} \times c^1$.
Now label all CC-type variable nodes of the form $c_{j}^2c^1$, for $c_{j}^2\in \mathcal{N}_{v^2}^2$,using the label of $c^1$ in $\mathscr{G}_1$, which ensures that all CC-type variable nodes connected to the $Z$ check have the same label.
All variable nodes in the Tanner graph $\mathscr{G}_{\mathrm{Z}}$ can be labeled by repeating the procedure for the other check nodes in $\mathscr{G}_{\mathrm{Z}}$.
The labeling procedure is consistent, meaning that no VV-type variable node is assigned multiple labels, as two $Z$-checks $v_{i}^2c^1$ and $v_{i'}^2c^1$ do not share any VV-type variable nodes as neighbors when $i\neq i'$.
In Figure~\ref{fig:TS_comp_map} observe that the labeling scheme also have properties 2 and 3. 
In Appendix~\ref{proof_of_properties_2_and_3}, we argue that the above labeling scheme, in general, have properties 2 and 3.

Next, we illustrate through an example how to compute the number of unsatisfied checks linked to a variable node when the stabilizer-induced graph is labeled. 
In Figure~\ref{fig:consise_rep_TS_comp}, consider the CC-type variable node that is present within the subgraph induced by stabilizer generators $S_1$ and $S_2$. 
For reference, we denote this CC-type variable node by $c^2c^1$.
Each check node in $\mathcal{N}^\mathrm{Z}_{c^2c^1}$ is specifically linked to the CC-type variable node $c^2c^1$. 
According to Lemma~\ref{lemma:TSs_induced_by_stabilizers}, within a subgraph induced by a stabilizer generator, each check node connects exactly to a CC-type and a VV-type variable node. 
Since the stabilizer-induced graph in Figure~\ref{fig:consise_rep_TS_comp} includes only $\mathscr{T}(S_1)$ and $\mathscr{T}(S_2)$, from Lemma~\ref{lemma:TSs_induced_by_stabilizers} it follows that the check nodes in $\mathcal{N}^\mathrm{Z}_{c^2c^1}$ do not connect to any CC-type variable node other than $c^2c^1$. 
From properties 1 and 2 of the labeling scheme, the check nodes in $\mathcal{N}_{c^2c^1}^\mathrm{Z}$ can be distinguished by the labels of their adjacent VV-type variable nodes.
Consequently, VV-type variable nodes sharing the same label in both $\mathscr{T}(S_1)$ and $\mathscr{T}(S_2)$ connect to the same check in $\mathcal{N}^\mathrm{Z}_{c^2c^1}$. 
This shows that a specific check node in $\mathcal{N}^\mathrm{Z}_{c^2c^1}$ is connected to the CC-type variable node $c^2c^1$ and the VV-type variable nodes labeled $1$ in both $\mathscr{T}(S_1)$ and $\mathscr{T}(S_2)$.
In Figure~\ref{fig:TS_comp_map}, the CC-type variable node labeled $1$ and the VV-type variable node labeled $1$ within $\mathscr{T}(S_2)$ are not in error, while the VV-type variable node labeled $1$ in the subgraph induced by $\mathscr{T}(S_1)$ is in error, indicating that the check node considered remains unsatisfied.
Similarly, it can be inferred that the two check nodes within $\mathcal{N}^\mathrm{Z}_{c^1c^2}$, individually associated with VV-type variable nodes labeled $2$ and $3$, are satisfied, while the check node in $\mathcal{N}^\mathrm{Z}_{c^1c^2}$, which is associated with VV-type variable nodes labeled as $4$, remains unsatisfied.
Therefore, in Figure~\ref{fig:consise_rep_TS_comp}, the CC-type variable node labeled $1$ is connected to two unsatisfied checks. 

Subsequently, we determine the quantity of unsatisfied checks associated with the VV-type variable node labeled as $2$ within $\mathscr{T}(S_1)$. 
Let this variable node be represented as $v^2v^1$. 
It should be noted, as per property 2 of the labeling scheme, that the check nodes linked to a VV-type variable node can be distinguished by the label of their adjacent CC-type variable nodes. 
Consider the check node within $\mathcal{N}_{v^2v^1}^\mathrm{Z}$, whose adjacent CC-type variable nodes are labeled $1$. 
Given that the variable node $v^2v^1$ is exclusively located in $\mathscr{T}(S_1)$, according to Lemma \ref{lemma:TSs_induced_by_stabilizers}, it follows that this check node is connected to the VV-type variable node labeled $2$ and the CC-type variable node labeled $1$ within $\mathscr{T}(S_1)$. 
Since the CC-type variable node bearing the label $1$ is also present in $\mathscr{T}(S_2)$, the check node under consideration resides in $\mathscr{T}(S_2)$, further implying that it is connected to the VV-type variable node labeled $2$ within $\mathscr{T}(S_2)$. 
Consequently, this check node is associated with the CC-type variable node labeled $1$ within $\mathscr{T}(S_1)$ and VV-type variable nodes labeled as 2 in both $\mathscr{T}(S_1)$ and $\mathscr{T}(S_2)$. 
Based on the state of variable nodes illustrated in Figure \ref{fig:TS_comp_map}, it can be inferred that this particular check node is satisfied.
Now, consider the check node located within $\mathcal{N}_{v^2v^1}^\mathrm{Z}$, whose neighboring CC-type variable nodes carry the label $2$.
From Lemma \ref{lemma:TSs_induced_by_stabilizers}, it is evident that this check node is linked to the VV-type variable node labeled $2$ and the CC-type variable node labeled $2$ in $\mathscr{T}(S_1)$. Given that the CC-type variable labeled $2$ is exclusively associated with $\mathscr{T}(S_1)$, the check node in question has no connection to any CC-type variable nodes within $\mathscr{T}(S_2)$. 
Consequently, based on the state of variable nodes depicted in Figure \ref{fig:TS_comp_map}, it can be concluded that the check node in $\mathcal{N}_{v_2v_1}^\mathrm{Z}$ with neighboring CC-type variable nodes labeled as 2 is indeed satisfied.
Therefore, the VV-type variable node labeled as 2 within $\mathscr{T}(S_1)$ is associated with precisely one unsatisfied check.

In the following lemma, we formalize the process of determining variable nodes connected to a check node in the concise representation of a stabilizer-induced TS.

\begin{lemma}
\label{lemma:vn_only_graphical_description}  
    Consider the concise representation of a stabilizer-induced subgraph in the $Z$ Tanner graph $\mathscr{G}_{\mathrm{Z}}$ of an HP code given by $\mathscr{T}(\mathbf{h})$,
    where $\mathbf{h}$ is a stabilizer of type $X$ formed as a linear combination of stabilizer generators of type $X$, i.e., $\mathbf{h} = \sum_{a \in \mathcal{I},b \in \mathcal{J}}c^2_bv_a^1$.
    Assume that this concise representation is labeled to satisfy properties 1 to 3. Consider the VV-type (or CC-type) variable node  $c^2c^1$ that has the label $\lambda$ in $\mathscr{T}(\mathbf{h})$.
    \begin{enumerate}
        \item  The collection of VV-type variable nodes connected to a check node in $\mathcal{N}^\mathrm{Z}_{c^2c^1}$, whose associated VV-type variable nodes are labeled as $\rho$,  within $\mathscr{T}(\mathbf{h})$corresponds to the VV-type variable nodes labeled $\rho$ within the subgraph induced by checks of type $X$ in $\mathcal{Q}_{v^2v^1}=\mathcal{N}^\mathrm{X}_{v^2v^1} \cap \mathscr{T}(\mathbf{h})$.
        \item Refer to the check node within $\mathcal{N}_{c^2c^1}^\mathrm{Z}$, which has the VV-type variable nodes labeled $\rho$ as its neighbor, by $v^2c^1$.
       Furthermore, let $\Omega_{v^2c^1}^\rho$ represent the collection of VV-type variable nodes that are linked to $v^2c^1$ within $\mathscr{T}(\mathbf{h})$.
        The collection of CC-type variable nodes connected to $v^2c^1$ corresponds to the CC-type variable nodes labeled $\lambda$ within the subgraph induced by checks of type $X$ in $\mathcal{R}_{v^2c^1}^\rho=\cup_{v^2v^1 \in \Omega_{v_2c_1}^\rho}\mathcal{Q}_{v^2v^1}$.
    \end{enumerate}
\end{lemma}
\begin{proof}
    The proof is given in Appendix~\ref{proof_of_lemma_vn_only_graphical_description}.
\end{proof}


\subsection{Study of decoding dynamics of TS using the concise representation}
Subsequently, we study the decoding dynamics within stabilizer-induced subgraphs that do not meet the sufficient conditions for TSs as established in Lemma~\ref{lemma:TS_compositions_condition}.
If the stabilizer is generated as a linear combination of three or more stabilizer generators, the corresponding subgraph induced by becomes too large to analyze.
Conclusions in Lemma~\ref{lemma:vn_only_graphical_description} facilitate a concise description of the subgraphs induced by a set of stabilizer generators given the common variable nodes in the stabilizer support in the set. 
Using this concise description, in the following theorem, we analyze the decoding dynamics of the subgraphs induced by up to four stabilizer generators and show that these subgraphs constitute TS.
Despite limiting our analysis in the following theorem to subgraphs induced by up to four stabilizer generators due to the exponential increase in the number of possible subgraphs, we conjecture that all subgraphs induced by subsets of stabilizer generators indeed constitute a TS.
\begin{theorem}
    \label{theorem:stabilizer_induced_TS}
     Let $\cup_{j \in \mathcal{J}}\mathscr{T}(c_j^2)$ denote a connected cycle-free subgraph of $\mathscr{G}_2$ induced by $\{c_j^2: j \in \mathcal{J}  \}$, and similarly let $\cup_{i \in \mathcal{I}}\mathscr{T}(v_i^2)$ denote a connected cycle-free subgraph of $\mathscr{G}_1$ induced by $\{v_i^1: i \in \mathcal{I}  \}$.
     Assume that the minimum degree of check nodes is four and the minimum degree of variable nodes is three within both Tanner graphs $\mathscr{G}_1$ and $\mathscr{G}_2$.
   Consider the subgraph formed by taking the graph product of  $\cup_{j \in \mathcal{J}}\mathscr{T}(c_j^2)$ and  $\cup_{i \in \mathcal{I}}\mathscr{T}(v_i^2),$ which is denoted by $\cup_{i \in \mathcal{I},j \in \mathcal{J}}\mathscr{T}(c_j^2v_i^1).$ If the condition $|\mathcal{I}||\mathcal{J}|\leq 4$ holds, the subgraph $\cup_{i \in \mathcal{I},j \in \mathcal{J}}\mathscr{T}(c_j^2v_i^1)$ forms a TS.
\end{theorem}
\begin{proof}
    The theorem asserts that any subgraph induced by a stabilizer, constructed as a linear combination of at most four stabilizer generators, constitutes a TS.
   Lemma~\ref{lemma:TS_compositions_condition} states sufficient conditions under which a subgraph induced by a stabilizer qualifies as a TS. Since sufficient conditions do not include all possible combinations of stabilizer generators, we analyze the cases that do not satisfy sufficient conditions case by case using the concise representation described in Lemma~\ref{lemma:vn_only_graphical_description}. 
   To this end, let us revisit the sufficient conditions derived in Lemma~\ref{lemma:TS_compositions_condition}, which asserts that a stabilizer formed through the linear combination of a set of stabilizer generators, wherein the support of each stabilizer generator within this has a minimum of $\left \lfloor \frac{d_v^1}{2}\right\rfloor+1$ CC-type variable nodes and $\left \lfloor \frac{d_c^2}{2}\right\rfloor+1$ VV-type variable nodes that do not occur in the support of any other generator from the set.
   Given that this theorem addresses stabilizers constituted by the linear combination of up to four stabilizer generators, the support of each of these stabilizer generators can, at most, share three variable nodes with the support of other stabilizer generators involved in the linear combination.
   Consequently, it can be deduced that configurations which fail to meet the sufficient conditions derived in Lemma~\ref{lemma:TS_compositions_condition} must exhibit CC-type variable nodes with degrees four or five, and VV-type variable nodes whose degrees are three, four, or five.
   We systematically examine all instances involving VV-type variable nodes of degree three and CC-type variable nodes of degree four in the Appendix~\ref{appendix:decoding_dynamics}. 
   All other cases are based on similar arguments.
\end{proof}
\section{Bit-flipping decoders based on TS dynamics}
\label{sec:TS_aware_BF_decoder}
The analysis of the dynamics of the TSs, in the previous section, reveals that the number of unsatisfied checks associated with CC-type and VV-type variable nodes exceeds their corresponding flipping thresholds in the case of TS-induced error patterns. This phenomenon can be attributed to the inherent symmetry between VV-type and CC-type variable nodes within the subgraph induced by the stabilizers.
To mitigate the occurrence of these TSs, we propose an enhancement to the bit-flipping decoder in Algorithm~\ref{alg:TS_aware_bit_flipping_decoder}. Initially, the decoding algorithm flips the VV-type variable nodes depending on the number of unsatisfied checks to which they are linked. Subsequently, it revises the status of the check nodes and, consequently, flips the CC-type variable nodes. The decoder iteratively alternates between these two steps until convergence is achieved or a predetermined iteration limit is attained.
In Theorem~\ref{theorem:TS_aware_bf_decoder}, we establish results pertaining to the error correction capacity of the decoder described in Algorithm~\ref{alg:TS_aware_bit_flipping_decoder}.
\SetKwComment{Comment}{/* }{ */}
\begin{algorithm}
    \caption{Trapping set-aware bit-flipping decoder}\label{alg:TS_aware_bit_flipping_decoder}
    \KwData{$\boldsymbol{\sigma},\mathbf{H}_{\mathrm{Z}} $}
\KwResult{$\hat{\mathbf{x}}$}
$\hat{\mathbf{x}} \gets \boldsymbol{0}$\ \Comment*[r]{0 denotes the all-zero vector}
\While{$\mathbf{H}_{\mathrm{Z}} \hat{\mathbf{x}} \neq \boldsymbol{\sigma} (\text{ mod 2 })$}{
$\hat{\boldsymbol{\sigma}} \gets \mathbf{H}_{\mathrm{Z}} \hat{\mathbf{x}} ( \text{mod 2}) $\;
$\boldsymbol{\beta} \gets \hat{\boldsymbol{\sigma}} \oplus \boldsymbol{\sigma}$\;
$\boldsymbol{\alpha} = \mathbf{H}_{\mathrm{Z}}^{\mathsf{T}}\boldsymbol{\beta}$\;
\For{$i=1$ to $n_1n_2$}{\ \Comment*[r]{$n_1n_2:$ the number of VV-type variable nodes}
\If{$\boldsymbol{\alpha}_i > |\mathcal{N}_{v_i}|/2 $}{
$\hat{\mathbf{x}}_i=\hat{\mathbf{x}} \oplus 1 ( \text { mod 2})$\;
}
}
$\hat{\boldsymbol{\sigma}} \gets \mathbf{H}_{\mathrm{Z}} \hat{\mathbf{x}} ( \text{mod 2}) $\;
$\boldsymbol{\beta} \gets \hat{\boldsymbol{\sigma}} \oplus \boldsymbol{\sigma}$\;
$\boldsymbol{\alpha} = \mathbf{H}_{\mathrm{Z}}^{\mathsf{T}}\boldsymbol{\beta}$\;
\If{$\mathbf{H}_{\mathrm{Z}} \hat{\mathbf{x}} \neq \boldsymbol{\sigma} (\text{ mod 2 })$}{
\For{$i=n_1n_2+1$ to $n$}{
\If{$\boldsymbol{\alpha}_i > |\mathcal{N}_{v_i}|/2 $}{
$\hat{\mathbf{x}}_i=\hat{\mathbf{x}} \oplus 1 ( \text { mod 2})$\;
}
}
}
}
\end{algorithm}
\begin{theorem}
   \label{theorem:TS_aware_bf_decoder} 
    Let $\cup_{j \in \mathcal{J}}\mathscr{T}(c_j^2)$ denote a connected cycle-free subgraph of $\mathscr{G}_2$ induced by $\{c_j^2: j \in \mathcal{J}  \}$, and similarly let $\cup_{i \in \mathcal{I}}\mathscr{T}(v_i^2)$ denote a connected cycle-free subgraph of $\mathscr{G}_1$ induced by $\{v_i^1: i \in \mathcal{I}  \}$, where the index sets $\mathcal{I}$ and $\mathcal{J}$ are non-empty.
   Consider the subgraph $\mathscr{T}(\mathbf{h}) \coloneqq \cup_{i \in \mathcal{I},j \in \mathcal{J}}\mathscr{T}(c_j^2v_i^1)=\cup_{j \in \mathcal{J}}\mathscr{T}(c_j^2) \times \cup_{i \in \mathcal{I}}\mathscr{T}(v_i^2) $.
   The TS-aware decoder, as described in Algorithm~\ref{alg:TS_aware_bit_flipping_decoder}, successfully corrects all error patterns within subgraph $\mathscr{T}(\mathbf{h})$, given that both $d_v^1$ and $d_c^2$ are odd and greater than two.
\end{theorem}
\begin{proof}
The proof is given in Appendix~\ref{appendinx_proof_TS_aware_decoder}.
\end{proof}
\begin{figure*}
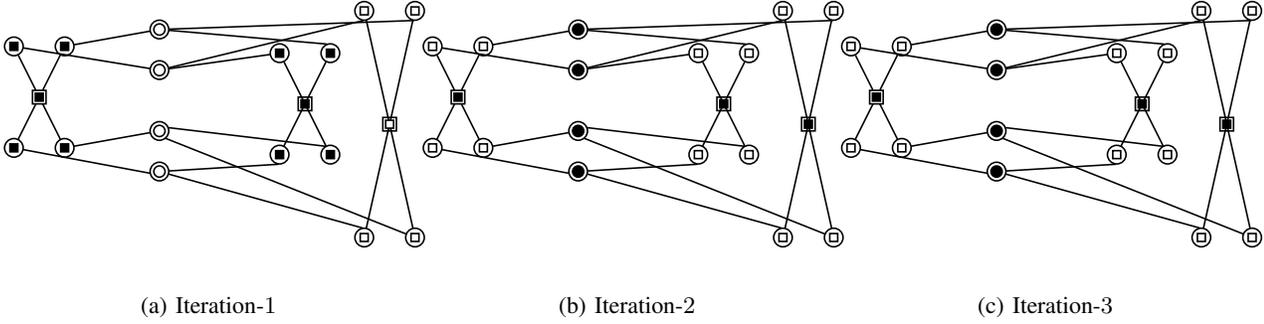

    \begin{subfigure}[t]{0.3\textwidth}
    \input{Figures/Journal_figures/Example_serial_BF_decoder_TS/iteration_0}
        \subcaption{Iteration-$1$} 
        \label{fig:iteration0_TS}
    \end{subfigure}
    \begin{subfigure}[t]{0.3\textwidth}
    \input{Figures/Journal_figures/Example_serial_BF_decoder_TS/iteration_1}
        \subcaption{Iteration-$2$}
        \label{fig:iteration1_TS}
    \end{subfigure}
    \begin{subfigure}[t]{0.3\textwidth}
    \input{Figures/Journal_figures/Example_serial_BF_decoder_TS/iteration_2}
        \subcaption{Iteration-$3$}
        \label{fig:iteration2_TS}
    \end{subfigure}
    \caption{The figure illustrates how the TS-aware BF decoder described in Algorithm~\ref{alg:TS_aware_bit_flipping_decoder} avoids TS induced by  a stabilizer generator under bit-flipping decoding, in which VV-type and CC-type variable nodes have degrees three and four, respectively. The inner shape of the erroneous variable nodes and unsatisfied checks is depicted in black.} 
    \label{fig:example_serial_bf_decoder_TS}
\end{figure*}

   \begin{figure*}
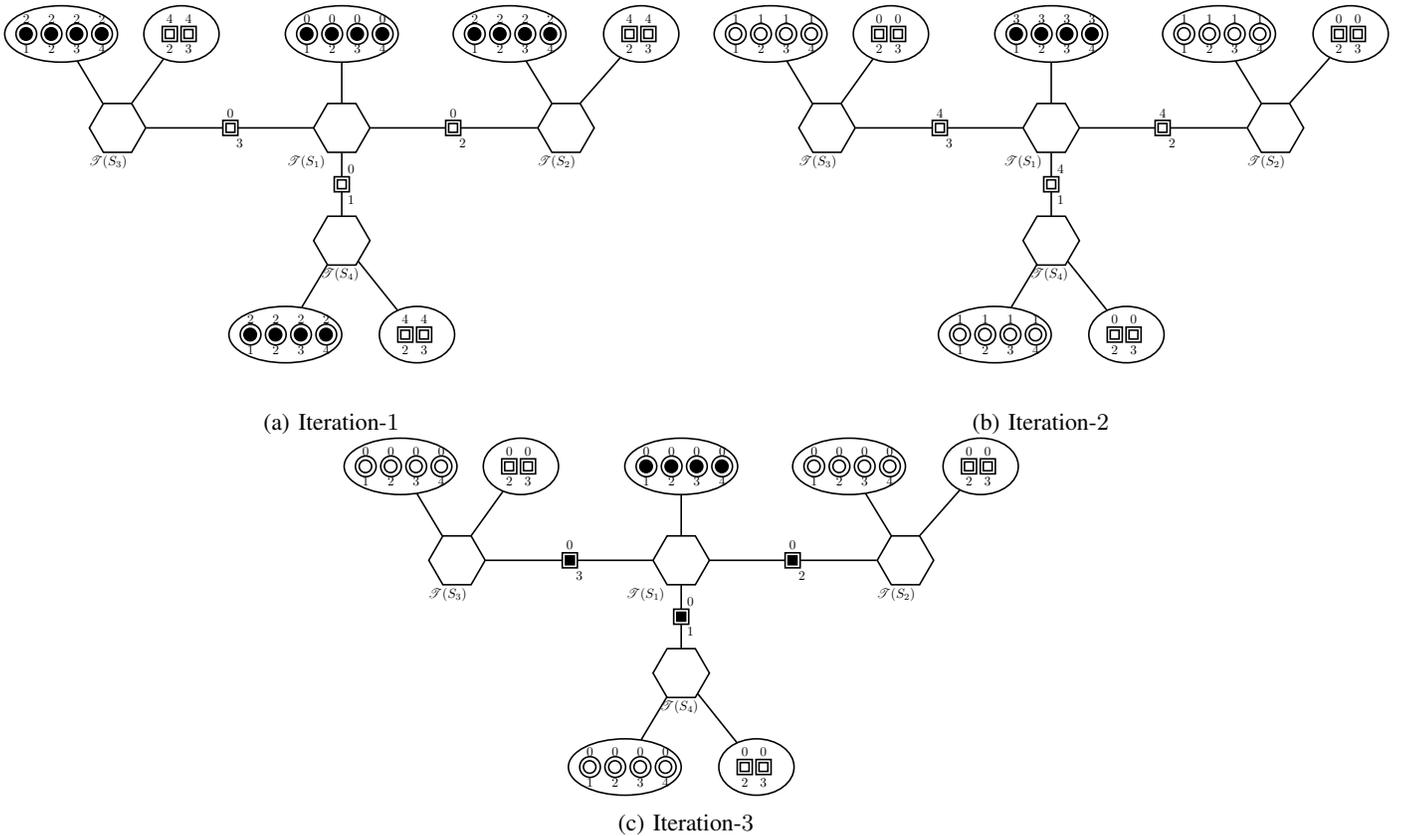

       \begin{subfigure}[t]{0.48\textwidth}
       \centering
       \input{Figures/Journal_figures/Example_serial_BF_decoders/iteration_1}
       \subcaption{Iteration-$1$}
        \label{fig:iteration1}
       \end{subfigure}
       \begin{subfigure}[t]{0.48\textwidth}
       \centering
       \input{Figures/Journal_figures/Example_serial_BF_decoders/iteration_2}
       \subcaption{Iteration-$2$}
       \label{fig:iteration2}
       \end{subfigure}
       \begin{subfigure}[t]{\textwidth}
       \centering
       \input{Figures/Journal_figures/Example_serial_BF_decoders/iteration_3}
       
       \subcaption{Iteration-$3$} 
       \label{fig:iteration3}
       \end{subfigure}
      
       \caption{The figure illustrates how the TS-aware BF decoder described in Algorithm~\ref{alg:TS_aware_bit_flipping_decoder} avoids TS induced by a linear combination of stabilizer generators, in which VV-type and CC-type variable nodes have degrees three and four, respectively. The inner shape of the erroneous variable nodes and unsatisfied checks is depicted in black. The numbers above and below each variable node denote the number assigned to it according to the procedure described in Lemma~\ref{lemma:vn_only_graphical_description} and the number of unsatisfied checks connected to it, respectively.}
        \label{fig:example_serial_bf_decoder_TS_comp}
   \end{figure*} 
   Next, we consider two examples of TSs under bit-flipping decoding and show that the TS-aware decoder, described in Algorithm, avoids the considered TSs.
   \begin{example}
       \label{ex:illustration_TS_aware_dec_avoids_TS_induced_by_generators}
       Consider the subgraph shown in Figure~\ref{fig:example_serial_bf_decoder_TS}. Note that this subgraph is an TS since it is induced by a stabilizer generator. Also, observe that two CC-type variable nodes are in error. From Part~\ref{lemma_TS_composition_bf_decoder:part2} of Lemma~\ref{lemma:TSs_induced_by_stabilizers}, it can be deduced that this is an TS-inducing error pattern. 
       Now, consider the decoding of this error pattern using the TS-aware decoder described in Algorithm~\ref{alg:TS_aware_bit_flipping_decoder}.  The decoder starts by flipping VV-type variable nodes that are connected to more than $\left\lfloor\frac{d_v^1}{2}\right\rfloor=1$ unsatisfied checks. Since all VV-type variable nodes are connected to two unsatisfied checks, the decoder flips the error estimate on all VV-type variable nodes. In iteration~2, the decoder flips the CC-type variable nodes that are connected to more than $\left\lfloor\frac{d_c^2}{2}\right\rfloor=2$ unsatisfied checks.
       At the beginning of the third iteration, none of the variable nodes is connected to any unsatisfied checks, thus indicating the convergence of the decoder. Note that the decoder converges to a stabilizer generator, taking advantage of the degeneracy of HP codes.
   \end{example}
\begin{example}
    \label{ex:serial_bf_decoder_TS_comp} 
    Consider the subgraph induced by the linear combination of stabilizer generator $S_1,S_2,S_3$, and $S_4$ shown in Figure~\ref{fig:example_serial_bf_decoder_TS_comp}.
    From Theorem~\ref{theorem:stabilizer_induced_TS}, we know that this subgraph is a TS under bit-flipping decoding.
    Now, consider the decoding of the error pattern within the subgraph shown in Figure~\ref{fig:example_serial_bf_decoder_TS_comp} using the TS-aware decoder described in Algorithm~\ref{alg:TS_aware_bit_flipping_decoder}.
    Algorithm~\ref{alg:TS_aware_bit_flipping_decoder}  flips only the error estimate on the VV-type variable nodes in the first iteration. As a result, the error estimates on the CC-type variable nodes remain unchanged regardless of the number of unsatisfied checks to which they are connected. 
    The VV-type variable nodes associated with stabilizers $S_1, S_2$ and $S_3$ are each connected to two unsatisfied checks, which exceed the threshold specified by $\left\lfloor \frac{d_v^1}{2} \right \rfloor=\left\lfloor\frac{3}{2}\right\rfloor=1$. Therefore, the VV-type variable nodes corresponding to the stabilizers $S_1,S_2$ and $S_3$ undergo flipping. In contrast, those corresponding to the stabilizer $S_1$ are connected to a single unsatisfied check and, as a result, are not flipped.
    In iteration~2, following the update of the error estimates on the VV-type variable nodes in the previous iteration, the decoder continues to update the error estimates on the CC-type variable nodes. The CC-type variable nodes associated with stabilizer $S_1$ are connected to four unsatisfied checks, thus exceeding the threshold given by $\left\lfloor\frac{d_c^2}{2} \right\rfloor=\left\lfloor\frac{4}{2} \right\rfloor=2$. Consequently, the CC-type variable nodes, corresponding to stabilizer $S_1$, are subjected to flipping, while the CC-type variable nodes without any unsatisfied checks as neighbors remain unchanged.
    At the beginning of the third iteration, the variable nodes are not connected to any unsatisfied checks, indicating convergence of the decoder. Observe that the decoder converges to a stabilizer, implying that the decoder exploits the inherent degeneracy of HP codes to converge.
\end{example}
In Theorem~\ref{theorem:TS_aware_bf_decoder}, it is shown that the decoder introduced in Algorithm~\ref{alg:TS_aware_bit_flipping_decoder} successfully circumvents stabilizer-induced TSs under the condition that the variable nodes possess an odd degree. In the following example, we elucidate the circumstances, in particular when the variable nodes exhibit even degrees, under which the proposed decoders encounter failures and propose a straightforward modification to prevent such failures. This modification is intentionally excluded from Theorem~\ref{theorem:TS_aware_bf_decoder} to avoid an additional number of cases in the proof.
\begin{example}
    \label{ex:TS_of_for_eve_degree}
    Consider the subgraph depicted in Figure~\ref{fig:TS_serial_bf_decoder}, which is induced by a $X$-stabilizer generator. The degrees of the VV-type and CC-type variable nodes are four and six, respectively. 
    In the subgraph illustrated in Figure~\ref{fig:TS_serial_bf_decoder}, three out of six VV-type variable nodes are erroneous, while two out of four CC-type variable nodes are erroneous.
    Now, consider the decoding of this error pattern using the TS-aware decoder described in Algorithm~\ref{alg:TS_aware_bit_flipping_decoder}.
    It is important to note that the VV-type variable nodes are connected to two unsatisfied checks. 
    The threshold for the number of unsatisfied neighboring nodes required to flip errors on VV-type variables is $\left\lfloor \frac{d_v^1}{2} \right\rfloor=2$. 
    Consequently, none of the VV-type variable nodes have changed their error estimates. 
    Similarly, CC-type variable nodes are connected to three unsatisfied checks, which do not exceed the threshold $\left\lfloor \frac{d_c^1}{2}\right\rfloor=3$ necessary to change their error estimates. As a result, the subgraph shown in Figure~\ref{fig:TS_serial_bf_decoder} constitutes a TS for the decoder described in Algorithm~\ref{alg:TS_aware_bit_flipping_decoder}.

  Consider now an adaptation of the decoder in Algorithm~\ref{alg:TS_aware_bit_flipping_decoder}, wherein a VV-type (or CC-type) variable node is selected at random and flipped among those VV-type (or CC-type) variable nodes that are connected to precisely $\left\lfloor \frac{d_v^1}{2} \right\rfloor$ $\left(\text{or }\left\lfloor \frac{d_c^2}{2} \right\rfloor\right)$ unsatisfied checks, provided that none of them is associated with more than $\left\lfloor \frac{d_v^1}{2} \right\rfloor$  $\left(\text{or }\left\lfloor \frac{d_c^2}{2} \right\rfloor\right)$ unsatisfied checks. 
  With this modification, the decoder flips a VV-type variable node in the first iteration and converges to the correct error estimate in subsequent iterations following Algorithm~\ref{alg:bit_flipping_decoder} as illustrated in Figure~\ref{fig:TS_serial_bf_decoder}.
\end{example}
\begin{figure*}
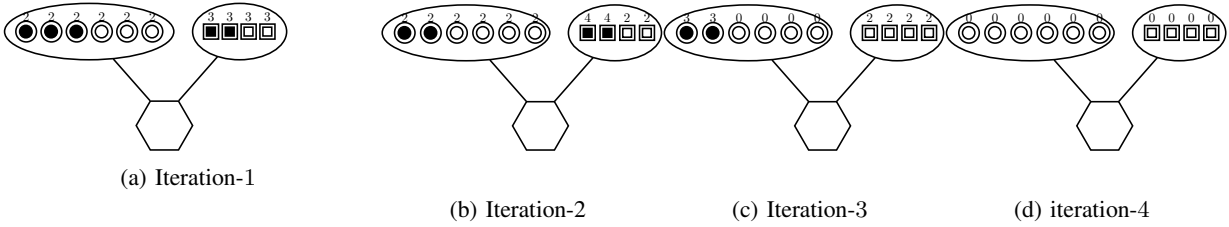

\begin{subfigure}[t]{0.27\textwidth}
\input{Figures/Journal_figures/TS_serial_BF_decoder/iteration_1}
\subcaption{Iteration-$1$}
\label{fig:TS_serial_bf_decoder_iteration1}
\end{subfigure}
\begin{subfigure}[t]{0.2\textwidth}
\input{Figures/Journal_figures/TS_serial_BF_decoder/iteration_2}
\subcaption{Iteration-2}
\label{fig:TS_serial_bf_decoder_iteration2}
\end{subfigure}
\begin{subfigure}[t]{0.2\textwidth}
\input{Figures/Journal_figures/TS_serial_BF_decoder/iteration_3}
\subcaption{Iteration-$3$}
\label{fig:TS_serial_bf_decoder_iteratio3}
\end{subfigure}
\begin{subfigure}[t]{0.2\textwidth}
\input{Figures/Journal_figures/TS_serial_BF_decoder/iteration_4}
\subcaption{iteration-$4$}
\label{fig:TS_serial_bf_decoder_iteration4}
\end{subfigure}
    \caption{The figure shows a TS of the decoder described in Algorithm~\ref{alg:TS_aware_bit_flipping_decoder} when VV-type and CC-type variable nodes have an even degree and illustrates how a modification of the decoder avoids the issue. The modification of updating the error estimate on a randomly selected VV-type or CC-type variable node, among those precisely linked to the same number of unsatisfied checks as the threshold determined by Algorithm~\ref{alg:TS_aware_bit_flipping_decoder}, serves as a simple modification that prevents such TS when the count of unsatisfied checks connected to any variable node does not exceed the threshold. The degrees of VV-type and CC-type variable nodes are four and six, respectively.  The inner shape of the erroneous variable nodes and unsatisfied checks is depicted in black. The number above each variable node denotes the number of unsatisfied checks connected to it.}
    \label{fig:TS_serial_bf_decoder}
\end{figure*}

\

 
\section{Stabilizer-induced TSs of LP codes}
\label{sec:QTS_LP_codes}
As described in Section~\ref{sec:LP_codes}, the construction of base graphs for LP codes from the base graphs of two LDPC codes is analogous to the construction of Tanner graphs for HP codes from the Tanner graphs of two classical codes. 
Consequently, the subgraphs induced by the $X$-stabilizer generators in the base graph pertaining to the $Z$-stabilizer generators exhibit a structural similarity to the stabilizer-induced subgraphs in the context of HP codes. 
Thus, when the decoder operates on the base graphs of the LP codes, the LP codes will possess stabilizer-induced TSs or QTSs similar to those found in the HP codes.
However, the decoder operates on the Tanner graphs of the LP codes, which are obtained by lifting the corresponding base matrices.
In order to validate the assertion that LP codes exhibit stabilizer-induced TSs, it is essential to demonstrate the non-trivial property that the lifting process retains the structure of TSs identified in base graphs within their corresponding Tanner graphs. 
This property is formalized in Lemma~\ref{lemma:stabilizer_induced_TSs_LP_codes}.
First, we prove a property of the protograph LDPC codes that will be subsequently used in the proof of the Lemma~\ref{lemma:stabilizer_induced_TSs_LP_codes}.
\begin{lemma}
    \label{lemma:intermediate_result_copy_and_permute}
   Consider a protograph LDPC code with base graph $\underline{\mathscr{G}}$. This base graph is lifted into a Tanner graph $\mathscr{G}$ of an LDPC code by employing the copy and permutation procedure as detailed in Section~\ref{sec:LP_codes}. Let the lifting size in the copy and permute operation be given by $\gamma$.
   If there is an edge between the variable node $\underline{v}$ and the check node $\underline{c}$ with the label $x^\rho$ in base graph $\underline{\mathscr{G}}$, the following applies: \begin{itemize}
       \item The $i$-th copy of $\underline{v}$ is connected to the $i+\rho-1 (\text{ mod } \gamma) +1$-th copy of $\underline{c}$ in Tanner graph $\mathscr{G}$;
       \item The $j$-th copy of $\underline{c}$ is connected to the  $j-\rho -1 (\text{ mod }\gamma)+1$-th copy of $\underline{v}$ in Tanner graph $\mathscr{G}$,
   \end{itemize}
\end{lemma}
\begin{proof}
    The proof follows from the copy and permute operation described in Section~\ref{sec:LP_codes}.
\end{proof}
\begin{figure*}
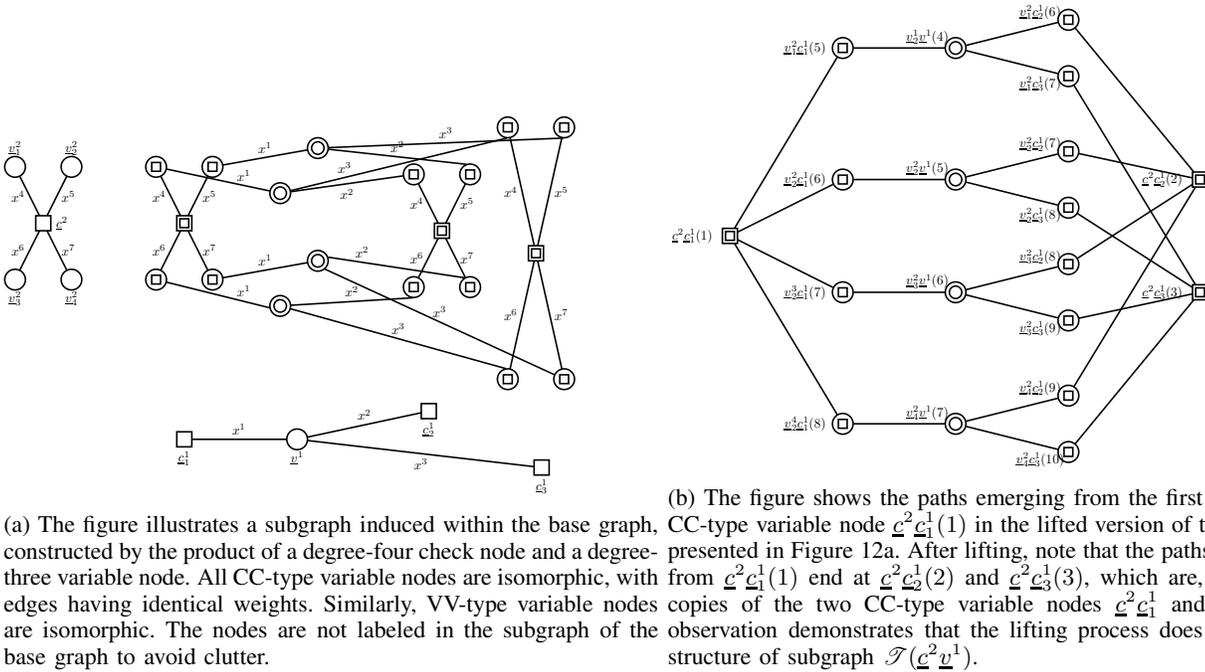

\begin{subfigure}{0.48\textwidth}
     \input{Figures/Journal_figures/TS_LP_codes/LP_TS}
    \subcaption{ The figure illustrates a subgraph induced within the base graph, constructed by the product of a degree-four check node and a degree-three variable node. All CC-type variable nodes are isomorphic, with edges having identical weights. Similarly, VV-type variable nodes are isomorphic. The nodes are not labeled in the subgraph of the base graph to avoid clutter.}
    \label{fig:LP_TS}
\end{subfigure}
\begin{subfigure}{0.48\textwidth}
     \input{Figures/Journal_figures/TS_LP_codes/LP_TS_copy_permute}
    \subcaption{The figure shows the paths emerging from the first copy of the CC-type variable node $\underline{c}^2\underline{c}^1_1(1)$ in the lifted version of the subgraph presented in Figure~\ref{fig:LP_TS}. After lifting, note that the paths originating from $\underline{c}^2\underline{c}^1_1(1)$ end at $\underline{c}^2\underline{c}^1_2(2)$ and $\underline{c}^2\underline{c}^1_3(3)$, which are, in fact, the copies of the two CC-type variable nodes $\underline{c}^2\underline{c}^1_1$ and $\underline{c}^2\underline{c}^1_3$. This observation demonstrates that the lifting process does not modify structure of subgraph $\mathscr{T}(\underline{c}^2\underline{v}^1)$. }
    \label{fig:lifting_LP_Ts}
\end{subfigure}
\caption {This figure illustrates that the lifting process does not alter the structures of the stabilizer-induced TSs.}
\label{fig:LP_TS_pictorial_proof}
\end{figure*}
\begin{lemma}
    \label{lemma:stabilizer_induced_TSs_LP_codes}
    Consider an LP code characterized by two base graphs, denoted $\underline{\mathscr{G}}_{\mathrm{X}}$ and $\underline{\mathscr{G}}_{\mathrm{Z}}$, which are derived from the base graphs of two classical codes associated with the Tanner graphs $\underline{\mathscr{G}}_1$ and $\underline{\mathscr{G}}_2$. 
    Let $\underline{\mathscr{G}}_1$ and $\underline{\mathscr{G}}_{2}$ represent $(d_c^1,d_v^1)$ and $(d_c^2,d_v^2)$ regular graphs, respectively. Furthermore, let the Tanner graphs $\mathscr{G}_{\mathrm{X}}$ and $\mathscr{G}_{\mathrm{X}}$ be obtained by lifting from the base graphs $\underline{\mathscr{G}}_{\mathrm{X}}$ and $\underline{\mathscr{G}}_{\mathrm{Z}}$, with the lifting size specified as $\gamma$. Denote $\mathscr{T}(\underline{c}^2\underline{v}^1)$ as the subgraph induced by the $X-$ check $\underline{c}^2\underline{v}^1$ on the base graph $\underline{\mathscr{G}}_{\mathrm{Z}}$, where $\underline{c}^2$ is a check node in $\mathscr{G}_2$ and $\underline{v}^1$ is a variable node in $\mathscr{G}_2$. Then, Tanner graph $\mathscr{G}_{\mathrm{Z}}$ contains $\gamma$ isomorphic copies of the subgraph $\mathscr{T}(\underline{c}^2\underline{v}^1)$.
\end{lemma}
\begin{proof}
  The proof is given in Appendix~\ref{sec:appendix_proof_stabilizer_induced_TSs_LP_codes}.
\end{proof}
The proof is shown graphically in Figure~\ref{fig:LP_TS_pictorial_proof} for a base graph which has degree-three VV-type variable nodes and degree-four CC-type variable nodes.
Lemma~\ref{lemma:stabilizer_induced_TSs_LP_codes} demonstrates that the lifting procedure preserves the structures inherent to the stabilizer-induced TSs. Consequently, the findings formulated in Theorem~\ref{theorem:stabilizer_induced_TS} and Theorem~\ref{theorem:TS_aware_bf_decoder} regarding HP codes apply to LP codes.

\section{Classical TSs of HP and LP codes}
\label{sec:classical_TS}

A graph derived from the product of two graphs contains multiple isomorphic replicas of each of the constituent graphs involved in the product. 
This implies that in the resultant HP or LP codes, each of the TSs inherent to the constituent classical codes is replicated multiple times. 
In the following lemmas, we formalize these results for both HP and LP codes.
Rather than providing formal proofs of these lemmas, we represent them visually in Figure~\ref{fig:classical_TS}, as they emerge inherently from the construction methodology of HP and LP codes.

\begin{lemma}
    \label{lemma:classical_TS_HP_codes}
    Let subgraphs $\mathscr{T}_1$ and $\mathscr{T}_2$, respectively, represent two TSs in Tanner graphs ${\mathscr{G}}_1$ and ${\mathscr{G}}_2$ with $n_1$ and $n_2$ variable nodes, corresponding to two classical LDPC codes.
    Let $\mathscr{G}_{\mathrm{X}}$ and $\mathscr{G}_{\mathrm{Z}}$ represent the Tanner graphs of an HP code, formed by the graph product of Tanner graphs $\mathscr{G}_1$ and $\mathscr{G}_2$.
    Then, Tanner graph $\mathscr{G}_{\mathrm{X}}$ has $n_2$ replicas of TS $\mathscr{T}_1$ and Tanner graph $\mathscr{G}_{\mathrm{Z}}$ has $n_1$ replicas of TS $\mathscr{T}_2$.
\end{lemma}

\begin{lemma}
    \label{lemma:classical_TS_LP_codes}
    Consider subgraphs $\mathscr{T}_1$ and $\mathscr{T}_2$, respectively, representing two TSs within codes corresponding to base graphs $\underline{\mathscr{G}}_1$ and $\underline{\mathscr{G}}_2$ with $n_1$ and $n_2$ variable nodes, corresponding to two protograph LDPC codes.
    The Tanner graphs $\mathscr{G}_1$ and $\mathscr{G}_2$ represent the graphs derived from the base graphs $\underline{\mathscr{G}}_1$ and $\underline{\mathscr{G}}_2$, respectively, with lifting parameter $\gamma$.
    Let $\mathscr{G}_{\mathrm{X}}$ and $\mathscr{G}_{\mathrm{Z}}$ represent the Tanner graphs of an LP code, formed by lifting the graph product of base graphs $\underline{\mathscr{G}}_1$ and $\underline{\mathscr{G}}_2$.
    Then Tanner graph $\mathscr{G}_{\mathrm{X}}$ has $\gamma n_2$ replicas of TS $\mathscr{T}_1$ and Tanner graph $\mathscr{G}_{\mathrm{Z}}$ has $\gamma n_1$ replicas of TS $\mathscr{T}_2$.
\end{lemma}
\begin{figure}[h]
\centering
    \input{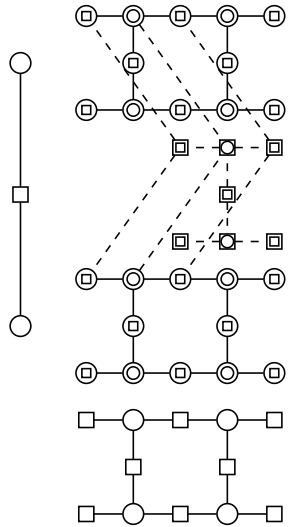}
    \caption{ The figure illustrates that a classical (4,4)-TS in the constituent codes of an HP code replicates in the HP code due to the graph product. The dotted lines indicate the edges that are not part of the TS. The Tanner graph is partially shown to focus on classical TSs. In particular, the TS in Code~$1$ appears twice, as it is multiplied by a length-two path from Code~$2$.}
    \label{fig:classical_TS}
    \end{figure}

\section{Decoder Diversity Approach for Avoiding TS}
\label{sec:min_sum_decoder}
In Section~\ref{sec:QTS_hypergraph}, we introduced an improved bit-flipping decoder to address QTSs. 
To minimize logical error rates or prevent error floors, it is essential to use decoders that can also handle TSs from classical codes alongside QTSs.
In this section, we take the diversity approach, introduced in  \cite{FAID_diversity}, to design a set of modified min-sum decoders that can circumvent both QTSs and TSs inherited from the classical LDPC codes when run in parallel.
The min-sum decoder is considered since it is capable of correcting more erroneous qubits compared to the bit-flipping decoder, even though the bit-flipping decoder proposed in Section~\ref{sec:QTS_hypergraph} can overcome QTSs.
We start by describing the message updating rules for the min-sum decoder, followed by the discussion about modification of update rules to overcome both QTSs and classical TSs.
\subsection{Min-sum decoder}
Consider the min-sum decoder over a depolarizing channel with depolarizing probability $p$ run on Tanner graph $G_{\mathrm{Z}}(\mathcal{V}\cup \mathcal{C},\mathcal{E})$, which is constructed from two Tanner graphs, say $\mathscr{G}_1$ and $\mathscr{G}_2$, corresponding to two classical LDPC codes.
Given an error vector $\mathbf{x}$, the syndrome vector $\boldsymbol{\sigma}=(\sigma_1,\sigma_2,...,\sigma_{|\mathcal{C}|})$ can be obtained by  
$ \boldsymbol{\sigma} = \mathbf{x}\mathbf{H}_{\mathrm{Z}}^{\mathsf{T}}.$
We adapt the message passing rules to enable distinct rules for each edge in the Tanner graph $G_{\mathrm{Z}}$, utilizing the unique channel log-likelihood ratios (\emph{llrs}) for the different edges.
Let $m_{v \rightarrow c}^t(e)$ be the message from variable node
to check node along edge $e$ for $e \in \mathcal{E}$ in the $t$-th iteration. Similarly, let
$m_{c \rightarrow v}^t(e)$ be the message from check node to
variable node in the $t$-th iteration. The
message-passing recursion is given by
\begin{align}
\label{eq:mp_cv}m^{t+1}_{c \rightarrow v}(e)&=\sigma_{c_e} \prod_{e'\in E_c(e)}m_{v \rightarrow c}^{t}(e') \min_{e'\in E_c(e)}|m_{v \rightarrow c}(e')|,\\
\label{eq:mp_vc}m_{v \rightarrow c}^{t+1}(e)&=b(e) + w \left(\sum_{e'\in E_v(e)}m_{c \rightarrow v}^{t+1}(j)\right),
\end{align}
for $t\ge0$, where
$E_c(e)=\{e':c(e)=c(e'),e\ne e'\}$ and $E_v(e)=\{e':v(e)=v(e'),e\ne e'\}$ are the set of neighboring edges incident to the same check node and variable node, respectively, as
the edge $e,$ $b(e)$, for $e \in \mathcal{E}$, is the bias corresponding to edge $e$, and $w$ is the normalization constant. 

\subsection{Scheduled min-sum decoder to avoid stabilizer-induced TSs}
\label{sec:modified_min_sum_to_avoid_QTS}
As discussed at the beginning of this section, in this approach, we design a set of decoders, denoted by $\mathcal{D}$, to minimize the probability of logical error rate when the designed decoders are run in parallel.
We develop a variant of the min-sum decoders by introducing changes that equip the new decoder to sidestep stabilizer-induced TSs while preserving the strong error correction capability of the original min-sum decoder.
Recall that in Section~\ref{sec:QTS_hypergraph}, the bit-flipping decoder avoids stabilizer-induced TSs by not updating the error estimate on the VV-type and CC-type variable nodes simultaneously.
The modified min-sum decoder, denoted by $D_1$, does not update the messages from the variable nodes to the check nodes along the edges connected to the VV-type and CC-type nodes within the same iteration. 
In other words, decoder $D_1$ begins by updating $m_{v \rightarrow c}(e)$ if $e$ is linked to a VV-type variable node. In the following iteration, it updates $m_{v \rightarrow c}(e)$ if $e$ is connected to a CC-type variable node. These two steps are alternated in successive iterations until the decoder either converges or the predetermined maximum number of iterations is reached.
The bias $b(e)$ is set to $\log \frac{1-p}{p}$ for all edges in set $\mathcal{E}$.
\subsection{Min-sum decoders to avoid classical TS}
\label{sec:dec_to_avoid_classical_TS}
The other decoders in set $\mathcal{D}$ are designed to prevent TS that are inherited from classical LDPC codes. Like decoder $D_1$, these decoders in $\mathcal{D}$ first update the outbound messages from VV-type variable nodes in the first iteration, followed by updating the messages from CC-type variable nodes in the next iteration. 
In contrast to decoder $D_1$, for decoders in $\mathcal{D}\setminus D_1$, the biases associated with the edges are carefully chosen to avoid classical TSs.
To see how these biases are chosen, consider a TS, denoted by $\mathscr{T}_2$, in Tanner graph $\mathscr{G}_2(\mathcal{V}_2\cup \mathcal{C}_2,\mathcal{E}_2)$ of a classical code.
Let $\mathcal{E}(\mathscr{T}_2) \subset \mathcal{E}_2$ denote the set of edges in subgraph $\mathscr{T}_2$ of Tanner graph $\mathscr{G}_2$.
In \cite{henery_neural}, it is shown that the TSs of a classical code can be avoided by carefully choosing the biases.
Assume that a min-sum decoder, say $D$, avoids the trapping set $\mathscr{T}_2$ when run on Tanner graph $\mathscr{G}_2$.
In decoder $D$, let $b(e)$, for $e \in \mathcal{E}_2$ denote the bias corresponding to edge $e$.
Recall from Section~\ref{sec:classical_TS} that Tanner graph $\mathscr{G}_\mathrm{Z}$ has $n_1$ isomorphic copies of TS $\mathscr{T}^2$, where $n_1$ is the number of variable nodes in $\mathscr{G}_1$.
Next, given decoder $D$, we present an approach to design a decoder that runs on Tanner graph $\mathscr{G}_\mathrm{Z}$ and avoids all $n_1$ isomorphic copies of TS $\mathscr{T}_2$ in it.
For this purpose, divide the set of edges, denoted by $\mathcal{E}$, of Tanner graph $\mathscr{G}_\mathrm{Z}$ into two groups. 
The first group contains the isomorphic copies of $\mathcal{E}_2$, and is given by $\cup_{i=1}^{n_1}\mathcal{E}_2(i)$, where $\mathcal{E}_2(i)$ denotes the set of edges associated with the $i$-th copy of $\mathscr{G}_2$ in $\mathscr{G}_{\mathrm{Z}}$. 
The second group contains the rest of the edges and is given by $\mathcal{E}\setminus \left(\cup_{i=1}^{n_2}\mathcal{E}_2(i)\right)$.
For each $e\in \mathcal{E}_2$, let $e(i)$ represent its isomorphic counterpart in $\mathcal{E}_2(i)$. Select $b(e(i))=b(e)$ for each $e \in \mathcal{E}_2$ to ensure that the isomorphic copies of edge $e$ in Tanner graph $\mathscr{G}_{\mathrm{Z}}$ maintain the same bias as edge $e$ in Tanner graph $\mathscr{G}_2$. 
The bias corresponding to the edges in the second group is set to $\log \frac{1-p}{p}$.
Given any TS of the constituent classical LDPC code and a decoder to avoid it, the above procedure can be repeated to design a decoder that avoids all the isomorphic copies of the TS in the corresponding HP (or LP) code.
In this work, we design a separate decoder for each of the TSs in the constituent classical LDPC code. 
These decoders run in parallel until one of them converges. 
The error estimation is then derived from the first decoder to converge.

\section{Numerical results}
\label{sec:simulation_results}
In this part, we analyze the logical error rate of our proposed decoder and compare it to the normalized min-sum decoder across various LP codes. 
To evaluate the performance of the decoder, we examine the $[[1054,140,20]]$ LP code derived from the classical $(155,64,20)$ Tanner code. 
The classical Tanner code includes a number of $(5,3)$ trapping sets (TS), which are carried over to the LP codes. 
In addition to these inherited classical TSs, the LP code also possesses stabilizer-induced TSs. 
To avoid errors due to stabilizer-induced TS, we employ a modified minimum sum decoder as presented in Section~\ref{sec:modified_min_sum_to_avoid_QTS}, while to counteract the classical TS, we derive decoders from the Tanner code decoder, as suggested in Section~\ref{sec:dec_to_avoid_classical_TS}.
Together, these decoders enhance the logical error rate. 
Figure~\ref{fig:simulation_results} shows that the logical error rate of the proposed decoder becomes smaller than that of the normalized minimum sum decoder as the depolarizing rate decreases and becomes substantially smaller in the error floor region.
Five decoders are used to achieve this superior logical error rate; one of these overcomes stabilizer-induced TSs, whereas the other four overcome classical TSs.
Each of these decoders is run for twenty iterations. 
For a fair comparison, the normalized min-sum decoder is run for hundred iterations.
Note that all decoders can operate in parallel; as a result, the proposed approach does not increase the decoding time due to the use of multiple decoders. 
If multiple decoders converge, we select the error estimate with the smallest Hamming weight.

Figure~\ref{fig:simulation_results_LP_620} displays the logical rate of an LP code $[[600,40,20]]$, comparing the results from the normalized minimum sum decoder and the proposed decoder detailed in Section~\ref{sec:modified_min_sum_to_avoid_QTS}. 
Both decoders are run for twenty-five iterations. Notably, the logical error rate decreases by an order of magnitude when the new decoder is used, as opposed to the normalized min-sum decoder. This scenario differs from previous ones, since no specialized decoder is used to circumvent classical TSs since the underlying classical codes do not have TSs induced by low-weight error patterns (typically of weight less than five). However, we highlight that employing decoders capable of bypassing classical TSs can further diminish the logical error rate.

\begin{figure}
    \centering
    \input{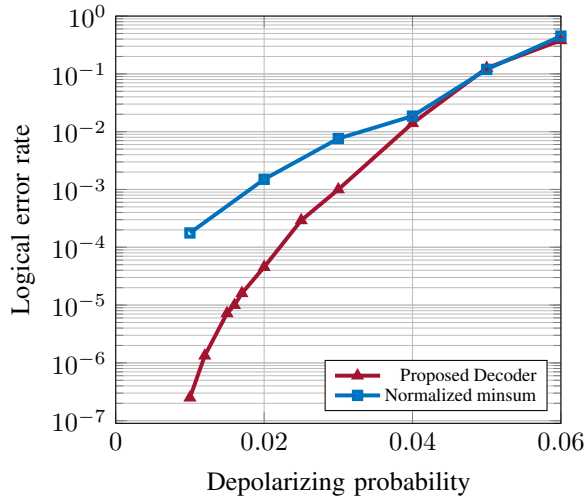}
    \caption{This figure plots the performance of the proposed decoder and normalized min-sum decoder for an $[[1054,140,20]]$ LP code. In the sequential approach, a decoder is used when the previous decoders fail. The proposed decoder outperforms normalized min-sum in the error floor region.}
    \label{fig:simulation_results}
\end{figure}

\begin{figure}
    \centering
    \input{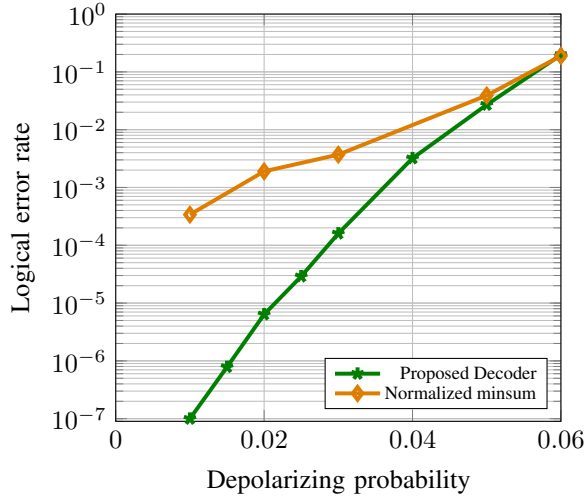}
    \caption{The figure shows the performance of both the proposed decoder and the normalized min-sum decoder for an $[[600,40,20]]$ LP code. In this instance, only one decoder is used to prevent any stabilize-induced TSs. }
    \label{fig:simulation_results_LP_620}
\end{figure}
\section{Conclusions}
\label{sec:conclusios}
We have performed a comprehensive characterization of classical and stabilizer-induced TSs associated with both HP and LP codes. Building upon this foundation, we propose a methodology for deriving decoders for HP and LP codes from the decoders of their constituent classical lLDPC codes. Within the scope of this paper, we consider the syndrome measurement circuit to be perfect. A potential avenue for future research is to investigate the extension of this methodology to scenarios involving circuit-level noise, where the syndrome measurement circuit is assumed to be noisy.
\section*{Acknowledgment}
B. Vasi\'{c} acknowledges the support of the NSF under grants CIF-2420424, CCF-1855879, CCF-2100013, CIF-2106189, CCSS-2027844,  CCSS-2052751, and a generous gift from our friends and modern Maecenases Dora and Barry Bursey. This work was also funded in part by Jet Propulsion Laboratory, California Institute of Technology, under a contract with the National Aeronautics and Space Administration and funded through JPL’s Strategic University Research Partnerships (SURP) program. B. Vasi\'{c} has disclosed an outside interest in Codelucida to the University of Arizona. Conflicts of interest resulting from this interest are being managed by The University of Arizona in accordance with its policies. 


\begin{appendices}

\section{Proof of Lemma~\ref{lemma:TSs_induced_by_stabilizers}}
\label{sec:appendix_proof_of_TSs_induced_by_stabilizers}
For all the lemmas in this section, we consider an HP code that is characterized by two Tanner graphs, $\mathscr{G}_{\mathrm{X}}$ and $\mathscr{G}_{\mathrm{Z}}$, derived from two classical codes, which themselves are associated with Tanner graphs $\mathscr{G}_1$ and $\mathscr{G}_2$. Assume that $\mathscr{G}_1$ and $\mathscr{G}_{2}$ be $(d_c^1,d_v^1)$ and  $(d_c^2,d_v^2)$  regular graphs, respectively.
\begin{lemma}
    \label{lemma:structures_of_graph_induced_by_stabilizers}
    Denote by $\mathscr{T}(c^2_jv^1_{i})$ the subgraph that is induced by $X-$check $c^2_jv^1_{i}$ on Tanner graph $\mathscr{G}_{\mathrm{Z}}$.  If neither $\mathscr{G}_1$ nor $\mathscr{G}_2$ contains cycles of length four, then the following statements about the induced subgraph hold.
    \begin{enumerate}
        \item \label{number_of_VNs_from_each_type} If the degree of the VV-type (or CC-type) nodes is $d_v^1$ (or $d_c^2$), then there are $d_v^1$ (or $d_c^2$) CC-type (or VV-type) variable nodes in $\mathscr{T}(c_j^2v_i^1).$ 
        \item \label{non_intersecting_sypport_vns}In $\mathscr{T}(c^2_jv^1_{i})$, there does not exist a pair of CC-type or VV-type variable nodes that shares a common check node as neighbor, i.e., for $v_{k}^2v_{i}^1, v_{k'}^2v_{i}^1 \in \mathcal{N}^{\mathrm{X}}_{c^2_jv^1_{i}}$ and  $c_j^2c_{k}^1, c_j^2c_{k'}^1 \in \mathcal{N}^{\mathrm{X}}_{c^2_jv^1_{i}},$ $\left(\mathcal{N}^{\mathrm{Z}}_{v_k^2v_{i}^1} \cap \mathcal{N}^{\mathrm{Z}}_{v_{k'}^2v_{i}^1} \right)= \emptyset$ and $\left(\mathcal{N}^{\mathrm{Z}}_{c_j^2c_{k}^1} \cap \mathcal{N}^{\mathrm{Z}}_{c_j^2c_{k'}^1} \right)= \emptyset.$
        \item \label{CN_connectivity}Each check node in $\mathscr{T}(c^2_jv^1_{i})$ is exactly connected to one VV-type node and one CC-type variable node.
        \item \label{outside_connectivity_constraint} There does not exist a variable node in $\mathscr{G}_\mathrm{Z}\setminus \mathscr{T}(c_j^2v_i^1)$ that is connected to two check nodes in $\mathscr{T}(c_j^2v_i^1)$. 
    \end{enumerate}
\end{lemma}
\begin{proof}
    \begin{enumerate}
        \item Consider the VV-type variable node $v_k^2v_i^1,$ where $v_k^2 \in \mathcal{N}^2_{c_j},$ in $\mathscr{T}(c_j^2v_i^1).$ 
       Since the degree of $v_k^2v_i^1$ is $d_v^1$ and $\mathcal{N}_{v_k^2v_i^1}^\mathrm{Z}=v^k_2 \times \mathcal{N}_{v_i^1}^1 ,$ $\left|\mathcal{N}_{v_i^1}^2\right|=d_v^1.$ 
       From \eqref{eq:hp_cx_neighbor}, recall that the set of CC-type variable nodes in $\mathscr{T}(c_j^2v_i^1)$ is $c_j^2 \times \mathcal{N}_{v_i^1}^1.$ 
       Since $\left|\mathcal{N}_{v_i^1}^1\right|=d_v^1$, it is concluded that there are $d_v^1$ CC-type variable nodes in $\mathscr{T}(c_j^2v_i^1).$ 

       \item Due to the analogous nature of the proofs involving two VV-type nodes and two CC-type check nodes, we shall restrict our proof to the case of VV-type nodes. Let us consider two VV-type variable nodes $v_{k}^2v_{i}^1$ and $ v_{k'}^2v_{i}^1$ within $\mathscr{T}(c^2_jv^1_{i}).$ The $Z$-type check nodes connected to both these variable nodes are given by
        \begin{equation*}
            \left(\mathcal{N}^{\mathrm{Z}}_{v_k^2v_{i}^1} \cap \mathcal{N}^{\mathrm{Z}}_{v_{k'}^2v_{i}^1} \right) = \left(v_k^2 \times \mathcal{N}^1_{v_i} \right) \cap \left(v_{k'}^2 \times \mathcal{N}^1_{v_i} \right).
        \end{equation*}
        On the right-hand side of the above equation, the cartesian products differ,  implying $\left(\mathcal{N}^{\mathrm{Z}}_{v_k^2v_{i}^1} \cap \mathcal{N}^{\mathrm{Z}}_{v_{k'}^2v_{i}^1} \right)=\emptyset.$ 
       \item According to Part~\ref{non_intersecting_sypport_vns} of Lemma~\ref{lemma:structures_of_graph_induced_by_stabilizers}, it is established that no pair of VV-type or CC-type variable nodes is connected to the same check node within $\mathscr{T}(c^2_jv^1_{i})$. Consequently, a check node in $\mathscr{T}(c^2_jv^1_{i})$ can be connected to a maximum of one VV-type and one CC-type variable node. Consider the selection of a specific check node, labeled $v_k^1c_k^2$, within $\mathscr{T}(c^2_jv^1_{i})$.
      It should be noted that $\mathscr{T}(c^2_jv^1_{i})$ represents the subgraph induced by the variable nodes that are connected to check node $c^2_jv^1_{i}$. Consequently, $v_k^2c_l^1$ is connected to a minimum of one variable node within $\mathscr{T}(c^2_jv^1_{i})$.
        We consider the case wherein $v_k^2c_l^2$ is associated with a VV-type node, given that the proof is analogous for the alternative case.
         Since $\mathscr{T}(c^2_jv^1_{i})$ is induced by variable nodes in $\mathcal{N}^{\mathrm{X}}_{c^2_jv^1_i},$ the set of variable nodes in $\mathscr{T}(c^2_jv^1_{i})$ is given by
        \begin{equation}
        \label{eq:lemma1_check_neighbors}
        \mathcal{N}^{\mathrm{X}}_{c^2_jv^1_i}=\left(\mathcal{N}^2_{c^2_j} \times v^1_{i}\right) \cup \left( c^2_j \times \mathcal{N}^1_{v_i^1}\right).
        \end{equation}
        From equation~\eqref{eq:lemma1_check_neighbors}, observe that all VV-type and CC-type variable nodes  in $\mathscr{T}(c^2_jv^1_{i})$ are, respectively, of the form $v_{k'}^2v^1_i,$ where $v_{k'}^2 \in \mathcal{N}_{c_j^2}^2,$ and  $c_j^2c_{l'}^1$, where $c_{l'}^1 \in \mathcal{N}^1_{v_i^1}.$
        So denote the VV-type variable node connected to check node $v_k^2c_l^1$ by $v_{k}^2v^1_i.$
        The CC-type variable nodes connected to $v_k^2c_l^1$ is given by 
        \begin{equation}
            \label{eq:lemma1_CCtype_VC_neighbor}
            \mathcal{N}_{v_k^2}^2 \times c_l^1.
        \end{equation}
        
       From equations~\eqref{eq:lemma1_check_neighbors}, \eqref{eq:lemma1_CCtype_VC_neighbor} and the facts
       that $c_l^1 \in \mathcal{N}^1_{v_i^1}, c_{j}^2 \in \mathcal{N}_{v_k^2}^2$, it follows that
       $$\left( c^2_j \times \mathcal{N}^1_{v_i^1}\right) \cap \left(\mathcal{N}_{v_k^2}^2 \times c_l^1\right) = c_j^2c_l^1.$$ This completes the proof of Part~\ref{CN_connectivity}.
       
       \item We assume that there exists a variable node in $\mathscr{G}_\mathrm{Z}\setminus \mathscr{T}(c_j^2v_i^1)$ that is connected two checks in $\mathscr{T}(c_j^2v_i^1)$ and prove a contradiction.
       Consider two check nodes, say $v_k^2c_l^1$ and $v_{k'}^2c_{l'}^1$ in $\mathscr{T}(c_j^2v_i^1)$ that are connected to a VV-type variable node $\mathscr{G}_\mathrm{Z}\setminus \mathscr{T}(c_j^2v_i^1)$.  
       The sets of VV-type variable nodes connected to $v_k^2c_l^1$ and $v_{k'}^2c_{l'}^1$ are $v_k^2 \times \mathcal{N}^1_{c_l^1}$ and $v_{k'}^2 \times \mathcal{N}^1_{c_l^1}$, respectively.
       Therefore, for $Z$-type check nodes $v_k^2c_l^1$ and $v_{k'}^2c_l^1$ to have a common VV-type variable node as neighbor, $v_k^2$ and $v_{k'}^2$ should be the same variable node. 
       Also, the common VV-type variable node is of the form $v_k^2v_{k''}^1$ such that 
       \begin{equation}
       \label{eq:intermediate_result}
       v_{k''}^1 \in \mathcal{N}^1_{c_l^1} \bigcap \mathcal{N}_{c_{l'}^1}^1. 
       \end{equation}
      According to part \ref{CN_connectivity}, check node $v_k^2c_l^1$ is connected to a variable node of the VV-type, denoted as $v_{f}^2v_{g}^1$, and check node $v_{k'}^2c_{l'}^1$ is similarly connected to a VV-type variable node, denoted as $v_{f'}^2v_{g'}^1$ within $\mathscr{T}(c_j^2v_i^1)$. 
      Variable nodes $v_f^2v_g^1$ and $v_{f'}^2v_{g'}^1$ are in $\mathscr{T}(c_j^2v_i^1)$, which implies $v_g^1=v_{g'}^1=v_i^1$.
      In addition, variable nodes $v_f^2v_g^1$ and $v_{f'}^2v_{g'}^1$ are connected to $v_k^2c_l^1$ and $v_k^2c_{l'}^1$, which implies $v_f^2=v_{f'}^2=v_k^2$. Hence,  $v_f^2v_g^1=v_{f'}^2v_{g'}^1=v_k^2v_i^2$.
      Since $v_k^2v_i^1 \in \mathcal{N}_{v_k^2c_l^1} \bigcap \mathcal{N}_{v_k^2c_{l'}^1}$, $$v_i^1 \in \mathcal{N}_{c_l}^1 \bigcap \mathcal{N}_{c_{l'}}^1.$$
      We have shown in \eqref{eq:intermediate_result} that $v_{k''}^1\in\mathcal{N}^1_{c_l^1} \cap \mathcal{N}^1_{c_{l'}^1}$.
       This implies that $\mathscr{G}_1$ has a cycle of length four, which is a contradiction.
       Similarly, it can be shown that when two check noes in $\mathscr{T}(c_j^2v_i^1)$ shares a VV-type variable node in $\mathscr{G}_{\mathrm{Z}}\setminus \mathscr{T}(c_j^2v_i^1)$, there is a cycle of length four in $\mathscr{G}_1$, which is a contradiction.
    \end{enumerate}
\end{proof}
\begin{lemma}
\label{lemma:unsatisfied_checks_in_TS}
     Consider subgraph $\mathscr{T}(c^2_jv^1_{i})$ that is induced by $X-$check $c^2_jv^1_{i}$ on Tanner graph $\mathscr{G}_{\mathrm{Z}}$. If $\alpha$ VV-type and $\beta$ CC-type variable nodes in $\mathscr{T}(c^2_jv^1_{i})$ are erroneous while none in $\mathscr{G}_{\mathrm{Z}} \setminus \mathscr{T}(c^2_jv^1_{i})$ are, each erroneous CC-type and VV-type node in $\mathscr{T}(c^2_jv^1_{i})$ is connected to $d_{c}^2-\alpha$ and $d_{v}^1-\beta$ unsatisfied checks, respectively. Moreover, each non-erroneous CC-type and VV-type node is connected to $\alpha$ and $\beta$ unsatisfied checks, respectively.
\end{lemma}
\begin{proof}
    Refer to Part~\ref{CN_connectivity} of Lemma~\ref{lemma:structures_of_graph_induced_by_stabilizers} and observe that all check nodes in $\mathscr{T}(c_l^2v_k^1)$ exhibit a degree of two, each check node being connected to a variable node of VV-type and CC-type. Consequently, if the error estimates on the two variable nodes connected to a check node do not correspond to the actual error, the check node remains satisfied.
Therefore, any pair of erroneous VV-type and CC-type variable nodes share one satisfied check node in $\mathscr{T}(c_l^2v_k^1)$. Thus, an erroneous VV-type variable node shares a satisfied check node with each of the $\beta$ erroneous CC-type variable nodes.
As a direct consequence, each erroneous VV-type variable node is connected to $d_v^1-\beta$ unsatisfied check nodes. 
Analogous arguments suggest that each erroneous VV-type variable node is connected to $d_c^2-\alpha$ unsatisfied check nodes.
Additionally, it should be noted that every erroneous CC-type variable node shares an unsatisfied check node with each of the correct VV-type variable nodes. Consequently, every correct VV-type variable node is connected to $\beta$ unsatisfied check nodes.
Similar reasoning applies, indicating that each correct CC-type variable node is connected $\alpha$ unsatisfied check nodes.
\end{proof}
\textbf{Proof of part 1 of Lemma~\ref{lemma:TSs_induced_by_stabilizers}:} 

 Consider the case in which exclusively all VV-type variable nodes are in error. 
During the first iteration, the bit-flipping decoder initializes the error estimate to $\boldsymbol{0}$. 
In other words, all $d_c^2$ VV-type variable nodes are in error, while none of the CC-type variable nodes are in error.
From Lemma~\ref{lemma:unsatisfied_checks_in_TS}, it follows that none of the check nodes adjacent to any variable node in $\mathscr{T}(c_j^2v_i^1)$ is satisfied, causing the decoder to flip the error estimate across all variable nodes.  
 As a result, at the beginning of the second iteration, the support of the mismatched error lies exactly on the CC-type nodes, again rendering all the check nodes unsatisfied.  
Hence, the decoder again flips the error estimate across all variable nodes. 
At the beginning of the third iteration, the support of mismatched error lies on all VV-type nodes, as was the case at the beginning of the first iteration. 
Hence, it is concluded that the support of the estimated error oscillates between being on VV-type nodes and being on CC-type nodes without converging.

To show that $\mathscr{T}(c_j^2v_i^1)$ constitutes a TS, we also need to prove that the
error patterns with the support located in $\mathscr{T}(c^2_jv^1_{i})$ are restricted to this region and will not extend to other sections of $\mathscr{G}_{\mathrm{Z}}$ throughout the decoding process.
 From part~\ref{outside_connectivity_constraint} of Lemma~\ref{lemma:structures_of_graph_induced_by_stabilizers}, we know that there does not exist a variable node in $\mathscr{G}_\mathrm{Z}\setminus \mathscr{T}(c_j^2v_i^1)$ that is connected to two check nodes in $\mathscr{T}(c_j^2v_i^1)$.
In other words, any variable node in $\mathscr{G}_\mathrm{Z}\setminus \mathscr{T}(c_j^2v_i^1)$ is connected to at most one check node in $\mathscr{T}(c_j^2v_i^1)$. 
Therefore, when the support of an error pattern lies exclusively in $\mathscr{T}(c_j^2v_i^1)$, all unsatisfied check nodes lie in $\mathscr{T}(c_j^2v_i^1)$, and consequently any variable node in $\mathscr{G}_\mathrm{Z}\setminus \mathscr{T}(c_j^2v_i^1)$ is connected to a maximum of one unsatisfied check node.
The bit flip decoder flips a variable node when the number of unsatisfied checks connected to it is strictly greater than one, when the variable node degree is at least two. As a result, if the support of the error pattern completely lies in $\mathscr{T}(c_j^2v_i^1)$, then the error does not spread to $\mathscr{G}_{\mathrm{Z}}\setminus \mathscr{T}(c_j^2v_i^1)$ throughout the decoding process.  

\textbf{Proof of part 2 of Lemma~\ref{lemma:TSs_induced_by_stabilizers}:}

Assume that $\alpha$ VV-type variable nodes and $\beta$ CC-type variable nodes within $\mathscr{T}(c_j^2v_i^1)$ are in error.
Consider the case which satisfies $\alpha\geq \lfloor\frac{d_c^2}{2}\rfloor+1$ and $\beta < \lfloor\frac{d_v^1}{2}\rfloor$.
 From Lemma~\ref{lemma:unsatisfied_checks_in_TS} it follows that each of the erroneous VV-type (or CC-type) variable nodes is connected to at least $d_v^1-\beta \geq \lfloor\frac{d_v^1}{2}\rfloor+1$ (or $d_c^2-\alpha \leq \lfloor\frac{d_c^2}{2}\rfloor$) unsatisfied checks.
Since the number of unsatisfied checks connected to the erroneous VV-type variable nodes exceeds the threshold set for flipping, the bit-flipping decoder flips the error estimate on them. 
 The number of unsatisfied check nodes connected to erroneous CC-type variable nodes does not exceed the threshold, and as a result the decoder does not flip the erroneous CC-type variable nodes.
Also, from Lemma~\ref{lemma:unsatisfied_checks_in_TS} it follows that each of the non-erroneous VV-type (or CC-type) variable nodes is connected to at least $\beta < \lfloor\frac{d_v^1}{2}\rfloor $ (or $\alpha \geq \lfloor\frac{d_c^2}{2}\rfloor+1$) unsatisfied checks.
Consequently, the decoder flips all the non-erroneous CC-type variable nodes and does not flip the non-erroneous VV-type variable nodes.
So, the bit-flipping decoder flips the error estimate on all the non-erroneous CC-type variable nodes, as well as those VV-type variable nodes that are in error. As a result, the support for mismatched error pattern lies exactly on the CC-type nodes. According to Part~\ref{lemma1_part1}, this is an error pattern that induces TS.
With similar arguments, the error pattern that satisfies $\alpha < \lfloor\frac{d_c^2}{2}\rfloor$ and $\beta \geq \lfloor\frac{d_v^1}{2}\rfloor+1$ can be argued to be a TS-inducing error pattern.

    \section{Proof of Lemma~\ref{lemma:TS_compositions_condition}}
    \label{sec:appendix_proof_of_sufficient_condtions}

    \begin{lemma}
    \label{lemma:TS_compositions_condition_intermediate}
    Consider the stabilizer-induced subgraph in the $Z$ Tanner graph $\mathscr{G}_{\mathrm{Z}}$ of an HP code  given by $\mathscr{T}(\mathbf{h})$,
    where $\mathbf{h}$ is a stabilizer of type $X$ formed as a linear combination of stabilizer generators of type $X$, i.e., $\mathbf{h} = \sum_{a \in \mathcal{I},b \in \mathcal{J}}c^2_bv_a^1$.
    Let $\Lambda_{(c_j^2v_i^1)}$ and $\Gamma_{(c_j^2v_i^1)}$, respectively, denote the collections of VV-type and CC-type variable nodes that are present in subgraph $\mathscr{T}(c_j^2v_i^1)$, and not in $\mathscr{T}(h)\setminus \mathscr{T}(c_j^2v_i^1)$. 
        Consequently, for a variable node in set $\Lambda_{(c_j^2v_i^1)}$ and a variable node in set $\Gamma_{(c_j^2v_i^1)}$, there exists a check node that connects exclusively to these two variable nodes and not to any other variable nodes within subgraph$\mathscr{T}(h) \setminus \mathscr{T}(c_j^2v_i^1)$.
        \end{lemma}
        \begin{proof}
    We start by establishing that for any pair of variable nodes, one chosen from set $\Lambda_{c_j^2v_i^1}$ and the other from set $\Gamma_{c_j^2v_i^1}$, there exists a unique check node within set $\mathscr{T}(c^2_jv^1_i)$ that exclusively connects to these particular variable nodes.
 From~\eqref{eq:hp_cx_neighbor} recall that the set of variable nodes linked to the $X$-stabilizer $c_j^2v_i^1$ is given by $$\mathcal{N}^{\mathrm{X}}_{c_j^2v_i^1}=\left(c_j^2 \times \mathcal{N}_{v_i^1}^1\right) \bigcup \left(\mathcal{N}^1_{v_i^1} \times c_j^2\right).$$ This implies that VV-type variable nodes associated with check node $c_j^2v_i^1$ conform to the form $v_{i'}^2v_i^1$ where $v_{i'}^2\in \mathcal{N}^2_{v_i}$, whereas CC-type variable nodes linked to check node $c_j^2v_i^1$ conform to the form $c_j^2c_{j'}^1$ where $c_{j'}^1\in \mathcal{N}_{v_i^1}^1$. 
 Consider a VV-type variable node, denoted as $v_{i'}^2v_i^1$, from set $\Lambda_{c_j^2v_i^1}$ and a CC-type variable node, denoted as $c_j^2c_{j'}^1$, from set $\Gamma_{c_j^2v_i^1}$.
    The set of $Z$ checks that are connected to both $v_{i'}^2v_i^1$ and $c_j^2c_{j'}^1$ is given by
    \begin{equation}
    \label{eq:lemma2_common_zcheck}
        \mathcal{N}_{v_{i'}^2v_i^1}^\mathrm{Z} \bigcap \mathcal{N}_{c_j^2c_{j'}^1}^{\mathrm{Z}} = \left(v_{i'}^2 \times \mathcal{N}_{v_{i}}^1\right) \bigcap \left(\mathcal{N}_{c_{j}}^2 \times c_{j'}^1\right).
    \end{equation}
    Given $v_{i'}^2 \in \mathcal{N}_{c^1_j}^1$ and $c_{j'}^1 \in \mathcal{N}_{v^1_i}^1$, it follows from~\eqref{eq:lemma2_common_zcheck} that $ \mathcal{N}_{v_{i'}^2v_i^1}^\mathrm{Z} \bigcap \mathcal{N}_{c_j^2c_{j'}^1}^{\mathrm{Z}}=v_{i'}^2c_{j'}^1$. 
    Furthermore, since the variable nodes $v_{i'}^2v_{i}^1$ and $c_{j}^1c_{j'}^2$ are neighbors of  the $X$ stabilizer $c_j^2v_i^1$, which subsequently induces subgraph $\mathscr{T}(c_j^2v_i^1)$, the $Z$-type check node $v_{i'}^2c_{j'}^1$ is present within subgraph $\mathscr{T}(c_j^2v_i^1)$. 
    According to Part~\ref{CN_connectivity} of Lemma~\ref{lemma:structures_of_graph_induced_by_stabilizers}, a $Z$ type check connects exclusively to a single VV type and a CC-type variable node within subgraph induced by any $X$ type stabilizer generators, which consequently indicates that the $Z$ type check $v_{i'}^2c_{j'}^1$ maintains no connections to variable nodes other than $v_{i'}^2v_{i}^1$ and $c_{j}^1c_{j'}^2$ within subgraph $\mathscr{T}(c_j^2v_i^1)$.
    
    To complete the proof, we still need to show that check node  $v_{i'}^2c_{j'}^1$ is not connected to any other variable other than these two variable nodes in subgraph $\mathscr{T}(\mathbf{h})$.
    To do so, we assume that check node  $v_{i'}^2c_{j'}^1$ is present in $\mathscr{T}(c_l^2v_k^1)$ for $k \in \mathcal{I}$ and $l \in \mathcal{J}$ and show a contradiction. 
    Since check node  $v_{i'}^2c_{j'}^1$ is present in $\mathscr{T}(c_l^2v_k^1)$ which is induced by the variable nodes connected to $X$-type check $c_l^2v_k^1$, there exists a variable node within $\mathscr{T}(c_l^2v_k^1)$ that connects to check node $v_{i'}^2c_{j'}^1$. 
    Without loss of generality, assume that the variable node connected to check node $v_{i'}^2c_{j'}^1$ within subgraph $\mathscr{T}(c_l^2v_k^1)$ is a VV-type variable node, say $v_f^2v_g^1$.
    Since the set of variable nodes connected to $c_l^2v_k^1$ is given by 
    \begin{equation}
    \label{eq:lemma2_clvk_neighbor}
    \mathcal{N}_{c_l^2v_k^1}^\mathrm{X}= \left(c_l^2 \times \mathcal{N}_{v_k^1}^1\right) \bigcup \left(\mathcal{N}^2_{c_j^2} \times v_k^1\right)
    \end{equation}
    and variable node $v_f^2v_g^1$ is from set $ \mathcal{N}_{c_l^2v_k^1}^\mathrm{X}$, it follows that $v_g^1=v_k^1$.
    Also, $v_f^2v_g^1$ is connected to $X$-check $v_{i'}^2c_{j'}^1$ and the set of variable nodes connected to $v_{i'}^2c_{j'}^1$ is given by 
    \begin{equation}
        \label{eq:lemma2_vicj_neighbor}
         \mathcal{N}_{v_{i'}^2c_{j'}^1}^\mathrm{Z}= \left(v_{i'}^2 \times \mathcal{N}_{c_{j'}^1}^1\right) \bigcup \left(\mathcal{N}^2_{v_{i'}^2} \times c_{j'}^1\right),
    \end{equation}
    it follows that $v_f^2=v_{i'}^2$. 
    From now on, we refer to the variable node $v_{f}^2v_{g}^1$ from subgraph $\mathscr{T}(c_l^2v_k^1)$ as $v_{i'}^2v_{k}^1$.
    Also, since $v_{i'}^2v_{k}^1 \in \mathcal{N}_{v_{i'}^2c_{j'}^1}^{\mathrm{Z}}$, from~\eqref{eq:lemma2_vicj_neighbor}, it follows that $v_{k}^1\in \mathcal{N}_{c_{j'}}^1$, which further implies that $c_{j'}^1 \in \mathcal{N}_{v_{k}^1}^1$.
    Now consider subgraph $\mathscr{T}(c_j^2v_k^1)$.  
    The CC-type variable nodes present in subgraph $\mathscr{T}(c_j^2v_k^1)$ is given by 
    $c_j^2 \times \mathcal{N}_{v_k^1}$. 
    Since $c_{j'}^1 \in \mathcal{N}_{v_{k}^1}^1$, it follows that $c_j^2c_{j'}^1 \in \mathcal{N}_{c_j^2v_k^1}$.
    We start with the assumption that variable node $c_{j}^2c_{j'}^1 \in \Gamma_{c_j^2v_i^1}$, which says that variable node $c_{j}^2c_{j'}^1$ is present in $\mathscr{T}(c_j^2v_i^1)$ and not in any other subgraph $\mathscr{T}(c_{j''}^2v_{i''}^1) \subset \mathscr{T}(\mathbf{h})$.
    However, we have shown that variable node $c_{j}^2c_{j'}^1$ is present in $\mathscr{T}(c_j^2v_k^1)$, which is a contradiction. 
    \end{proof}
    \textbf{Proof of Lemma~\ref{lemma:TS_compositions_condition}:}
    
     Consider the sets $\Lambda_{c_j^2v_i^1}$ and $\Gamma_{c_j^2v_i^1}$ as defined in Lemma~\ref{lemma:TS_compositions_condition_intermediate}.
    Define $$\Lambda(\mathbf{h})\coloneqq \bigcup_{a \in \mathcal{I}, b \in \mathcal{J}}\Lambda_{c_b^2v_a^1} \text{ and } \Gamma(\mathbf{h}) \coloneqq \bigcup_{a \in \mathcal{I}, b \in \mathcal{J}}\Gamma_{c_b^2v_a^1}.$$
     To show that the stabilizer-induced graph $\mathscr{T}(\mathbf{h})$ constitutes a TS, consider an error pattern whose support lies in $\mathscr{T}(\mathbf{h})$; further, assume that $\Lambda(\mathbf{h})$ is contained within the support of this error pattern and that $\Gamma(\mathbf{h})$ intersects trivially with the support of the error pattern.
     Next, we show that the bit-flipping decoder flips the error estimate on the variable nodes in $\Lambda_{c_j^2v_i^1}$.
     Based on Lemma~\ref{lemma:TS_compositions_condition_intermediate}, for a given $v^2v_i^1 \in \Lambda_{c_j^2v_i^1}$, there exists a check node in $\mathscr{T}(c_j^2v_i^1)$ that is only connected to the VV-type variable node $v_i^2v^1$ and CC-type variable node $c^2c^1 $, for $c^2_jc^1 \in \Gamma_{c^2_jc^1}$, within $\mathscr{T}(c_j^2v_i^1)$.
     This holds for every CC-type variable node in $\Gamma_{c^2_jv^1}$, suggesting that there are at least $\left|\Gamma_{c^2_jv^1}\right|$ degree two check nodes in the neighborhood of $v^2v_i^1$ that are connected to a CC-type variable node in $\Gamma_{c_j^2v_i^1}$.
    Since we have assumed that the variable nodes in $\Lambda_{c^2_jv_i^1}$ are in error and the variable nodes in $\Gamma_{c_j^2v_i^1}$ are not, variable node $v^2v_i^1$ is at least connected to $|\Gamma_{c^2_jv^1}|$ unsatisfied check nodes, regardless of the status of the variable nodes that are not in both $\Lambda_{c_j^2v^1_i}$ and $\Gamma_{c_j^2v^1_i}$.
    Since $|\Gamma_{c^2_jv^1}|\geq\lfloor\frac{d_v^1}{2}\rfloor+1$, variable node $v^2v_i^1$ is connected to at least $\lfloor\frac{d_v^1}{2}\rfloor+1$ unsatisfied checks, and as a result, the decoder flips the error estimate in $v^2v_i^1$.
    The above arguments hold for any $v^2v_i^1 \in \Lambda_{c_j^2v_i^1}$, indicating that the decoder flips the error estimate on every variable node in set $\Lambda_{c_j^2v_i^1}$.
    Furthermore, the arguments hold for $\Lambda_{c_j^2v_i^1}$, for any $i \in \mathcal{I}$ and $j \in \mathcal{J}$, indicating that the decoder flips the error estimate on all variable nodes in $\Lambda(\mathbf{h})$.

    Subsequently, we demonstrate that the decoder also flips the error estimates on the variable nodes within $\Gamma_{\mathbf{h}}$. To this end, we consider the CC-type variable node $c_j^2c^1$ in $\Gamma_{c_j^2v_i^1}$. 
    According to Lemma~\ref{lemma:TS_compositions_condition_intermediate}, there exists a degree-two check node within $\mathscr{T}(\mathbf{h})$ for each VV-type variable node in $\Lambda_{c_j^2v_i^1}$ that is connected to both $c_j^2c^1$ and the considered VV-type variable node from $\Lambda_{c_j^2v_i^1}$. 
    Given that all VV-type variable nodes in $\Lambda_{c_j^2v_i^1}$ are erroneous while $c_j^2c^1$ is correct, the variable node $c_j^2c^1$ is linked to at least $|\Lambda_{c_j^2v_i^1}|$ unsatisfied checks. Since $|\Lambda_{c_j^2v_i^1}|\geq \lfloor \frac{d_c^2}{2}\rfloor+1$, the variable node $c^2_jc^1$ is associated with a minimum of $\lfloor \frac{d_c^2}{2}\rfloor+1$ unsatisfied checks, leading the decoder to flip $c^2_jc^1$. 
    This reasoning applies to every CC-type variable node in $\Gamma(\mathbf{h})$, resulting in the decoder flipping all of them. 
    Consequently, the error estimate differs from the actual error on the variable nodes in $\Gamma_{c_j^2v_i^1}$. 
    By employing similar arguments, it can be shown that the decoder inverts the error estimates on all variable nodes in $\Gamma(\mathbf{h}) \cup \Lambda(\mathbf{h})$ in the subsequent iteration, thus indicating that $\mathscr{T}(\mathbf{h})$ forms a TS.
       
        
\section{Proof of properties of the labeling scheme}
\label{proof_of_properties_2_and_3}
To see that the labeling scheme satisfies the second property, consider two check nodes in the neighborhood of a VV-type variable node, say $c^2c^1$.
Denote these two check nodes in $\mathcal{N}_{c^2c^1}^\mathrm{Z}$ by $v_i^2c^1$ and $v_{i'}^2c^1$, respectively, where $v_i^2, v_{i'}^2 \in \mathcal{N}^2_{c^2}$.
According to the labeling scheme, the VV-type variable nodes connected to $v_{i}^2c^1$ have the same label as that of the variable node $v_{i}^2$ in the Tanner graph $\mathscr{G}_2$, while the VV-type variable nodes connected to $v_{i'}^2c^1$ have the same label as that of variable node $v_{i'}$.
Given that variable nodes $v_i^2$ and $v_{i'}^2$ are labeled differently in the Tanner graph $\mathscr{G}_2$, it follows that the VV-type variable nodes associated with check node $v_{i}^2c^1$ have distinct labels from those connected to the check node $v_{i'}^2c^1$.

To see that the labeling scheme satisfies the third property, consider the subgraph $\mathscr{T}(c^2v^1)$ induced by the $X$-type stabilizer generator or check $c^2v^1$.
The set of variable nodes connected to $X$ check is given by
$$\mathcal{N}^\mathrm{X}_{c^2v^1}= \left(\mathcal{N}^2_{c^2} \times v^1 \right) \bigcup \left(c^2 \times \mathcal{N}^\mathrm{1}_{v^1}\right).$$
Note that the VV-type variable nodes in $\mathcal{N}^\mathrm{X}_{c^2v^1}$ are of the form $v_{i}^2v^1$ for $v_i \in \mathcal{N}_{c^2}^2$.
Given that the VV-type variable node $v_{i}^2v^1$ is assigned the same label as variable node $v_i^2$ within the Tanner graph $\mathscr{G}_2$, and that any pair of variable nodes, specifically $v_i^2$ and $v_{i'}^2$ from $\mathcal{N}^2_{c^2}$, are assigned distinct labels, it follows that the VV-type variable nodes $v_{i}^2v^1$ and $v_{i'}^2v^1$ also possess distinct labels.
Following similar arguments, it can be shown that two CC-type variable nodes within subgraph $\mathscr{T}(c^2v^1)$ have distinct labels, as stated in the third property. 
\section{Proof of Lemma~\ref{lemma:vn_only_graphical_description}}
 \label{proof_of_lemma_vn_only_graphical_description} 
 \textbf{Proof of Part 1:} 
 
 Consider a check node, say $v^2c^1$, in $\mathcal{N}_{c^2c^1}^\mathrm{Z}$, for $c^2c^1 \in \mathscr{T}(\mathbf{h})$, that is connected to the VV-type variable nodes with label $\rho$.
Note that all the VV-type variable nodes connected to the check node $v^2c^1$ have label $\rho$, however, all VV-type variable nodes with label $\rho$ within the stabilizer-induced graph $\mathscr{T}(\mathbf{h})$ are not connected to $v^2c^1$.
To determine the set of VV-type variable nodes connected to check node $v^2c^1$, recall that $\mathscr{T}(\mathbf{h})$ is collectively induced by $X$ checks of the form $c_b^2v_a^1$ where $a \in \mathcal{I}$ and $b \in \mathcal{J}$.
Therefore, the CC-type variable node $c^2c^1$ must be present in one or several subgraphs induced by $X$ checks of the form $c_b^2v_a^1$ where $a \in \mathcal{I}$ and $b \in \mathcal{J}$.
The collection of $X$ checks that contains the CC-type variable node $c^2c^1$ in their support and has its induced subgraph located within $\mathscr{T}(\mathbf{h})$ is given by $$\mathcal{Q}_{c^2c^1}=\mathcal{N}_{c^2c^1}^{\mathrm{X}} \cap \mathscr{T}(\mathbf{h}).$$
For every X-type check node $c^1v^2 \in \mathcal{N}_{c^2c^1}^{\mathrm{X}} \cap \mathscr{T}(\mathbf{h})$, the induced subgraph $\mathscr{T}(c^2v^1)$ contains CC-type variable node $c^2c^1$ and check node $v^2c^1$.
Recall that Part~\ref{CN_connectivity} of Lemma~\ref{lemma:structures_of_graph_induced_by_stabilizers} states that every check node in the subgraph induced by a stabilizer generator is connected to exactly one CC-type variable node and one VV-type variable node.
Therefore, there exists a VV-type variable node within $\mathscr{T}(c^2v^1)$ that is connected to the check node $v^2c^1$.
As per the first property of the labeling scheme described in Section~\ref{sec:concise_rep_stabilizer_induced_TS},  all VV-type variable nodes associated with check node $v^2c^1$ have label $\rho$.
Consequently, this assists in identifying the VV-type variable node linked to check node $v^2c^1$ within subgraph $\mathscr{T}(c^2v^1)$. 
Similarly, the VV-type variable nodes associated with check node $v^2c^1$ within the stabilizer-induced subgraph $\mathscr{T}(\mathbf{h})$  corresponds to the VV-type variable nodes with label $\rho$ in $\cup_{S \in \mathcal{Q}_{c^2c^1}}\mathscr{T}(S)$.

\textbf{Proof of Part 2:}  
Consider a check node, say $v^2c^1$, in $\mathcal{N}^\mathrm{Z}_{c^2c^1}$, for the CC-type variable node $c^2c^1 \in \mathscr{T}(\mathbf{h})$.
Assume that the CC-type variable node $c^2c^1$ has label $\lambda$.
According to Property 1 of the labeling scheme described in Section~\ref{sec:concise_rep_stabilizer_induced_TS}, the CC-type variable nodes connected to a specific check node possess the same label. Consequently, this indicates that all CC-type variable nodes linked to check node $v^2c^1$ have label $\lambda$.
However, all CC-type variable nodes labeled $\lambda$ in subgraph $\mathscr{T}(\mathbf{h})$ are not connected to check node $v^2c^1$.
Assume that the VV-type variable nodes connected to the check node $v^2c^1$ have label $\rho$ and denote the set of such VV-type variable nodes within subgraph  $\mathscr{T}(\mathbf{h})$ by $\Omega_{v^2c^1}^\rho$.
Consider a VV-type variable node, say $v^2v^1$, from set $\Omega_{v^2c^1}^\rho$.
The collection of $X$ checks that contains the VV-type variable $v^2v^1$ and have their induced graphs located within stabilizer-induced graph $\mathscr{T}(\mathbf{h})$ is given by 
$$\mathcal{Q}_{v^2v^1}=\mathcal{N}^\mathrm{X}_{v^2v^1} \cap \mathscr{T}(\mathbf{h}).$$
Note that for every $S \in \mathcal{Q}_{v^2v^1}$, the induced graph $\mathscr{T}(S)$ contains the VV-type variable node $v^2v^1$ and check node $v^2c^1$.
Moreover, there exists exactly one CC-type variable node within graph $\mathscr{T}(S)$ that is connected to check node $v^2c^1$ according to Part~\ref{CN_connectivity} of Lemma~\ref{lemma:structures_of_graph_induced_by_stabilizers}.
In accordance with Property 1 of the labeling scheme, it is established that all CC-type variable nodes linked to check node $v^2c^1$ are assigned label $\lambda$. Furthermore, according to Property 3 of the labeling scheme, it is determined that there exists precisely one CC-type variable node bearing the label $\lambda$ within Subgraph $\mathscr{T}(S)$.
Therefore, the CC-type variable node within $\mathscr{T}(S)$ that is connected to check node $v^2c^1$ can be identified from its label $\lambda$.
The above discussion holds for any $X$ check $S \in \mathcal{Q}_{v^2v^1}^\rho$, where $v^2v^1\in \Omega_{v^2c^1}$.
This suggests that if we define 
$$\mathcal{R}_{v^2c^1}^\rho=\bigcup_{v^2v^1 \in \Omega_{v^2c^1}^\rho}\mathcal{Q}_{v^2v^1},$$
the CC-type variable node in subgraph $\mathscr{T}(S)$, for $S \in \mathcal{R}_{v^2c^1} $, with label $\lambda$ is connected to check node $v^2c^1$.

Subsequently, we show that each CC-type node associated with the check node $v^2c^1$ within $\mathscr{T}(\mathbf{h})$ corresponds to the CC-type variable node bearing label $\lambda$ in a stabilizer-induced graph $\mathscr{T}(S)$, for $S \in \mathcal{R}_{v^2c^1}^\rho$.
For this purpose, consider a CC-type variable node within $\mathscr{T}(\mathbf{h})$ that is connected to check node $v^2c^1$.
This CC-type variable node must be present in at least one subgraph of the form $\mathscr{T}(c_j^2v_i^1)$ for $i \in \mathcal{I}$ and $j \in \mathcal{J}$.
According to Part~\ref{CN_connectivity} of Lemma~\ref{lemma:structures_of_graph_induced_by_stabilizers} check node $v^2c^1$ is connected to a VV-type variable node within $\mathscr{T}(c_j^2v_i^1)$.
Recall that VV-type variable nodes connected to check node $v^2c^1$ have label $\rho.$
However, the VV-type variable nodes bearing label $\rho$ and connected to check node $v^2c^1$ within $\mathscr{T}(\mathbf{h})$ are in the set $\Omega_{v^2c^1}^\rho$.
This shows that every CC-type variable node connected to $\mathcal{N}_{v^2c^1}^\mathrm{Z}$ is present in $\mathscr{T}(S)$ for $S \in \mathcal{R}_{v^2c^1}^\rho$ and is labeled $\lambda$. 
\section{Proof of Theorem~\ref{theorem:TS_aware_bf_decoder}}
\label{appendinx_proof_TS_aware_decoder}
Before proving Theorem~\ref{theorem:TS_aware_bf_decoder}, we prove the following lemma, which is used to prove the theorem.
\begin{lemma}
   \label{lemma:TS_aware_bf_decoder_intermediate_results} 
    Let $\cup_{j \in \mathcal{J}}\mathscr{T}(c_j^2)$ denote a connected cycle-free subgraph of $\mathscr{G}_2$ induced by $\{c_j^2: j \in \mathcal{J}  \}$, and similarly let $\cup_{i \in \mathcal{I}}\mathscr{T}(v_i^2)$ denote a connected cycle-free subgraph of $\mathscr{G}_1$ induced by $\{v_i^1: i \in \mathcal{I}  \}$, where the index sets $\mathcal{I}$ and $\mathcal{J}$ are non-empty.
   Consider the subgraph $\mathscr{T}(\mathbf{h}) \coloneqq \cup_{i \in \mathcal{I},j \in \mathcal{J}}\mathscr{T}(c_j^2v_i^1)=\cup_{j \in \mathcal{J}}\mathscr{T}(c_j^2) \times \cup_{i \in \mathcal{I}}\mathscr{T}(v_i^2) $.
   \begin{enumerate}
       \item \label{lemma_TS_composition_bf_decoder:part1} There exists $l \in \mathcal{J}$ such that $\mathscr{T}(c_l^2v_i^1)$ and $ \mathscr{T}(\mathbf{h})\setminus\mathscr{T}(c_l^2v_i^1) $ for any $i \in \mathcal{I}$ have exactly one common VV-type variable node when $\left|\mathcal{J}\right|> 1.$
       Assuming the above statement is true and such a $l$ exists, define $\mathcal{J}'=\mathcal{J}\setminus l$.
       Then, the same statement is also true for $\mathscr{T}(\mathbf{h}') \coloneqq \cup_{i \in \mathcal{I},j \in \mathcal{J}' }\mathscr{T}(c_{j}^2v_{i}^1)$ when $\left|\mathcal{J}'\right|> 1.$
       \item \label{lemma_TS_composition_bf_decoder:part2} There exists  $k \in \mathcal{I}$  such that subgraph  $\mathscr{T}(c_j^2v_k^1)$, for any $j \in \mathcal{J}$, and $ \mathscr{T}(\mathbf{h})\setminus \mathscr{T}(c_j^2v_k^1)$  have exactly one common CC-type variable node when $\left|\mathcal{I}\right|> 1.$
       Assuming that the above statements are true, define recursively $\mathcal{I}'=\mathcal{I} \setminus k$. Then, the same statement is also true for $\mathscr{T}(\mathbf{h}') \coloneqq \cup_{i \in \mathcal{I}',j \in \mathcal{J}}\mathscr{T}(c_{j}^2v_{i}^1)$ when $\left|\mathcal{I}'\right|> 1.$
   \end{enumerate}
\end{lemma}
\begin{proof}
    \begin{enumerate}
        \item Note that subgraph $\mathscr{T}(\mathbf{h})$ of graph $\mathscr{G}_\mathrm{Z}$ is the graph product of connected subgraphs $\cup_{j \in J}\mathscr{T}(c_j^2)$ and $\cup_{i \in I}\mathscr{T}(v_i^1)$ of Tanner graphs $\mathscr{G}_2$ and $\mathscr{G}_1$, respectively. 
     Consider the case where $\left|\mathcal{J}\right|>1$.
      Since subgraph $\cup_{j \in J}\mathscr{T}(c_j^2)$ is connected and finite, there exists a $l \in \mathcal{J}$ such that check node $c_l^2$ shares a variable node with only one other check node in $\cup_{j \in J}\mathscr{T}(c_j^2)$, i.e., there exists a $l' \in \mathcal{J}$, $\mathcal{N}_{c_{l}} \cap \mathcal{N}_{c_{l'}} \neq \emptyset$ and for every $l'' \in \mathcal{J} \setminus \{l,l'\}$, $\mathcal{N}_{c_{l}} \cap \mathcal{N}_{c_{l''}} = \emptyset$.
      Furthermore, $\left|\mathcal{N}_{c_{l}} \cap \mathcal{N}_{c_{l'}}\right|=1$ since subgraph $\cup_{j \in J}\mathscr{T}(c_j^2)$ is cycle-free.
      Denote the variable node that is connected to both $c_l^2$ and $c_{l'}^2$ by $v^2$.
      From the above discussion, it follows that $\mathscr{T}(c_{l'}^2v_i^1)$ and $\mathscr{T}(c_{l}^2v_i^1)$, for any $i \in \mathcal{I}$, have one common variable node, which is given by
      \begin{align*}
          \left( \mathcal{N}_{c_l}^2 \bigcap  \mathcal{N}_{c_{l'}}^2\right) \times v_i^1
          =v^2v_i^1,
      \end{align*}
       proving the case where $\left|\mathcal{J}\right|>1$.
      The remainder of the proof follows by defining $\mathcal{J}=\mathcal{J}'\setminus l$ and subsequently following the above steps.
      \item The proof follows similar arguments as in the proof of part~\ref{lemma_TS_composition_bf_decoder:part1}.
    \end{enumerate}
\end{proof}

\textbf{Proof of Theorem~\ref{theorem:TS_aware_bf_decoder}}

    Consider subgraph $\mathscr{T}(c_l^2v_i^1)$, for $i \in \mathcal{I}$ and $l \in \mathcal{J}$ , which has  one variable node in common with the remainder of the graph, i.e., 
      \begin{equation}
       \label{eq:vn_share_condition}
      \left|\mathscr{T}(c_l^2v_i^1) \bigcap \left(\mathscr{T}(\mathbf{h})\setminus \mathscr{T}(c_l^2v_i^1)\right)\right|=1
      \end{equation}
      In part~\ref{lemma_TS_composition_bf_decoder:part1}, it has been shown that such subgraphs always exist.
     We start by studying the convergence of the decoder within subgraph $\mathscr{T}(c_l^2v_i^1)$,  independent of the decoder's behavior within $\mathscr{T}(\mathbf{h})\setminus \mathscr{T}(c_l^2v_i^1)$. 
     In particular, we demonstrate that the TS-aware decoder, as detailed in Algorithm ~\ref{alg:TS_aware_bit_flipping_decoder}, achieves convergence within the subgraph $\mathscr{T}(c_l^2v_i^1)$, with the exception of the variable node it shares with $\mathscr{T}(\mathbf{h})\setminus \mathscr{T}(c_l^2v_i^1)$. 
     Upon convergence of the decoder within subgraph $\mathscr{T}(c_l^2v_i^1)$, the errors from $\mathscr{T}(\mathbf{h})\setminus \mathscr{T}(c_l^2v_i^1)$ do not propagate into the subgraph $\mathscr{T}(c_l^2v_i^1)$ since the subgraph $\mathscr{T}(\mathbf{h})\setminus \mathscr{T}(c_l^2v_i^1)$ forms a TS. Consequently, these procedures may be iterated to study the decoder's convergence in $\bigcup_{i \in \mathcal{I}' j \in \mathcal{J}' }\mathscr{T}(c_j^2v_i^1)$, where $\mathcal{I}'= \mathcal{I}'\setminus k$ and $\mathcal{J}'=\mathcal{J} \setminus l$, provided it is established that the decoder converges after a finite number of iterations across all subgraphs of the form $\mathscr{T}(c_l^2v_i^1)$ for all $i \in \mathcal{I}$, and of the form $\mathscr{T}(c_{j}^2v_{k}^1)$ for all $j \in \mathcal{J}$, where they respectively share a VV-type and CC-type variable node with $\mathscr{T}(\mathbf{h})\setminus \mathscr{T}(c_l^2v_i^1)$ and  $\mathscr{T}(\mathbf{h})\setminus \mathscr{T}(c_j^2v_k^1)$.
    
     Subsequently, we demonstrate that under the conditions delineated in \eqref{eq:vn_share_condition}, the decoder described in Algorithm~\ref{alg:TS_aware_bit_flipping_decoder} converges to the accurate estimate or to the erroneous estimate on all variable nodes located exclusively within subgraph $\mathscr{T}(c_l^2v_i^1)$, after a predetermined number of iterations.
    Assume that subgraph $\mathscr{T}(c_l^2v_i^1)$ shares a VV-type variable node with graph $\mathscr{T}(\mathbf{h})\setminus \mathscr{T}(c_l^2v_i^1)$.
    Also, assume that among the variable nodes that are exclusively present within $\mathscr{T}(c_l^2v_i^1)$, $\beta$ CC-type variable nodes are erroneous and $\alpha$ VV-type variable nodes are erroneous.
    Recall that in the first iteration, the decoder flips the error estimate only on the VV-type variable nodes based on the number of unsatisfied checks connected to them. 
     Using similar arguments as in Lemma~\ref{lemma:unsatisfied_checks_in_TS}, it can be deduced that each of the erroneous VV-type variable nodes exclusive to subgraph $\mathscr{T}(c_l^2v_i^1)$ is connected to $d_v^1-\beta$ unsatisfied checks.
     The status of the check nodes connected to the VV-type variable node common to both $\mathscr{T}(c_l^2v_i^1)$ and $\mathscr{T}(\mathbf{h})\setminus \mathscr{T}(c_l^2v_i^1)$ is unknown since the status of the variable nodes within $\mathscr{T}(\mathbf{h})$ are unknown.
    However, these check nodes do not affect the number of unsatisfied checks connected to VV-type variable nodes that are exclusively present in $\mathscr{T}(c_l^2v_i^1)$, because according to Part~\ref{non_intersecting_sypport_vns} of Lemma~\ref{lemma:structures_of_graph_induced_by_stabilizers}, no two variable nodes present in $\mathscr{T}(c_l^2v_i^1)$ have common check nodes as neighbors.
    Therefore, according to Lemma~\ref{lemma:unsatisfied_checks_in_TS}, the VV-type variable nodes present in $\mathscr{T}(c_l^2v_i^1)$, and not in $\mathscr{T}(\mathbf{h})\setminus \mathscr{T}(c_l^2v_i^1)$, are connected to $d_v^1-\beta$ unsatisfied checks.
     Using analogous reasoning, it can be inferred that each of the non-erroneous VV-type variable nodes, excluding the common VV-type variable node,  is connected to $\beta$ a unsatisfied checks.
    We divide the proof into two cases: i) $\beta \leq \left\lfloor\frac{d_v^1}{2}\right\rfloor$, ii)$\beta > \left\lfloor\frac{d_v^1}{2}\right\rfloor$.

    Case i): Since $\beta \leq \left\lfloor\frac{d_v^1}{2}\right\rfloor$, the erroneous VV-type variable nodes are connected to 
    $$d_v^1-\beta > d_v^1-\left\lfloor\frac{d_v^1}{2}\right\rfloor\geq\frac{d_v^1}{2}+1$$
    unsatisfied checks, leading the decoder to flip these erroneous VV-type variable nodes.
    On the other hand, the non-erroneous VV-type variable nodes are connected to $\left\lfloor\frac{d_v^1}{2}\right\rfloor$ unsatisfied checks and therefore are not flipped.
    As a result, at the beginning of the next iteration,  there are now only $\beta$ erroneous CC-type variable nodes, and the state of the shared VV-type variable node is unknown.
    From Lemma~\ref{lemma:unsatisfied_checks_in_TS}, it can be concluded that each erroneous CC-type variable node is at least connected to the $d_c^2-1$ checks, which leads the decoder to flip the erroneous CC-type variable nodes.
    The correct CC-type variable nodes are connected to at most one unsatisfied checks, and therefore the decoder leaves them changed.
    The above arguments show that, in this case, the decoder converges to the correct estimate of the error on the variable nodes that are only present in the subgraph $\mathscr{T}(c_l^2v_i^1)$.
   
    Case ii): Since $\beta > \left\lfloor\frac{d_v^1}{2}\right\rfloor$, the erroneous VV-type variable nodes are connected to 
    $$d_v^1-\beta \leq d_v^1-\left\lfloor\frac{d_v^1}{2}\right\rfloor-1\leq\frac{d_v^1}{2}$$ unsatisfied checks.
    The non-erroneous VV-type variable nodes are connected to at least $\left\lfloor\frac{d_v^1}{2}\right\rfloor+1$ unsatisfied checks. Therefore, the decoder leaves the erroneous VV-type variable unchanged and flips the correct VV-type variable node.
     As a result, there are $d_c^2-1$ erroneous VV-type variable nodes and $\beta$ erroneous CC-type variable nodes at the beginning of the subsequent iteration. The state of the shared VV-type variable node is unknown.
     From Lemma~\ref{lemma:unsatisfied_checks_in_TS}, it can be concluded that each erroneous CC-type variable node is connected to at most one unsatisfied check nodes, and hence the decoder does not flip the erroneous CC-type variable nodes.
     The correct variable nodes are connected to $d_c^2-1$ unsatisfied checks, so the decoder flips the correct CC-type variable nodes.
     In subsequent iterations, the decoder does not change the state of the variable nodes that are only present in subgraph $\mathscr{T}(c_l^2v_i^1)$.
     
Now consider the case where $\mathscr{T}(c_j^2v_k^1)$ and $\mathscr{T}(\mathbf{h})\setminus \mathscr{T}(c_j^2v_k^1)$ have a common CC-type variable node.  The proof for this case is similar to that for the case in which $\mathscr{T}(c_l^2v_i^1)$ shares a VV-type variable node with $\mathscr{T}(\mathbf{h})\setminus \mathscr{T}(c_l^2v_i^1)$.
This case is analyzed from the second iteration.
Regardless of what happened in the first iteration, it can be assumed that among the variable nodes that are exclusively present within $\mathscr{T}(c_l^2v_i^1)$, $\beta$ CC-type variable nodes are erroneous and $\alpha$ VV-type variable nodes are erroneous.
Using similar arguments as in Lemma~\ref{lemma:unsatisfied_checks_in_TS}, it can be deduced that each of the erroneous CC-type variable nodes exclusive to subgraph $\mathscr{T}(c_l^2v_i^1)$ is connected to $d_c^2-\alpha$ unsatisfied checks.
It can also be inferred that each of the non-erroneous CC-type variable nodes, excluding the common CC-type variable node, is connected to $\alpha$  unsatisfied checks.
Like in the earlier case, the analysis can be divided into two cases: i) $\alpha \leq \left\lfloor\frac{d_c^2}{2}\right\rfloor$, ii)$\alpha > \left\lfloor\frac{d_c^2}{2}\right\rfloor$.
The decoder converges to the correct error estimate on all variable nodes, excluding the common CC-type variable nodes, when $\alpha \leq \left\lfloor\frac{d_c^2}{2}\right\rfloor$. In contrast, it converges to the wrong estimate on all variable nodes, excluding the common CC-type variable node, when $\alpha > \left\lfloor\frac{d_c^2}{2}\right\rfloor$.
Consequently, it has been shown that the decoder reaches convergence on all subgraphs of $\mathscr{T}(\mathbf{h})$ that shares only one variable node with rest of the graph.

  When the decoder converges to the actual error in each of the variable nodes of the subgraphs that share only one variable node with the rest of the graph, the convergence analysis can be extended to the rest of the graph by defining $\mathcal{I}'=\mathcal{I} \setminus k$ and $\mathcal{J}'=\mathcal{J}\setminus l$ .
 Now consider the scenario in which the decoder converges to the wrong error on each of the variable nodes of the subgraphs that share only one variable node with the rest of the graph. 
 Considering that these subgraphs are induced by a stabilizer, this scenario is equivalent to the decoder converging to a wrong estimate on the common variable node and a correct estimate on the remaining variable nodes.
Consequently, the analysis is similarly applicable to the remainder of the graph in this case.
With this, it has been shown that the TS-aware decoder converges to the correct error estimate up to a stabilizer.   

    \section{Proof of Lemma~\ref{lemma:stabilizer_induced_TSs_LP_codes}}
    \label{sec:appendix_proof_stabilizer_induced_TSs_LP_codes}
Let $\underline{v}^2_i$ denote the $i$-th variable node connected to the check node $\underline{c}^2$ within the base graph $\underline{\mathscr{G}}_2$, for $i \in [d_c^2]$. 
  The label corresponding to the edge between the check node $\underline{c}^2$ and the variable node $\underline{v}_i^2$ is denoted by $x^{\mu_i}$. 
  Similarly, let $\underline{c}^1_j$ denote the $j$-th check node connected to the variable node $\underline{v}^1$ within the base graph $\underline{\mathscr{G}}_1$, for $j \in [d_v^1]$. 
  The label corresponding to the edge between variable node $v^1$ and check node $c_j^1$ is denoted by $x^{\nu_j}$. 
  The neighbors of variable node $\underline{c}^2\underline{c}^1_j$ are given by $\mathcal{N}^2_{\underline{c}} \times \underline{c}_j^1$ and the label of the edge that connects $\underline{c}^2\underline{c}^j$ to $\underline{v}^2_i\underline{c}^j$ is the same as that connects $\underline{c}^2$ to $\underline{v}_i^2$ in Tanner graph $\mathscr{G}_2$.
  Observe that within subgraph $\mathscr{T}(\underline{c}_j^2\underline{v}_i^1)$, the paths originating from the CC-type variable node $\underline{c}^2\underline{c}_1^1$ end at the CC-type variable nodes of the form $\underline{c}^2\underline{c}_{j}^1$ for $j\in [d_v^1]\setminus1$. 
 Let the $k$-th copy of the CC-type variable node $\underline{c}^2\underline{c}_{j}^1$, the VV-type variable node $\underline{v}_i^2\underline{v}^1$, and check node $\underline{v}_i^2\underline{c}_j^1$ be represented by $\underline{c}^2\underline{c}_{j}^1(k)$, $\underline{v}_i^2\underline{v}^1(k)$, and $\underline{v}_i^2\underline{c}_j^1(k)$, respectively.
  Without loss of generality, after the lifting procedure, consider the paths emanating from the first copy of the CC-type variable node $\underline{c}^2\underline{c}_{1}^1(1)$ in the lifted version of $\mathscr{T}(\underline{c}^2\underline{v}^1)$. 
  From the construction of protograph LDPC codes, it becomes evident that among the paths emanating from the first copy of $\underline{c}^2\underline{c}_1^1$, there exist exactly $d_c^2$ paths that end at different copies of the CC-type variable node $\underline{c}^2\underline{c}_{j}^1$ for $j\in [d_v^1]\setminus 1$. 
  However, it remains unclear whether each of the $d_c^2$ paths, which should end at the copies of $\underline{c}^2\underline{c}_{j}^1$, ends at a single copy of the variable node $\underline{c}^2\underline{c}_{j}^1$ in the lifted version.
 In the subsequent analysis, we demonstrate the aforementioned statements, thus showing the existence of $\gamma$ isomorphic copies of subgraph $\mathscr{T}(\underline{c}^2\underline{v}^1)$ after the lifting process.

We start by traversing one step along the $d_c^2$ paths originating from CC-type variable node  $\underline{c}^2\underline{c}_1^1(1)$.
Given that $\underline{c}^2\underline{c}_1^1$ is connected to the check node $\underline{v}^2_i\underline{c}^1_j$ via an edge labeled $x^{\mu_i}$ within subgraph $\mathscr{T}(\underline{c}^2\underline{v}^1)$, Lemma~\ref{lemma:intermediate_result_copy_and_permute} implies that $\underline{c}^2\underline{c}_1^1(1)$ is connected to $\underline{v}^2_i\underline{c}^1_1(\mu_i+1)$, where $i \in [d_c^2]$.
Given that all the check nodes within subgraph $\mathscr{T}(\underline{c}^2\underline{v}^1)$ have degree two, the paths originating from $\underline{c}_j^2\underline{c}_1^1(1)$ do not branch into more paths at check node $\underline{v}^2_i\underline{c}^1_1(\mu_i+1)$ for any $i \in [d_c^2]$.
Observe that in subgraph $\mathscr{T}(\underline{c}^2\underline{v}^1)$ of the base graph, check node $\underline{v}^2_i\underline{c}^1_1$ is connected to $\underline{v}^2_i\underline{v}^1$ through an edge labeled $x^{\nu_1}$ for any $i \in [d_c^2]$.
Lemma~\ref{lemma:intermediate_result_copy_and_permute} implies that in the lifted version, check node $\underline{v}^2_i\underline{c}^1_1(\mu_i+1)$ is connected to variable node $\underline{v}^2_i\underline{v}^1(\tau_i)$, where $\tau_i= \mu_i+1-\nu_1-1 (\text{ mod } \gamma)+1$.
Thus, in the lifted version of subgraph $\mathscr{T}(\underline{c}^2\underline{v}^1)$, there exist $d_c^2$ paths of length two that originate from $\underline{c}^2\underline{c}^1_1(1)$ and end at $\underline{v}_i^2\underline{v}^1(\tau_i)$, for $i \in [d_c^2]$.
Each path among these $d_c^2$ paths further branches into $d_v^1-1$ paths since the VV-type variable nodes are of degree $d_v^1$ in subgraph $\mathscr{T}(\underline{c}^2\underline{v}^1)$.
It should be noted that a path of length two exists between the variable node $\underline{v}_i^2\underline{v}^1$ and check node $\underline{c}\underline{c}_j^1$, for every $i \in [d_c^2]$ and for any fixed $j \in [d_v^1]\setminus 1$, within subgraph $\mathscr{T}(\underline{c}^2\underline{v}^1)$ of the base graph.
Moreover, among these paths of length two, the edge connecting the variable node $\underline{v}_i^2\underline{v}$ to the check node $\underline{v}_i^2\underline{c}_j^1$ has label $\nu_j$, while the edge linking the check node $\underline{v}_i^2\underline{c}_j^1$ to the variable node $\underline{c}^2\underline{c}_j^1$ has label $\mu_i$.
Now consider the copies of these length two paths in the lifted version of subgraph $\mathscr{T}(\underline{c}^2\underline{v}^1)$ starting from variable node $\underline{v}_i^2\underline{v}(\tau_i)$.
From the application Lemma~\ref{lemma:intermediate_result_copy_and_permute} twice, it follows that
variable node  $\underline{v}_i^2\underline{v}(\tau_i)$ is connected to $\tau_i+\nu_j-\mu_i-1 (\text{ mod } \gamma )+1$-th copy of variable node $\underline{c}^2\underline{c}^1_j$, for any $i \in [d_c^2]$ and a fixed $j \in [d_v^1]\setminus 1 $.
Note that $\tau_i = \mu_i-\nu_1(\text{ mod } \gamma)+1$; therefore, $\underline{v}_i^2\underline{v}(\tau_i)$ is connected to $\nu_j-\nu_1-1 (\text{ mod } \gamma )+1$-th copy of $\underline{c}^2\underline{c}^1_j$ for any $i \in [d_v^2]$.
We have shown that there exist $d_c^2$ paths among those that start from a copy of $\underline{c}^1\underline{c}_1^1$ and end at the same copy of $\underline{c}^1\underline{c}_j^1$ for any $j \in [d_c^2]\setminus 1$ in the lifted version of subgraph $\mathscr{T}(\underline{c}^2\underline{v}^1)$, which implies that the lifted graph has $\gamma$ isomorphic copies of $\mathscr{T}(\underline{c}^2\underline{v}^1)$. 
The proof is shown graphically in Figure~\ref{fig:LP_TS_pictorial_proof} for a base graph which has degree-three VV-type variable nodes and degree-four CC-type variable nodes.

\section{Decoding Dynamics of TSs}
\label{appendix:decoding_dynamics}
The appendix illustrates the decoding dynamics of TSs induced by several linear combinations of stabilizer generators, in which VV-type and CC-type variable nodes have degrees three and four, respectively. The inner shape of the erroneous variable nodes and unsatisfied checks is depicted in black. The numbers above and below each variable node denote the number assigned to it according to the procedure described in Lemma~\ref{lemma:vn_only_graphical_description} and the number of unsatisfied checks connected to it, respectively.
Note that in each of these illustrations, the erroneous variable nodes in the first iteration are the same as those in the last iteration, concluding that each of these subgraphs constitutes a TS. 
\begin{figure*}
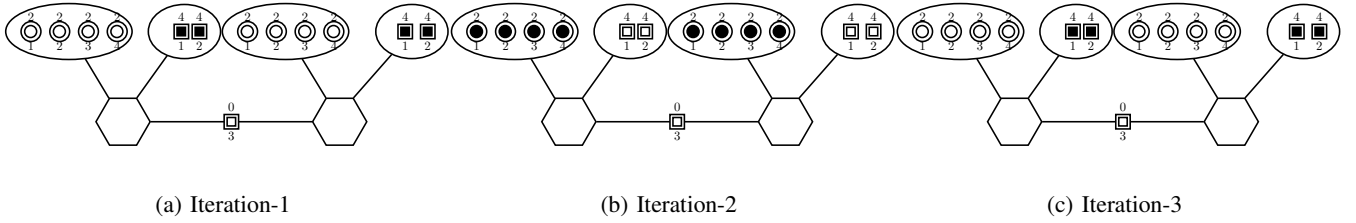

\begin{subfigure}{0.32\textwidth}
\input{Figures/Journal_figures/decoding_dynamic_TS1/iteration1}
    \caption{Iteration-1}
\end{subfigure}
\begin{subfigure}{0.32\textwidth}
\input{Figures/Journal_figures/decoding_dynamic_TS1/iteration2}
    \caption{Iteration-2}
\end{subfigure}
\begin{subfigure}{0.32\textwidth}
\input{Figures/Journal_figures/decoding_dynamic_TS1/iteration3}
    \caption{Iteration-3}
\end{subfigure}
    \caption{Trapping set~1}
    \label{fig:decoding_dynamics_of_TS1}
\end{figure*}
\begin{figure*}
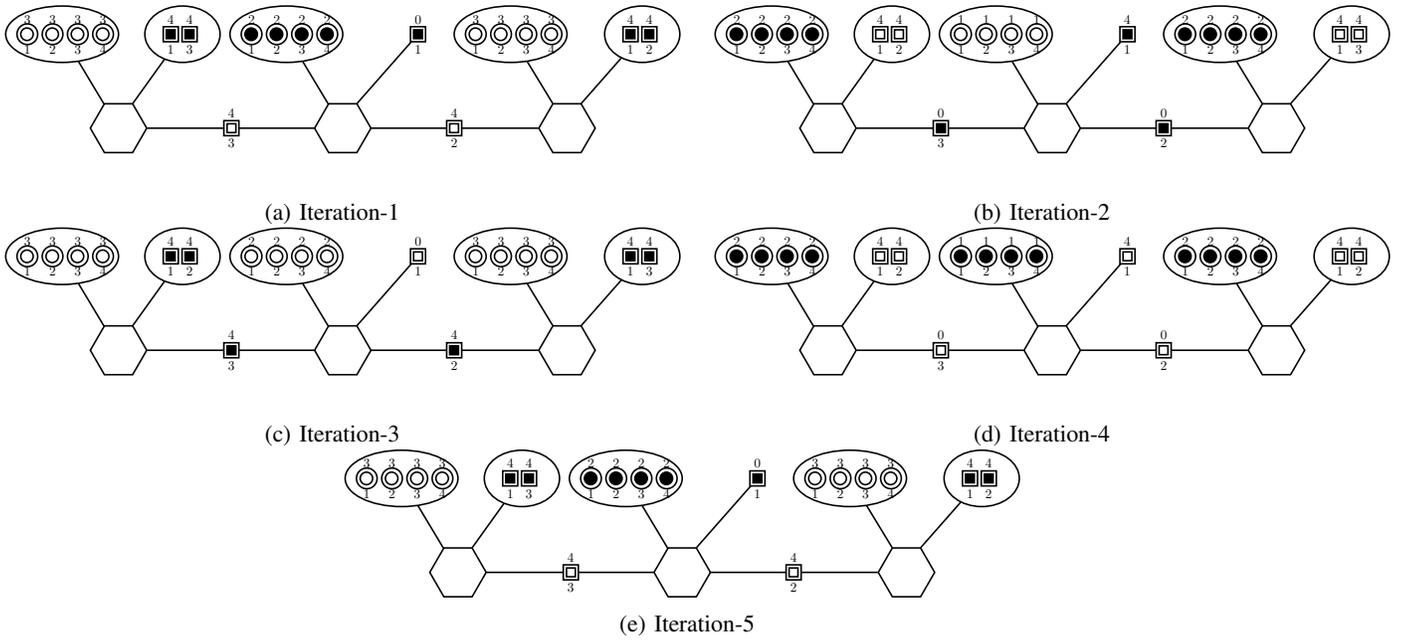

\begin{subfigure}{0.48\textwidth}
\centering
\input{Figures/Journal_figures/decoding_dynamics_TS4/iteration1}
    \caption{Iteration-1}
\end{subfigure}
\begin{subfigure}{0.48\textwidth}
\input{Figures/Journal_figures/decoding_dynamics_TS4/iteration2}
\centering
    \caption{Iteration-2}
\end{subfigure}
\begin{subfigure}{0.48\textwidth}
\centering
\input{Figures/Journal_figures/decoding_dynamics_TS4/iteration3}
    \caption{Iteration-3}
\end{subfigure}
\begin{subfigure}{0.48\textwidth}
\centering
\input{Figures/Journal_figures/decoding_dynamics_TS4/iteration4}
    \caption{Iteration-4}
\end{subfigure}
\begin{subfigure}{\textwidth}
\centering
\input{Figures/Journal_figures/decoding_dynamics_TS4/iteration5}
    \caption{Iteration-5}
\end{subfigure}
    \caption{Trapping set~2}
    \label{fig:decoding_dynamics_of_TS2}
\end{figure*}
\begin{figure*}
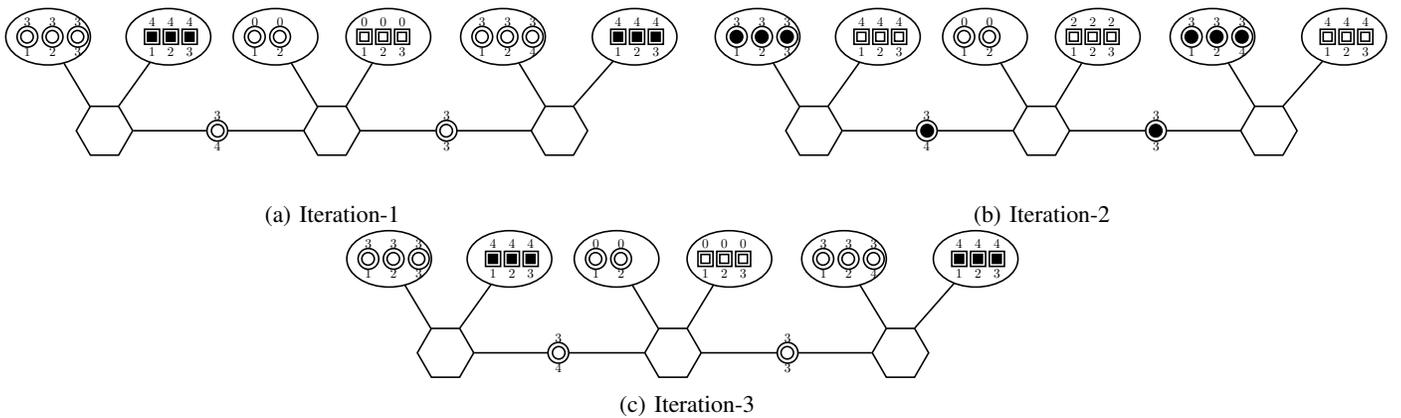

\begin{subfigure}{0.48\textwidth}
\centering
\input{Figures/Journal_figures/decoding_dynamics_TS5/iteration5}
    \caption{Iteration-1}
\end{subfigure}
\begin{subfigure}{0.48\textwidth}
\input{Figures/Journal_figures/decoding_dynamics_TS5/iteration6}
\centering
    \caption{Iteration-2}
\end{subfigure}
\begin{subfigure}{\textwidth}
\centering
\input{Figures/Journal_figures/decoding_dynamics_TS5/iteration7}
    \caption{Iteration-3}
\end{subfigure}
    \caption{Trapping set~3}
    \label{fig:decoding_dynamics_of_TS3}
\end{figure*}
\begin{figure*}
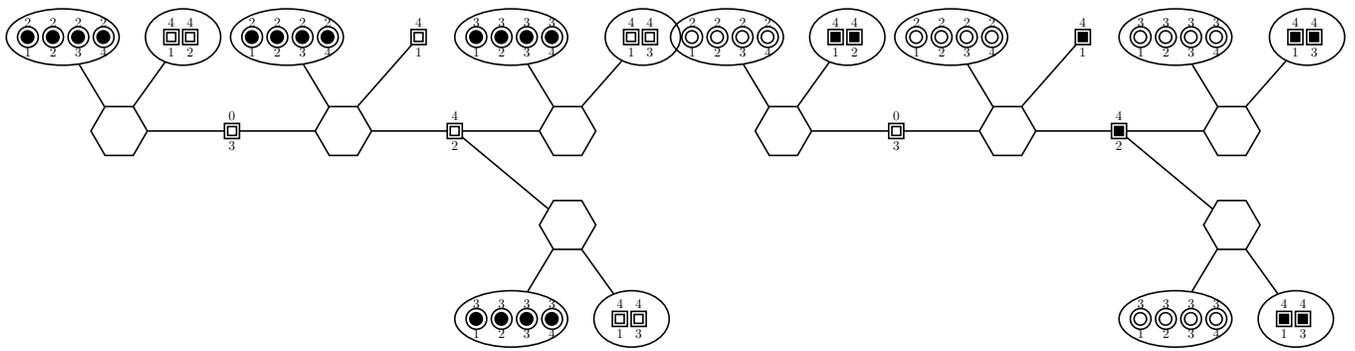
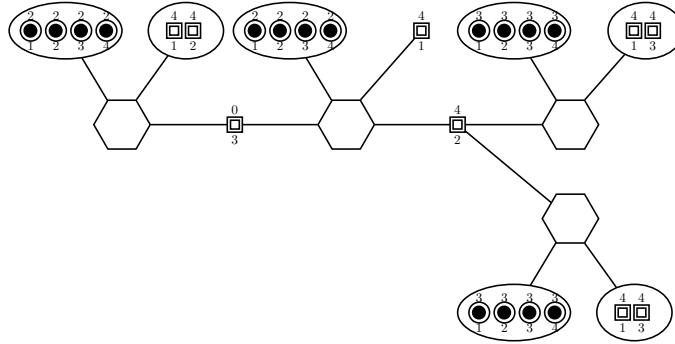

\begin{subfigure}{0.48\textwidth}
\centering
\input{Figures/Journal_figures/decoding_dynamics_TS2/iteration1}
    \caption{Iteration-1}
\end{subfigure}
\begin{subfigure}{0.48\textwidth}
\input{Figures/Journal_figures/decoding_dynamics_TS2/iteration2}
\centering
    \caption{Iteration-2}
\end{subfigure}
\centering
\begin{subfigure}{\textwidth}
\centering
\input{Figures/Journal_figures/decoding_dynamics_TS2/iteration3}
    \caption{Iteration-3}
\end{subfigure}

    \caption{Trapping set~4}
    \label{fig:decoding_dynamics_of_TS4}
\end{figure*}
\begin{figure*}
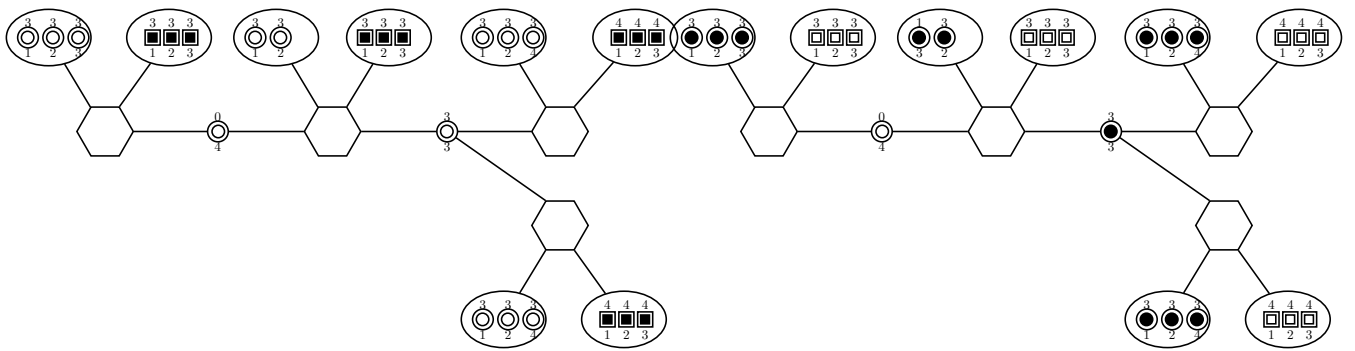
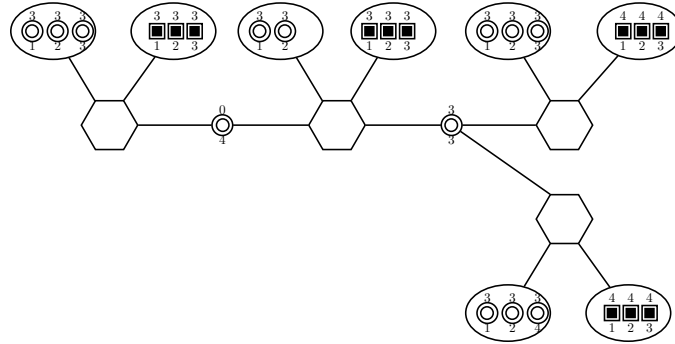

\begin{subfigure}{0.48\textwidth}
\centering
\input{Figures/Journal_figures/decoding_dynamics_TS3/iteration1}
    \caption{Iteration-1}
\end{subfigure}
\begin{subfigure}{0.48\textwidth}
\input{Figures/Journal_figures/decoding_dynamics_TS3/iteration2}
\centering
    \caption{Iteration-2}
\end{subfigure}
\centering
\begin{subfigure}{\textwidth}
\centering
\input{Figures/Journal_figures/decoding_dynamics_TS3/iteration3}
    \caption{Iteration-3}
\end{subfigure}

    \caption{Trapping set~5}
    \label{fig:decoding_dynamics_of_TS5}
\end{figure*}
\begin{figure*}
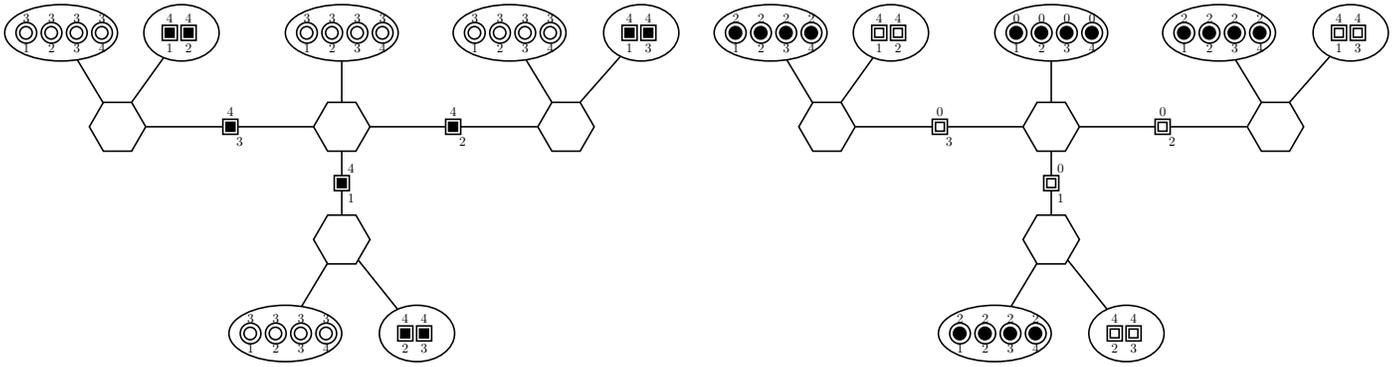
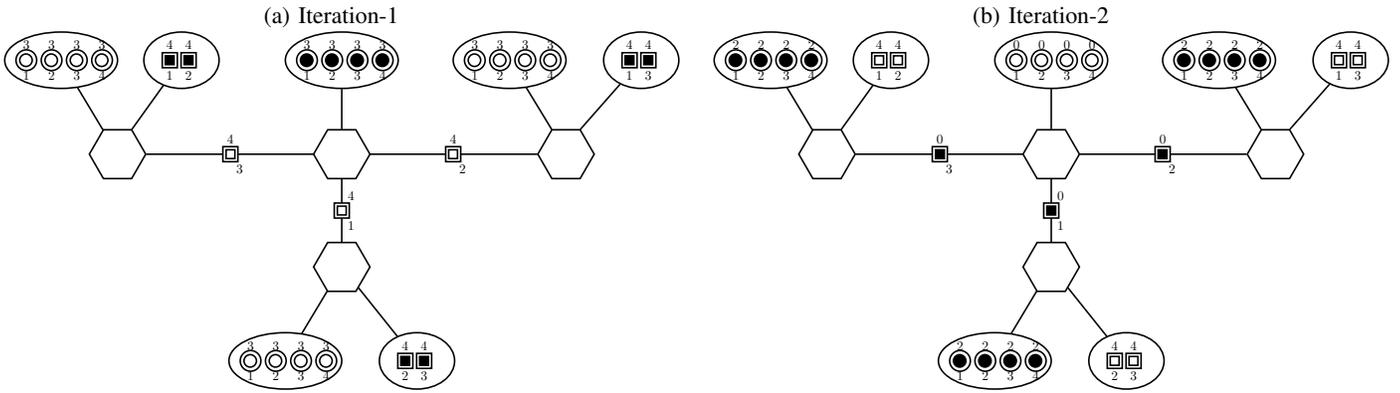
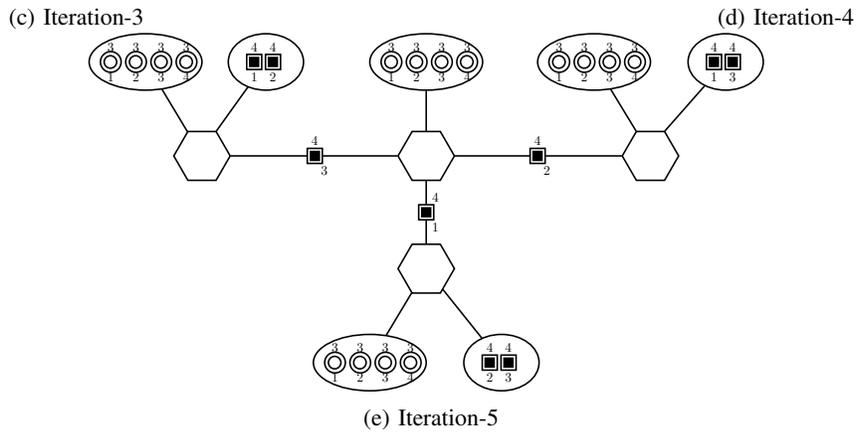

\begin{subfigure}{0.48\textwidth}
\centering
\input{Figures/Journal_figures/decoding_dynamics_TS6/iteration1}
    \caption{Iteration-1}
\end{subfigure}
\begin{subfigure}{0.48\textwidth}
\input{Figures/Journal_figures/decoding_dynamics_TS6/iteration2}
\centering
    \caption{Iteration-2}
\end{subfigure}
\begin{subfigure}{0.48\textwidth}
\centering
\input{Figures/Journal_figures/decoding_dynamics_TS6/iteration3}
    \caption{Iteration-3}
\end{subfigure}
\begin{subfigure}{0.48\textwidth}
\centering
\input{Figures/Journal_figures/decoding_dynamics_TS6/iteration4}
    \caption{Iteration-4}
\end{subfigure}
\begin{subfigure}{\textwidth}
\centering
\input{Figures/Journal_figures/decoding_dynamics_TS6/iteration5}
    \caption{Iteration-5}
\end{subfigure}
    \caption{Trapping set~6}
    \label{fig:decoding_dynamics_of_TS6}
\end{figure*}
\begin{figure*}
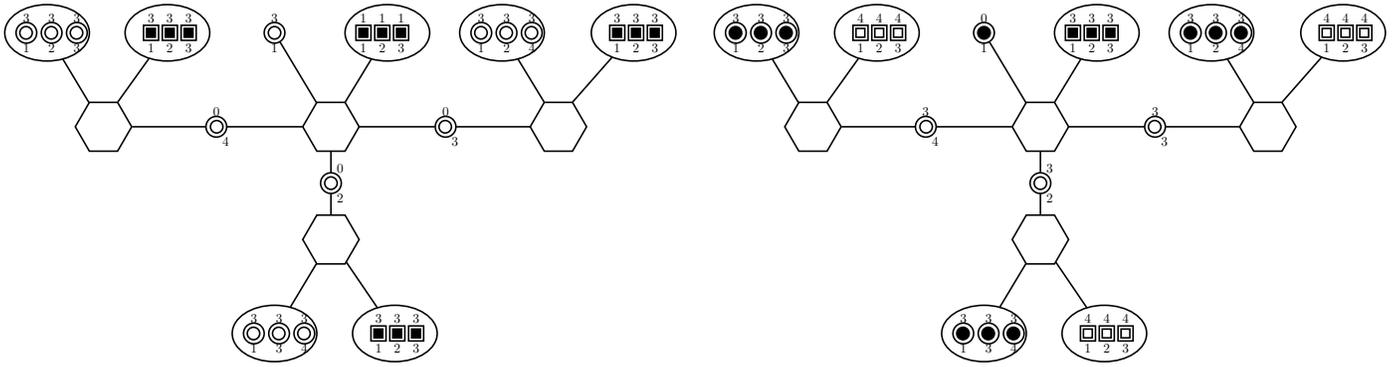
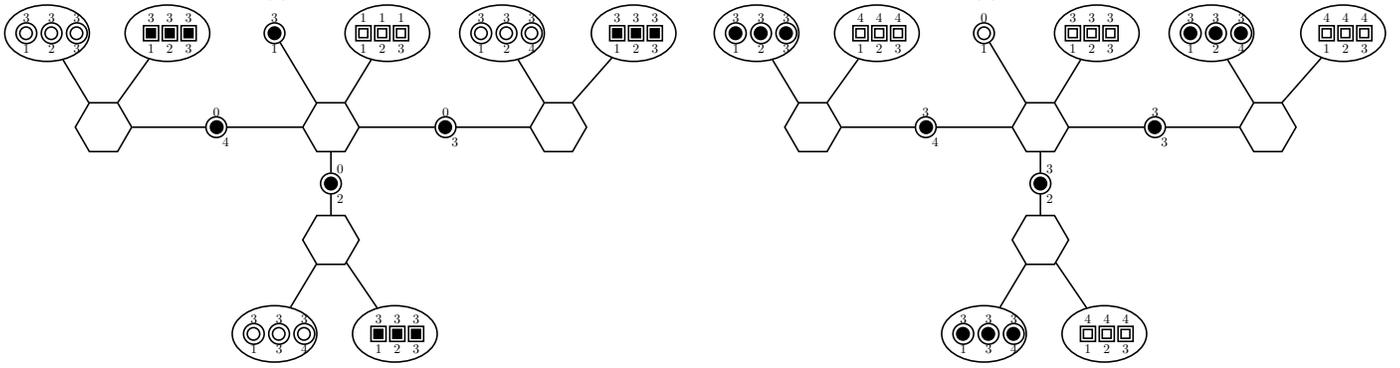
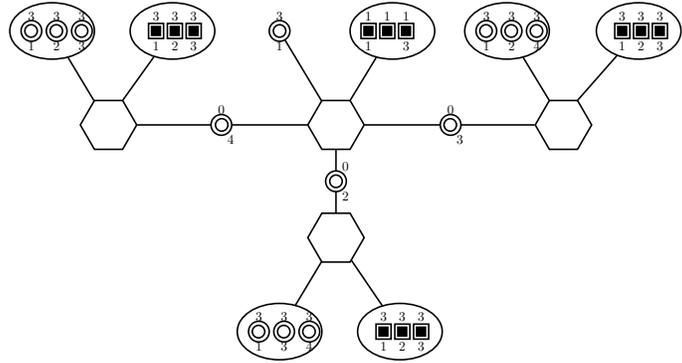

\begin{subfigure}{0.48\textwidth}
\centering
\input{Figures/Journal_figures/decoding_dynamics_TS7/iteration1}
    \caption{Iteration-1}
\end{subfigure}
\begin{subfigure}{0.48\textwidth}
\input{Figures/Journal_figures/decoding_dynamics_TS7/iteration2}
\centering
    \caption{Iteration-2}
\end{subfigure}
\begin{subfigure}{0.48\textwidth}
\centering
\input{Figures/Journal_figures/decoding_dynamics_TS7/iteration3}
    \caption{Iteration-3}
\end{subfigure}
\begin{subfigure}{0.48\textwidth}
\centering
\input{Figures/Journal_figures/decoding_dynamics_TS7/iteration4}
    \caption{Iteration-4}
\end{subfigure}
\begin{subfigure}{\textwidth}
\centering
\input{Figures/Journal_figures/decoding_dynamics_TS7/iteration5}
    \caption{Iteration-5}
\end{subfigure}
    \caption{Trapping set~7}
    \label{fig:decoding_dynamics_of_TS7}
\end{figure*}
\end{appendices}
\clearpage
\enlargethispage{-1.2cm} 
\bibliographystyle{IEEEtran}
\bibliography{ref_quantum_decoders.bib}

%
%
%
%
%
%
%

\end{document}